\documentclass[10pt]{elsarticle}
\usepackage{amssymb,amsmath}
\usepackage[small,nohug,heads=littlevee]{diagrams}
\diagramstyle[labelstyle=\scriptstyle]
\usepackage{array,multirow,graphicx}
\usepackage{enumitem}
\usepackage{feynmf}
\usepackage{lscape}
\usepackage{multirow}
\usepackage{bbm}
\usepackage{makecell}
\usepackage{tikz}
\usepackage{tikz-feynman}
\usepackage{graphicx}
\usepackage{tabularx}
\usepackage{forest}
\usepackage{tikz-qtree}
\usepackage[utf8]{inputenc}
\usepackage[T1]{fontenc}
\usepackage{hyperref}
\usepackage{dirtytalk}
\usepackage{mathtools}
\usepackage[margin=1in,papersize={8.5in,11in}]{geometry}
\usepackage[utf8]{inputenc}
\usepackage[english]{babel}
\usepackage{amsthm}
\usepackage{float}
\usepackage{amsfonts,mathrsfs}
\usepackage{bm}
\usepackage{mathrsfs}
\usepackage{makecell}

\usepackage{tikz}
\usetikzlibrary{matrix}
\newcommand{\tikzcircle}[2][red,fill=red]{\tikz[baseline=-0.5ex]\draw[#1,radius=#2] (0,0) circle ;}%

\usepackage{epstopdf}
\usepackage{subcaption}
            
\usepackage{tablefootnote}
\graphicspath{ {images/} }
\usepackage{comment}

\newtheorem{theorem}{Theorem}
\newtheorem{corollary}{Corollary}[theorem]
\newtheorem{lemma}[theorem]{Lemma}

\theoremstyle{definition}

\newcommand{\stir}{\genfrac{\{}{\}}{0pt}{}}
\def\Z{\mathbb{Z}}
\def\R{\mathbb{R}}

\def\Dd{\mathbb{D}_{\mathbb{A}}}
\def\Dda{\mathbb{D}_{\mathbb{A}'}}
\def\Ddb{\mathbb{D}_{\mathbb{B}}}

\def\H{\mathcal{H}}

\def\Oo{\mathbb{O}}

\def\Tr{\overrightarrow{\mathcal{T}}}
\def\W{\mathcal{W}}
\def\A{\mathcal{A}}
\def\Aa{\mathbb{A}}
\def\B{\mathcal{B}}
 
\def\S{\mathcal{S}} 
\def\L{\mathcal{L}} 
 
\def\M{\mathcal{M}}
\def\N{\mathcal{N}}

\def\P{\mathcal{P}}
\def\U{\mathcal{U}} 
\def\X{\mathcal{X}}

\def\d{\dagger}
\def\Q{\mathcal{Q}}

\def\G{\mathcal{G}}
\def\K{\mathcal{K}} 
\def\F{\mathcal{F}} 
\def\I{\mathcal{I}}
\def\U{\mathcal{U}}

\def\Pi{\mathbf{P}}

\def\tB{\tilde{B}}
\def\hB{\hat{B}}
\def\Bb{\mathbb{B}}
\def\tP{\tilde{P}}

\def\gdbig{\tikzcircle[green, fill]{2pt}}

\def\rdbig{\tikzcircle[red, fill]{2pt}}

\def\blkdbig{\tikzcircle[black, fill]{2pt}}

\def\bledbig{\tikzcircle[blue, fill]{2pt}}

\def\pdbig{\tikzcircle[pink, fill]{2pt}}

\def\ydbig{\tikzcircle[yellow, fill]{2pt}}
\def\lgraf{\overrightarrow{\otimes}}

\journal{arXiv}

\begin{document}
\begin{frontmatter}
\title{Combinatorial Mori-Zwanzig theory}
\author[ucm]{Yuanran Zhu\corref{correspondingAuthor}}
\cortext[correspondingAuthor]{Corresponding author}
\ead{yzhu4@lbl.gov}
\address[lbl]{Applied Mathematics \& Computational Research Division, Lawrence Berkeley National Laboratory,\\
Berkeley, CA 94720, USA}

\begin{abstract}
We introduce a combinatorial version Mori-Zwanzig theory and develop from it a family of self-consistent evolution equations for the correlation function or Green's function of interactive many-body systems. The core idea is to use an ansatz to rewrite the memory kernel (self-energy) of the regular Mori-Zwanzig equation as a function composition of the correlation (Green's) function. Then a series of algebraic combinatorial tools, especially the commutative and noncommutative Bell polynomials, are used to determine the exact Taylor series expansion of the composition function. The resulting combinatorial Mori-Zwanzig equation (CMZE) yields novel non-perturbative expansions of the equation of motion for the correlation (Green's) function. The structural equation for deriving such a combinatorial expansion resembles the combinatorial Dyson-Schwinger equation and may be viewed as its temporal-domain analogue. After introducing the abstract word and tree representation of the CMZE, we show its wide-range application in classical, stochastic, and quantum many-body systems. In all these examples, the new self-consistent expansions we obtained with the CMZE are similar to the diagrammatic skeleton expansions used in quantum many-body theory and lattice statistical field theory. We expect such a new framework can be used to calculate the correlation (Green's) function for strongly correlated/interactive many-body systems.
\end{abstract}
\end{frontmatter}
\section{Introduction}
\subsection{Motivation}
For interactive many-body systems, the correlation/Green's function captures meaningful information from which we can obtain important physical quantities such as the spectral function, density of states, relaxation times, and response functions. As one of the main themes in modern many-body theories, traditionally, the calculation of the correlation/Green's function heavily relies on the many-body perturbation theory (MBPT). This approach was systematically developed by Feynman, Tomonaga, Schwinger, and Dyson in quantum electrodynamics \cite{zinn2021quantum}, and since then has been widely applied to condensed matter physics \cite{mahan2013many} and statistical field theory \cite{parisi1988statistical}. For weakly interactive many-body systems, MBPT usually yields a rather accurate prediction of Green's function due to the asymptotic characteristic of the perturbation series \cite{altland2010condensed}. However, when strong interaction between particles is imposed, which is the case for say strongly correlated electron systems, the classical MBPT fails for the obvious reason. Hence, the calculation of Green's function for such strongly interactive/correlated systems becomes one of the most challenging tasks in condensed matter physics and quantum field theory.   

For simplicity, from now on our discussion will focus on the calculation of Green's function for discrete/lattice systems since the continuous field requires renormalization, which is an independent research topic \cite{zinn2021quantum}. Although far from being satisfactory, many established methods have obtained success in this regard. For strongly interactive/correlated systems, a large class of the developed methods employs a {\em modified} MBPT and {\em renormalized}
\footnote {We particularly note the difference between the terminology renormalization and renormalized perturbation theory. Renormalization in quantum field theory \cite{zinn2021quantum} or stochastic partial differential equations \cite{bruned2019algebraic} refers to the formal procedure to eliminate the infinity appearing in the perturbation series. The renormalized expansion theory in our context means the resummation technique as mentioned below.} 
perturbation series to overcome, at least numerically, the divergence difficulty encountered when applying the original perturbation series. This type of method is often termed ``skeleton expansion'', ``dressed expansion'',``renormalized expansion''
, or ``bold diagrammatic expansion'' in the literature \cite{stefanucci2013nonequilibrium,lin2021bold,lin2021bold2}. Typical numerical schemes that can be classified in the category include the self-consistent Hartree-Fock approximation, Born's approximation, and the GW approximation used in the approximated Kadanoff-Baym equation \cite{stefanucci2013nonequilibrium}, as well as the dynamical mean-field theory \cite{georges1996dynamical}. The core procedure adopted in the renormalized expansion is to replace the bare, non-interacting Green's function with the interactive Green's function via the resummation of the Feynman diagram. In the end, one is led to solve a self-consistent Dyson's equation where the self-energy becomes a function of the interactive Green's function.     

From a computational point of view, the calculation of Green's function is a dimension-reduction problem since one aims to get a closed equation for the low-dimensional one-particle, or two-particle Green's function while the exact evolution equation for Green's function depends on all $n$-particle Green's functions, as clearly seen from the Martin-Schwinger-Dyson hierarchy \cite{stefanucci2013nonequilibrium,zinn2021quantum}. In a seemly different context, this type of problem can be well-formulated using the Mori-Zwanzig (MZ) theory, which is a rather general framework first developed in nonequilibrium statistical mechanics \cite{mori1965transport,zwanzig1960ensemble,zwanzig2001nonequilibrium} for explaining the appearance of irreversibility and memory effect in the microscopic particle system, and since then gradually became one of the standard paradigms to formulate dimension-reduction problems appearing in different fields. In fact, many researchers have successfully applied it to different interactive many-body systems to get the evolution equation for the correlation/Green's function \cite{snook2006langevin,fulde1995electron,kakehashi2009full}. In the MZ framework, instead of using the Dyson series expansion in the interaction picture to get Dyson's equation, one stays in the Heisenberg picture and introduces a formal projection operator $\P$ to isolate the low-dimensional quantity of interest, which in our case is the correlation/Green's function, and uses the differential-form Dyson's identity to get the evolution equation for it. Once the equation is built, one can introduce a series expansion to approximate the MZ memory function, a term analogous to the self-energy in Dyson's equation. Eventually, we obtain a closed evolution equation for Green's function that is similar to Dyson's equation. When compared with the MBPT and the renormalized MBPT, it is fair to say that the MZ approach is relatively less popular in the physics community, especially for those who work on {\em quantum} many-body theory, albeit some experts such as P. Fulde pointed out that the MZ framework has advantages to address the strongly correlated system since it is {\em free}-Wick's theorem \cite{fulde1995electron}. In our humble opinion, this phenomenon is due to two drawbacks of the traditional MZ theory. The first one is the computational difficulty of the series expansion for the memory function. The second one is more fundamental. Unlike the renormalized MBPT, one lacks systematic renormalization methods to make the MZ equation a self-consistent equation for the correlation/Green's function
\footnote{Some previous works \cite{fulde1995electron,kakehashi2009full,fulde2006strongly} already introduced different self-consistent approximations to the Mori-Zwanzig equation. The methodology adopted in this paper is presented in a self-contained manner and is generally different from what has been used in the above references.}. Due to this limitation, many approximation schemes developed under the MZ framework are either system-dependent and hence hard to generalize, or only valid locally, therefore yielding bad asymptotics. In this paper, we try to address these two problems, in particular, the second one. This is the main motivation for the development of a combinatorial version of the Mori-Zwanzig theory. Now we briefly sketch our main results. 

\subsection{Main results}
\label{sec:into_main_result}
Although we will show in Section \ref{sec:app} that in principle, the combinatorial Mori-Zwanzig theory applies to any statistical, stochastic, and quantum mechanical lattice fields, it is convenient to choose a quantum lattice field, say the Hubbard model, as an example to show the main result of the paper and make comparisons with the renormalized Dyson's equation, which is the counterpart we aim to mimic in the Mori-Zwanzig framework. Recall that the fundamental equation of motion (EOM) for interactive electron systems is Dyson's equation:
\begin{gather}
    G=G_0+G_0\Sigma G, \label{intro_Dyson_eqn}\\
    G=G_0+G_0\Sigma G_0+G_0\Sigma G_0\Sigma G_0+\cdots. \label{intro_Dyson_eqn_series}
\end{gather}
Eqn \eqref{intro_Dyson_eqn_series} is the series expansion representation for Eqn \eqref{intro_Dyson_eqn}. In \eqref{intro_Dyson_eqn}-\eqref{intro_Dyson_eqn_series}, $G=G(p,i\omega_n)$ is the (exact) interactive Green's function in the momentum-energy space, $G_0=G_0(p,i\omega_n)$ is known as the bare, or non-interactive Green's function which can be calculated analytically from the non-interactive part of the modeling Hamiltonian. $\Sigma=\Sigma(p,i\omega_n)$ is the self-energy, which can be written as the series expansion $\Sigma=\sum_{n=1}^{\infty}\Sigma^{(n)}[G_0,v]$, where the $n$-th order contribution $\Sigma^{(n)}[G_0,v]$ is a function of $G_0$ and the interaction strength $v$. In MBPT, $\Sigma^{(n)}[G_0,v]$ can be obtained using Feynman diagrams. The main objective of the renormalized MBPT is to find a way to represent the self-energy $\Sigma$ as the series expansion with respect to the function of the interactive Green's function $G$, instead of $G_0$. For instance, if using $G$-skeleton expansion, we have:
\begin{align}\label{intro_self-energy}
   \Sigma=\Sigma_r=\sum_{n=1}^{\infty}\Sigma^{(n)}_r[G,v],
\end{align}
where $\Sigma_r$ is the called the {\em renormalized self-energy}, with $\Sigma^{(n)}_r[G,v]$ corresponding to the skeleton diagrams after the dressing replacement $G_0\rightarrow G$ \cite{stefanucci2013nonequilibrium}. Substituting \eqref{intro_self-energy} into Dyson's equation \eqref{intro_Dyson_eqn}, we get the following self-consistent Dyson's equation for $G$:
\begin{gather}\label{intro_renorm_Dyson_eqn}
    G=G_0+\sum_{n=1}^{\infty}G_0\Sigma^{(n)}_r[G,v]G.
\end{gather}
Eqn  \eqref{intro_renorm_Dyson_eqn} serves as the starting point of many renormalized MBPT. In Mori-Zwanzig theory, the following operator equation, known as the differential-form Dyson's identity, is used to calculate the correlation/Green's function: 
\begin{align}\label{intro_op_ROM}
\frac{d}{dt}\P e^{t\L}\P =\P e^{t\L}\P\L\P 
+\int_0^t\P e^{(t-s)\L}\P\L e^{s\Q\L}\Q\L\P ds.
\end{align}
Here $\U(t,0)=e^{t\L}$ is the time propagator of the system under investigation, $\P$ is a projection operator and $\Q=\I-\P$ is its orthogonal. $\U_{\Q}(t,0)=e^{t\Q\L}$ is the time propagator for the orthogonal dynamics. By choosing a suitable projection operator $\P$ and acting this operator equation in the range $\P$, we obtain the EOM for the correlation/Green's function $G(t)$, which can be roughly written as:
\begin{align}\label{intro_G_EOM}
\frac{d}{dt}G(t) =\Omega G(t) 
+\int_0^t \hat\Sigma(s)G(t-s)ds,
\end{align}
where $G(t)$ can be a scalar, vector, or matrix, depending on the definition of $\P$, The Laplace transform of Eqn \eqref{intro_G_EOM} yields the following equations that are similar to Dyson's equation and its series expansion:
\begin{gather}
G(z)= S^{-1}(z)G(0) + S^{-1}(z)\hat\Sigma(z)G(z),\\
G(z)= S^{-1}(z)G(0) +S^{-1}(z)\hat\Sigma(z)S^{-1}(z)G(0)
+S^{-1}(z)\hat\Sigma(z)S^{-1}(z)\hat\Sigma(z)S^{-1}(z)G(0)+\cdots.
\end{gather}
Here $S^{-1}(z)=(zI-\Omega)^{-1}$, assuming invertibility. One can view $\hat\Sigma(z)$ as the ``self-energy'' in the MZ framework. Via rigorous combinatorial derivation, we prove that operator EOM \eqref{intro_op_ROM} admits series expansion:
\begin{align}
    \frac{d}{dt}\P e^{t\L}\P =\P e^{t\L}\P\L\P+\sum_{n=0}^{\infty}\sum_{k=0}^n\frac{(-1)^{n-k}}{(n-k)!k!}\int_0^t\P e^{(t-s)\L}\F_n\left(\P e^{s\L}\P\right)^kds,
\end{align}
where $\{\F_n\}_{n=0}^{\infty}$ is a sequence of operators which can be obtained explicitly by solving a recurrence equation. As a result, we obtain the following self-consistent evolution equation for the correlation/Green's function:
\begin{align}\label{intro_CMZE}
\frac{d}{dt}G(t)&=\Omega G(t) 
+\sum_{n=0}^{\infty}\sum_{k=0}^n\frac{(-1)^{n-k}}{(n-k)!k!}\int_0^t [G(s)\Omega_n]^kG(t-s)ds.
\end{align}
EOM of the type \eqref{intro_CMZE} will be called the {\em combinatorial Mori-Zwanzig equation (CMZE)}. It also admits a frequency space representation:
\begin{gather}
G(z)= S^{-1}(z)G(0) + S^{-1}(z)\hat\Sigma(z)G(z),\label{intro_G(z)}\\
G(z)= S^{-1}(z)G(0) + \sum_{n=0}^{\infty}\sum_{k=0}^n\frac{(-1)^{n-k}}{(n-k)!k!}S^{-1}(z)\hat\Sigma_k[G(z),\Omega_n]G(z).\label{intro_G(z)_CMZE}
\end{gather}
where $\hat\Sigma_k[G(z),\Omega_n]=\mathfrak{L}\{[G(t)\Omega_n]^k\}$, $\mathfrak{L}$ is the Laplace transform. All $\Omega_n$s are explicitly computable and can be formally written as the function of $G(0)$, i.e. $\Omega_n=\Omega_n[G(0)]$. One can immediately find the similarity between the renormalized Dyson's equation \eqref{intro_renorm_Dyson_eqn} and CMZE \eqref{intro_G(z)_CMZE}.
The combinatorial Mori-Zwanzig theory can be readily generalized to nonequilibrium systems with the evolution operator $\U(t,0)=\Tr e^{\int_0^t\L(\tau)d\tau}$ generated by time-dependent operator $\L(\tau)$. We refer to Section \ref{sec:Time-d-CMZE} for details.
\subsection{Outline}
This paper is organized as follows. In Section \ref{sec:CCMZE}, we first briefly review the formal derivation of the Mori-Zwanzig equation for time-independent and time-dependent many-body systems. We also derive the commutative CMZE using a well-known recurrence relation in Mori-Zwanzig theory and Fa\`a di Bruno's formula. With this simple example, we wish to provide a quick introduction to the core idea behind the combinatorial expansion theory. Section \ref{sec:NCC-MZE_word} contains the main result of this paper. In order to extend the commutative CMZE to the noncommutative case, which is required for the development of a general theory based on operator algebras, in Section \ref{sec:bell_poly} we introduce several new noncommutative polynomials of words and in \ref{sec:alge_eqn_words} a general symbolic algebraic equation between these polynomials. In Section \ref{sec:main_thm} and \ref{sec:Time-d-CMZE}, we show that the noncommutative CMZE can be obtained using algebraic homomorphisms to map the solution to the symbolic algebraic equation into operators. Up to this section, the main construction of the combinatorial Mori-Zwanzig theory is finished. 

Section \ref{sec:NCC-MZE_tree} introduces an equivalent, tree representation for the CMZE. The main purpose of this section is to develop a tree diagrammatic method for deriving CMZE, which is comparable with the Feynman diagrammatic used in the renormalized skeleton expansion. Moreover, we show that the symbolic algebraic equation that leads to the CMZE can be reformulated into a functional equation for the generating function of trees. This builds a connection between CMZE and the combinatorial Dyson-Schwinger equation. The latter was heavily discussed in combinatorial quantum field theory \cite{yeats2008growth,yeats2017combinatorial,kreimer2006etude}. Section \ref{sec:app} contains applications of the combinatorial Mori-Zwanzig theory to classical Hamiltonian systems, stochastic dynamical systems, and quantum many-body systems. With these three examples, we demonstrate that CMZE provides a coherent framework to derive self-consistent EOM for the correlation/Green's function for virtually all kinds of lattice field systems. The main findings of the paper are summarized in Section \ref{sec:summary}. At last, we note that although the derivation of the CMZE heavily relies on algebraic combinatorics, the usage of algebraic and combinatorial terminology would be kept to a minimum quantity that is just enough for deriving the equation.  

\section{Commutative combinatorial Mori-Zwanzig equation}
 
\subsection{Mori-Zwanzig equation in a nutshell}
Consider a general evolution operator $\U(t,0)$ for a dynamical system evolving on a Banach space $X$. The system can be deterministic, stochastic, or quantum mechanical, depending on the choice of $\U(t,0)$ and $X$. We are interested in deriving the exact evolution equation for a phase space observable function $u(t)=u(x(t))$ (for quantum mechanical systems $u(t)$ is often an observable operator), where $x\in X$ is the coordinate observable in the Banach space $X$. To this end, we introduce a projection operator $\P:X\rightarrow X$, and its complementary projection $\Q=\I-\P:X\rightarrow X$, where $\I$ is the identity operator in $X$. If furthering assuming the dynamical system is time-independent and $\U(t,0)$ is generated by a time-independent infinitesimal generator $\L$, then by differential the well-known Dyson's identity:
\begin{align}\label{eqn:Dyson_ID}
    \U(t,0)=e^{t\L}=e^{t\Q\L}+\int_0^te^{s\L}\P\L e^{(t-s)\Q\L}ds
\end{align}
with respect to $t$, we can get the following operator identity for the evolution operator $\U(t,0)=e^{t\L}$:
\begin{align}\label{eqn:time_ind_MZE_OP}
\frac{d}{dt}e^{t\L}=e^{t\L}\P\L
+\int_0^t e^{s\L}\P\L e^{(t-s)\Q\L}\Q\L ds+e^{t\Q\L}\Q\L.
\end{align} 
Applying the above operator identity to the initial of the observable function $u(0)=u(x(0))$ yields the (full) Mori-Zwanzig equation (MZE):
\begin{align}\label{eqn:time_ind_MZE_u_t}
\frac{d}{dt}e^{t\L}u(0)=e^{t\L}\P\L u(0)
+\int_0^t e^{s\L}\P\L e^{(t-s)\Q\L}\Q\L u_0ds+e^{t\Q\L}\Q\L u(0).
\end{align} 
The three terms on the right-hand side (RHS) of Eqn \eqref{eqn:time_ind_MZE_u_t} are called, respectively, streaming term, memory term, and fluctuation (or noise) term. In applications, a more useful MZE is the {\em projected} MZE. Specifically, we apply the projection operator $\P$ from the left into \eqref{eqn:time_ind_MZE_u_t}. The orthogonality between projection operators implies $\P\Q=0$, which leads to $\P e^{t\Q\L}\Q\L u(0)	\equiv 0$. As a result, the projected MZE will evolve in a projected space $\text{Ran}(\P)$, i.e. the range of the projection operator $\P$ and takes the form:
\begin{align}\label{eqn:time_ind_PMZE_OP}
\frac{d}{dt}\P e^{t\L}u(0)=\P e^{t\L}\P\L u(0)
+\int_0^t\P e^{s\L}\P\L e^{(t-s)\Q\L}\Q\L u(0)ds.
\end{align} 
If further applying the projection operator $\P$ from right and then acting on $u(0)$, we get 
\begin{align}\label{eqn:time_ind_PMZE_POP}
\frac{d}{dt}\P e^{t\L}\P u(0)=\P e^{t\L}\P\L\P u(0)
+\int_0^t\P e^{s\L}\P\L e^{(t-s)\Q\L}\Q\L\P u(0)ds.
\end{align}
This projected MZE induces an endomorphism $\text{Ran}(\P)\rightarrow \text{Ran}(\P)$, which gives the exact evolution of the projected quantity $\P \U(t,0)\P$ in a presumably low-dimensional submanifold $\text{Ran}(\P)\subset X$. This is the essence of why the Mori-Zwanzig framework can be used as a universal dimension-reduction tool in mathematics and physics. To be noticed that, up to this point, everything is derived in a formal way and no approximations have been made. Therefore, the MZE \eqref{eqn:time_ind_MZE_u_t} and its projected form \eqref{eqn:time_ind_PMZE_OP}-\eqref{eqn:time_ind_PMZE_POP} are both formally exact. For time-independent systems, we further note that using the change of variable $t-s\rightarrow s$, one can express Dyson's identity \eqref{eqn:Dyson_ID} equivalently as
\begin{align}\label{memory_reformulation}
    e^{t\L}=e^{t\Q\L}+\int_0^te^{(t-s)\L}\P\L e^{s\Q\L}ds
\end{align}
and the resulting MZEs can be obtained following the same procedure. In fact, MZEs expressed in this format will be the starting point of our derivation for the time-independent combinatorial MZE since it leads to equations of motions with simpler forms. 

For time-dependent dynamical systems, the generator of the evolution operator $\U(t,0)$, denoted as $\L(t)$, is naturally time-dependent. The Mori-Zwanzig equation for any observable function $u=u(x(t))$ for time-dependent systems can be obtained using the time-dependent Dyson's identity:
\begin{align}\label{Dyson_ID_time-depen}
    \U(t,0)=\U_{\Q}(t,0)+\int_0^t\P\U(t,0)\P\L(t)\U_{\Q}(t,s)\Q\L(t)ds.
\end{align}
If using the time-ordered operator $\Tr$ which places operators from the right to the left with respect to the decreasing time order, then the Dyson series representation enables us to formally write the evolution operators in \eqref{Dyson_ID_time-depen} as time-ordered exponential maps $\U(t,s)=\Tr e^{\int_s^t\L(\tau)d\tau}$ and $\U_{\Q}(t,s)=\Tr e^{\int_s^t\Q\L(\tau)d\tau}$. Again, taking the time derivative and applying the resulting operator identity into $u(0)=u(x(0))$, we obtain the Mori-Zwanzig equation for the time-dependent system: 
\begin{align}\label{eqn:time_d_MZE_OP}
\frac{d}{dt}\U(t,0)u(0)=\U(t,0)\P\L(t)u(0)
+\U_{\Q}(t,0)\Q\L(t)u(0)
+\int_0^t\U(s,0)\P\L(s)\U_{\Q}(t,s)\Q\L(t)u(0)ds.
\end{align}
Similarly, we can get its projected version:
\begin{align}
\frac{d}{dt}\P\U(t,0)u(0)&=\P\U(t,0)\P\L(t)u(0)
+\int_0^t\P\U(s,0)\P\L(s)\U_{\Q}(t,s)\Q\L(t)u(0)ds
\label{eqn:time_d_MZE_POP}\\
\frac{d}{dt}\P\U(t,0)\P u(0)&=\P\U(t,0)\P\L(t)\P u(0)
+\int_0^t\P\U(s,0)\P\L(s)\U_{\Q}(t,s)\Q\L(t)\P u(0)ds.
\label{eqn:time_d_MZE_POPP}
\end{align}
A notable difference between the time-independent MZEs and the time-dependent ones is that reformulation \eqref{memory_reformulation} is not valid for time-dependent Dyson's identity \eqref{Dyson_ID_time-depen} and hence for all the follow-up time-dependent MZEs \eqref{eqn:time_d_MZE_OP}-\eqref{eqn:time_d_MZE_POPP}.  
\subsection{Commutative combinatorial Mori-Zwanzig equation}
\label{sec:CCMZE}
In this section, we use a simple example to derive the combinatorial expansion of the Mori-Zwanzig equation so that readers can quickly grasp the main idea of the combinatorial Mori-Zwanzig theory. In fact, for the simplest case, the combinatorial expansion we are going to derive is just an application of the classical Fa\`a di Bruno formula which generalizes the chain rule to higher-order derivatives. 

Consider a time-independent, classical Hamiltonian system with the modeling Hamiltonian $\H$. The associated state space evolution operator of the system is given by $\U(t,0)=e^{t\L}$, where $\L=\{\cdot,\H\}$ is known as the Liouville operator. Here $\{\cdot,\cdot\}$ denotes the classical Poisson bracket. Now we introduce a projection operator $\P$ and its orthogonal $\Q=\I-\P$. To be noticed that as projection operators orthogonal to each other, $\P,\Q$ satisfy the idempotence $\P^2=\P$, $\Q^2=\Q$ and $\P\Q=0$. 
With all these definitions, the following operator identity can be proved easily using the aforementioned operator properties and mathematical induction \cite{zhu2021effective}:
\begin{align}\label{recurrence_relation_op}
    \P\L(\Q\L)^n\P=\P\L^{(n+1)}\P-\sum_{k=1}^n\P\L(\Q\L)^{(k-1)}\P\L^{n-k+1}\P,\qquad n\geq 0.
\end{align}
One should notice that \eqref{recurrence_relation_op} is a {\em recurrence} relation since the second term on the right-hand side (RHS) of \eqref{recurrence_relation_op} contains operators $\{\P\L(\Q\L)\P\}_{i=1}^{n-1}$. Further decomposing these operators using this recurrence relation, eventually, we can rewrite operator $\P\L(\Q\L)^n\P$ as an (n+1)-order multivariate polynomial of operators
$\P\L^{j}\P$ for $1\leq j\leq n+1$, i.e. $\P\L(\Q\L)^n\P=P_{n+1}(\P\L\P,\P\L^2\P,\cdots,\P\L^{(n+1)}\P)$. Now, consider the simplest case when $X=H$ is a certain Hilbert space equipped with the inner product $\langle\cdot,\cdot\rangle$. $\P:H\rightarrow H$ is a finite-rank projection operator which admits a canonical form $\P(\cdot)=\langle u(0),\cdot\rangle u(0)$ with $u(0)\in H$ and $\langle u(0),u(0)\rangle=\| u(0)\|^2=1$. Then applying the above operator identity to $u(0)$ leads to a simple recurrence formula:

\begin{align}\label{recurrence_relation}
\begin{dcases}
\mu_0&=\gamma_0,\\
\mu_{n}&=\gamma_{n}-\sum_{k=1}^n\mu_{k-1}\gamma_{n-k},\qquad n\geq 1,
\end{dcases}
\end{align}
where we introduced shorthand notation $\mu_n,\gamma_n$ which are constants defined by: $\mu_n u(0)=\P\L(\Q\L)^n\P u(0)=\langle\L(\Q\L)^nu(0),u(0)\rangle u(0)$ for $n\geq 0$ and $\gamma_n u(0)=\P\L^{(n+1)}\P u(0)=\langle\L^{(n+1)}u(0),u(0)\rangle u(0)$ for $n\geq -1$. Note that $\gamma_{-1}=\langle u(0),u(0)\rangle=1$. As a result of the recurrence relation, any $\mu_n$ can be expressed as a $n$-dimensional polynomial of $\gamma_j$, $0\leq j\leq n-1$. For instance, $\mu_3=P_3(\gamma_0,\gamma_1,\gamma_2)=\gamma_2-2\gamma_1\gamma_0+\gamma_0^3$. For the Mori-Zwanzig equation of this special case, we have the following series expansions:
\begin{equation}\label{taylor_C(t)_K(t)}
\begin{aligned}
\P e^{t\L}\P u(0)&=\sum_{m=0}^\infty\frac{t^m}{m!}\P\L^m\P u(0)=u(0)\sum_{m=0}^\infty\frac{t^m}{m!}\gamma_{m-1}=u(0)C(t),\\
\P\L e^{t\Q\L}\Q\L\P  u(0)&=\sum_{n=0}^\infty\frac{t^n}{n!}\P\L(\Q\L)^n\Q\L \P u(0)=u(0)\sum_{n=0}^\infty\frac{t^n}{n!}\mu_{n+1}=u(0)K(t).
\end{aligned}
\end{equation}
Omitting the Hilbert space vector $u(0)$ in the above equations, these two formulas are just the Taylor series expansion for functions $C(t)$ and $K(t)$. Of course, we need to assume that they are smooth functions of time and therefore differentiable up to any order. For those who are familiar with the Mori-Zwanzig framework and its application in classical statistical mechanics, one immediately identifies that $C(t)=\langle u(t),u(0)\rangle$ corresponds to the time autocorrelation function (normalized since $\langle u(0),u(0)\rangle^2=1$) of the time-independent observable function $u$ and $K(t)=\langle \L e^{t\Q\L}u(0),u(0)\rangle=-\langle  e^{t\Q\L}u(0),\Q\L u(0)\rangle$ is the MZ memory kernel. Here we have used the common setting in classical equilibrium statistical mechanics that $\langle a, b\rangle=\frac{1}{Z}\int ab e^{-\beta\H}dpdq$. To get equality $\langle \L e^{t\Q\L}u(0),u(0)\rangle=-\langle  e^{t\Q\L}u(0),\Q\L u(0)\rangle$, we used operator properties that $\L$ is skew-Hermitian with respect to the inner product $\langle\cdot,\cdot\rangle$, and $\Q^2=\Q$ is symmetric with respect to $\langle\cdot,\cdot\rangle$. The so-called combinatorial expansion of the Mori-Zwanzig equation is just a simple ansatz assuming that the memory kernel can be written as a function composition of the correlation function:
\begin{align}\label{simple_anstaz}
K(t)=f\circ \hat{C}(t)=f[C(t)-C(0)],
\end{align}
where $\circ$ denotes the function composition and $\hat{C}(t)= C(t)-C(0)$. Since the $n$-th order derivatives of $C(t)$ and $K(t)$ are actually related through the recurrence relation \eqref{recurrence_relation}, we can use Fa\`a di Bruno formula to determine the exact form of the composition function $f$. The result can be summarized as follows:
\begin{theorem}\label{thm_combi_expansion}
Assuming $\gamma_0\neq 0$, then the $k$-th order derivative of $f$, denoted as $f_k$, are uniquely determined by a set of Laurent polynomials $\{L_k\}_{k=1}^{\infty}$ with $f_k=L_k(\gamma_0,\cdots,\gamma_{k-1})$. Moreover, we have the following series expansion of $K(t)$:
\begin{align}\label{K_combi_expansion}
    K(t)=\sum_{k=0}^{\infty}f_k\frac{(\hat{C}(t))^k}{k!}=
    \sum_{k=0}^{\infty}L_k(\gamma_0,\cdots,\gamma_{k-1})\frac{(C(t)-C(0))^k}{k!}.
\end{align}
In particular, $L_k$ can be constructed recursively by solving the system of recurrence equations:
\begin{equation}\label{eqn:rec}
\begin{aligned}
\begin{dcases}
\mu_0&=\gamma_0\\
f_0&=\mu_1\\
\mu_{n}&=\gamma_{n}-\sum_{k=1}^{n}\mu_{k-1}\gamma_{n-k},\quad n\geq 1\\
\mu_{n+1}&=\sum_{k=1}^nf_kB_{n,k}(\gamma_0,\gamma_1,\cdots,\gamma_{n-k}), \quad n\geq 1,
\end{dcases}
\end{aligned}
\end{equation}
where $B_{n,k}(x_1,x_2,\cdots,x_{n-k+1})$ is the partial Bell polynomial \cite{comtet2012advanced}. 
\end{theorem}
Theorem \ref{thm_combi_expansion} can be easily proved using the Fa\`a di Bruno's formula, which can be determined by the combinatorics of set partitions \cite{comtet2012advanced}. As we pointed out, $K(t)$ in \eqref{K_combi_expansion} corresponds to the memory kernel of the Mori-Zwanzig equation and $C(t)$ will be the correlation function. Hence, we will call \eqref{K_combi_expansion} a combinatorial expansion of the memory kernel and the resulting EOM for the correlation function (i.e. \eqref{eqn:Com_MEZ_time}) a combinatorial Mori-Zwanzig equation (CMZE). 
\begin{proof}
Recall that the Fa\`a di Bruno's formula for the derivative of the composite function $K(t)=f(\hat{C}(t))$ is given by \cite{comtet2012advanced}:
\begin{align*}
    \frac{d^n}{dt^n}K(t)=\frac{d^n}{dt^n}f(\hat{C}(t))&=\sum_{k=1}^nf^{(k)}(\hat{C}(t))\cdot B_{n,k}(\hat{C}'(t),\hat{C}''(t),\cdots,\hat{C}^{(n-k+1)}(t))\\
    &=\sum_{k=1}^nf^{(k)}(\hat{C}(t))\cdot B_{n,k}(C'(t),C''(t),\cdots,C^{(n-k+1)}(t)),
\end{align*}
where $B_{n,k}(x_1,x_2,\cdots,x_{n-k+1})$ is the partial Bell polynomial \cite{comtet2012advanced}. Evaluating this formula at $t=0$ and set $f_k=f^{(k)}(\hat{C}(t))|_{t=0}$, with the recurrence relation \eqref{recurrence_relation}, Taylor series expansions \eqref{taylor_C(t)_K(t)} for $C(t)$ and $K(t)$, and the Taylor series expansion for $f(w)$ at $w=0$:
\begin{align*}
    f(w)= f_0+ f_1w+ \frac{1}{2}f_2w^2+\cdots,
\end{align*}
simple computations will lead to \eqref{K_combi_expansion} and \eqref{eqn:rec}.
\end{proof}
For the reader's convenience, here we provide the specific form of the Bell polynomials. The partial Bell polynomials are given by explicitly by:
\begin{align*}
    B_{n,k}(x_1,x_2,\cdots,c_{n-k+1})=\sum\frac{n!}{j_1!j_2!\cdots j_{n-k+1}!}
    \left(\frac{x_1}{1!}\right)^{j_1}
    \left(\frac{x_2}{2!}\right)^{j_2}
    \cdots
    \left(\frac{x_{n-k+1}}{(n-k+1)!}\right)^{j_{n-k+1}},
\end{align*}
where the sum taken over all sequences $j_1,j_2,\cdots,j_{n-k+1}$ of non-negative integers such that these two conditions are satisfied:
\begin{align}
\begin{dcases}
&j_1+j_2+\cdots+j_{n-k+1}=k,\\
&j_1+2j_2+\cdots+(n-k+1)j_{n-k+1}=n.
\end{dcases}   
\end{align}
The sum of the partial Bell polynomial yields the complete Bell polynomial:
\begin{align*}
    B_n(x_1,\cdots,x_n)=\sum_{k=1}^{n}B_{n,k}(x_1,x_2,\cdots,x_{n-k+1})
\end{align*}
and the first few Bell polynomials are given by:
\begin{equation}\label{f_k_complete_Bell_poly}
\begin{aligned}
B_0&=1\\
B_1&=x_1\\
B_2&=x_1^2+x_2\\
B_3&=x_1^3+3x_1x_2+x_3\\
B_4&=x_1^4+6x_1^2x_2+4x_1x_3+3x_2^2+x_4.\\
\end{aligned}
\end{equation}
When $\gamma_0\neq 0$, using \eqref{f_k_complete_Bell_poly}, we can explicitly solve the recurrence equation \eqref{eqn:rec} and get the first few $f_k$s:
\begin{align*}
    f_0&=L_0(\gamma_0,\gamma_1)=\gamma_1-\gamma_0^2\\
    f_1&=L_1(\gamma_0,\gamma_1,\gamma_2)=\frac{\gamma_2}{\gamma_0}-2\gamma_1+\gamma_0^2\\
    f_2&=L_2(\gamma_0,\gamma_1,\gamma_2,\gamma_3)=-\frac{\gamma_1\gamma_2}{\gamma_0^3}+\frac{\gamma_1^2+\gamma_3}{\gamma_0^2}-2\frac{\gamma_2}{\gamma_0}+2\gamma_1-\gamma_0^2,
\end{align*}
which indeed are Laurent polynomials of $\gamma_k$s. To be noticed that for the example we gave where $\L$ is a skew-Hermitian operator with respect to the Hilbert space inner product  $\langle\cdot,\cdot\rangle$, we actually have $\gamma_{2j}=\mu_{2j}=0$, $j\geq 0$. For this case, the assumption $\gamma_0\neq 0$ is no longer valid. However, it is easy to show that a similar recurrence equation can be obtained. Eventually, we can get that the first few $f_k$s are given by
\begin{equation}\label{f_n_CCMZE}
\begin{aligned}
    f_0&=L_0(\gamma_1)=\gamma_1\\
    f_1&=L_1(\gamma_1,\gamma_3)=\gamma_3-\gamma_1^2\\
    f_2&=L_2(\gamma_1,\gamma_3,\gamma_5)=\frac{\gamma_5}{\gamma_1}-2\gamma_3+\gamma_1^2.
\end{aligned}
\end{equation}
One can see that $f_j$s are given by the same Laurent polynomial $L_k$ with different arguments. The series expansion \eqref{K_combi_expansion} can be plugged in MZE \eqref{eqn:time_ind_PMZE_POP}. Using Dyson's identity of this version \eqref{memory_reformulation}, we obtain the (commutative) combinatorial Mori-Zwanzig equation (CMZE) for the correlation function:
\begin{align}\label{eqn:Com_MEZ_time}
\frac{d}{dt}C(t)=\Omega C(t)+\sum_{n=0}^{\infty}\sum_{k=0}^n\frac{(-1)^{n-k}}{(n-k)!k!}f_n\int_0^tC^{k}(s)C(t-s)ds,
\end{align}
where $\Omega=\langle\L u(0),u(0)\rangle$. Here we used the binomial expansion and the fact that normalized time autocorrelation function $C(t)$ satisfies $C(0)=1$. \eqref{eqn:Com_MEZ_time} is a {\em self-consistent} equation of motion (EOM) for $C(t)$ since the memory kernel $K(s)$ is replaced by the polynomials of $C(s)$. At the crudest level of approximations, the equation resembles the mode-coupling equation developed for describing glassy dynamics \cite{bengtzelius1984dynamics,reichman2005mode}. In fact, the combinatorial expansion approach can be used to generalize the mode-coupling theory as we will show in Section \ref{sec:app}.

%
%
%
%
\section{Noncommutative combinatorial Mori-Zwanzig equation}
\label{sec:NCC-MZE_word}
Having considered the simplest combinatorial expansion for the Mori-Zwanzig equation, we need to consider common cases and generalize the basic idea to obtain a more fundamental theory based on operator algebra. This turns out to be trickier than the simple application of the Fa\`a di Bruno formula because operators, unlike real numbers, are non-commutative. One can easily see this when comparing the operator-form recurrence relation \eqref{recurrence_relation_op} and the real-number recurrence relation \eqref{recurrence_relation}. For instance, $\mu_1\gamma_1=\gamma_1\mu_1$, while for the corresponding operators $\P\L(\Q\L)\P\L^2\P\neq \P\L^2\P\L(\Q\L)\P$ in general. Moreover, the noncommutativity also comes from the projected evolution operator $\P\U(t,0)\P$ and its derivatives. To resolve all these issues, we need to introduce noncommutative combinatorial expansions for the operator-form Mori-Zwanzig equation. To this end, a series of new tools from enumerative combinatorics have to be introduced in order to get the most general theory we aimed to obtain. 

Specifically, we need to use the noncommutative Bell polynomials introduced by Munthe-Kaas \& Lundervold \cite{lundervold2011hopf,munthe1995lie,munthe1998runge}. In addition to that, two new families of noncommutative polynomials are also introduced. The first one is a new variant of the Bell polynomials, which will be called the {\em Type-II noncommutative Bell polynomials}. The second one will be called the {\em noncommutative bipartition polynomials} since it is closely related to the bipartition property of the projection operators $\P$ and $\Q$. With all these new polynomials, for the time-independent case, we prompt to use the following ansatz to get the combinatorial expansion for evolution operators:
\begin{align}\label{Com_expanion_time_depen}
   \P\L e^{t\Q\L}\Q\L\P=\sum_{n=0}^{\infty}\F_n(\P e^{t\L}\P-\P)^n.
\end{align}
Then we use combinatorics to determine explicitly what $\F_n$ operators really are. This is a brief prefatory remark on what we are going to do.
\subsection{Noncommutative Bell and bipartition polynomials}
\label{sec:bell_poly}
\paragraph{Type-I noncommutative Bell polynomials} Type-I noncommutative Bell polynomials was introduced by Munthe-Kaas \& Lundervold \cite{lundervold2011hopf,munthe1995lie,munthe1998runge} when discussing the Lie-Butcher theory on differential manifold. Noncommutative polynomials can be represented using abstract words or rooted planar trees. In this section, we will only consider the word representation. The tree representation will be introduced in Section \ref{sec:NCC-MZE_tree}. For the sake of completeness, we largely borrowed material on the Type-I noncommutative Bell polynomials from \cite{lundervold2011hopf}. For which, we do not claim any originality here.

Let $\mathbb{A}=\{a_j\}_{j=1}^{\infty}$ be an infinite alphabet and consider the free, unital, associative algebra $\Dd=\R\langle\mathbb{A} \rangle$ with grading given by $|a_j|=j$ and $|a_{j_1}a_{j_2}\cdots a_{j_k}|=j_1+j_2+\cdots+j_k$.  $\partial:\Dd\rightarrow \Dd$ is a linear differential operator satisfying
$\partial a_j=a_{j+1}$ and the Leibniz rule $\partial(w_1w_2)=\partial(w_1)w_2+w_1\partial(w_2)$ for all $w_1,w_2\in \mathbb{A}^*$. Here $w$ is called a word that consists of a finite string of elements selected with replacement from the alphabet $\mathbb{A}$. All possible words form a set $\mathbb{A}^*$. Each word has length $\#(w)$. For instance, a word $w=a_1a_2a_4a_4$ has length 4. The Type-I noncommutative Bell polynomial on the free, associative, noncommutative algebra $\Dd$ is a series of polynomials $\tB_n(a_1,\cdots,a_n)\in\Dd$ given by the recurrence relation:
\begin{equation}\label{rec_non_Bell_poly1}
\begin{aligned}
\begin{dcases}
    \tB_0&=\I\\
    \tB_n&=(a_1+\partial)\tB_{n-1}=(a_1+\partial)^n\I,\qquad n\geq 1.
\end{dcases}
\end{aligned}
\end{equation}
The first few $\tB_n$s are:
\begin{equation}\label{type1_Bell_example}
\begin{aligned}
    \tB_0&=\I\\
    \tB_1&=a_1\\
    \tB_2&=a_1^2+a_2\\
    \tB_3&=a_1^3+2a_1a_2+a_2a_1+a_3\\
    \tB_4&=a_1^4+a_2a_1^2+2a_1a_2a_1+3a_1^2a_2+3a_2^2+3a_1a_3+a_3a_1+a_4.
\end{aligned}
\end{equation}
According to this definition, we find that the noncommutative Bell polynomial can also be constructed recurrently without using $\partial$ by
\begin{align}\label{rec2_non_Bell_poly1}
    \tB_{n+1}(a_1,\cdots, a_{n+1})=\sum_{k=0}^n\binom{n}{k}\tB_k(a_1,\cdots,a_k)a_{n-k+1}.
\end{align}
We further introduce the {\em noncommutative partial Bell polynomials} $\tB_{n,k}=\tB_{n,k}(a_1,\cdots,a_{n-k+1})$, which only contains words $w$ of length $\#(w)=k\leq n$. Specifically,
\begin{align}\label{rec3_non_Bell_poly}
    \tB_{n,k}:=\sum_{\substack{w\in\mathbb{A}^{*},\\ |w|=n,\#(w)=k}}\kappa(\|w\|)\stir{n}{\|w\|}w,
\end{align}
where for $w=a_{j_1}a_{j_2}\cdots a_{j_k}$, we have $\|w\|=|a_{j_1}|,|a_{j_2}|,\cdots,|a_{j_k}|=j_1,j_2,\cdots,j_k$, and 
\begin{align}
\stir{n}{\|w\|}&=\binom{n}{|a_{j_1}|,|a_{j_2}|,\cdots,|a_{j_k}|}:=\frac{n!}{j_1!j_2!\cdots j_k!},\label{multi-binomial}\\
\kappa(\|w\|)&=\kappa(|a_{j_1}|,|a_{j_2}|,\cdots,|a_{j_k}|):=\frac{j_1j_2\cdots j_k}{j_1(j_1+j_2)\cdots(j_1+j_2+\cdots+j_k)}.\label{symmetric_group_partition}
\end{align}
With the  noncommutative partial Bell polynomials defined as \eqref{rec2_non_Bell_poly1}, it can be verified that $\tB_n$ and $\tB_{n,k}$ satisfies:
\begin{align*}
    \tB_n(a_1,\cdots,a_n)=\sum_{k=1}^n\tB_{n,k}(a_1,\cdots,a_{n-k+1}),
\end{align*}
which is similar to the commutative Bell polynomials. Coefficients \eqref{multi-binomial}-\eqref{symmetric_group_partition} can be interpreted combinatorially. One identifies $\stir{n}{\|w\|}$ as the multinomial coefficient, which is the number of distinct ways to permute a multiset of $n$ elements, where $j_i$ is the multiplicity of each of the $i$-th element. $\kappa$ forms a partition of unity on the symmetric group $S_k$, i.e. $\sum_{\sigma\in S_k}\kappa(\sigma(w))=1$, where $\sigma(w)$ denotes a permutations of letters in $w$. The noncommutative (partial) Bell polynomials become commutative (partial) Bell polynomials if the product of the associative algebra $\Dd$ is commutative, i.e. $a_ia_j=a_ja_i$ for any $i,j\geq 1$. One can easily verify this by summing up terms with the same {\em set} of subscript indices in \eqref{type1_Bell_example} and compare them with \eqref{f_k_complete_Bell_poly}.

The Type-I noncommutative Bell polynomials were introduced in \cite{munthe1995lie,munthe1998runge} to develop a Lie-Butcher theory to construct exponential integrator for time-dependent dynamical systems. The core problem in that context can be summarized as follows. Given a time-dependent vector field $\F_t$ on a smooth manifold $\M$. If $\Phi_{t,\F_t}$ is the associated flow map of diffeomorphism, then how to construct a series $\{\X_i\}_{i=1}^{\infty}\subset U(\X\M)$, where $U(\X\M)$ is the universal enveloping algebra  of all vector fields of all orders on $\M$, such that the pullback map $\Phi_{t,\F_t}^*\psi=\psi\circ\Phi_{t,\F_t}$ admits a series representation:
\begin{align*}
    \Phi_{t,\F_t}^*\psi=\sum_{j=0}^{\infty}\frac{t^j}{j!}\X_j[\psi],
\end{align*}
where $\X_j[\psi]$ is the {\em Lie derivative} of $\psi$ with respect to $\X_j$, which satisfies the noncommutative product rule: $\X_i\X_j[\psi]=\X_j[\X_i[\psi]]$. The Type-I noncommutative Bell polynomials provide a direct answer to this problem. Namely, using a homomorphism $\Dd\rightarrow U(\X\M)$ given by the map $a_i\rightarrow \F_t^{(i-1)}\big|_{t=0}=\F_i$, then $\X_j$ is explicitly given by the noncommutative Bell polynomial $\tB_j(\F_1,\cdots,\F_j)$, which leads to 
\begin{align}\label{exp_map_lie}
    \Phi_{t,\F_t}^*\psi=\sum_{j=0}^{\infty}\frac{t^j}{j!}\tB_j(\F_1,\cdots,\F_j)[\psi].
\end{align}
This can be viewed as a special Lie-Butcher series in the word representation. More generally, We denote $\Dd^*$ the dual of $\Dd$, i.e. $\Dd^*:\Dd\rightarrow \R$. Further define a concatenation homomorphism $\mathcal{C}_v:\Dd\rightarrow U(\X\M)$, then any Lie-Butcher series $\B_t:\Dd^*\rightarrow U(\X\M)^*$ is a map between the dual space explicitly given by
\begin{align}\label{general_lie_butcher}
    \B_t(\alpha)=\sum_{w\in\mathbb{A}^*}t^{|w|}\alpha(w)\mathcal{C}_v(w).
\end{align}
The Type-I noncommutative Bell polynomials will be used in Section \ref{sec:Time-d-CMZE} when we develop the combinatorial expansion for the time-dependent MZE \eqref{eqn:time_d_MZE_POPP}. In fact, the problem we will encounter is similar to the construction of Lie-Butcher series \eqref{exp_map_lie} with $\F_i$ replacing by operators $\L_i$ or $\Q\L_i$. Moreover, the resummation series introduced at the end of Section \ref{sec:Time-d-CMZE} are similar to the general Lie-Butcher series \eqref{general_lie_butcher}.  

\paragraph{Type-II noncommutative Bell polynomials} 
The construction of Type-II noncommutative Bell polynomials is much tailored for its usage when expanding the RHS of the ansatz \eqref{Com_expanion_time_depen} with respect to $t$. Hence, the rationale behind its construction is not clearly seen now but will become obvious in a minute. Specifically, using a different infinite alphabet $\Aa'=\{a'_1,a'_2,\cdots\}$ and the associative word algebra $\Dda$, the first few $\hB_n$s for the Type-II noncommutative Bell polynomials are given by :
\begin{equation}\label{type2_bell_example}
\begin{aligned}
    \hB_0&=\I\\
    \hB_1&=a'_1\\
    \hB_2&=a_1^{'2}+a'_2\\
    \hB_3&=a_1^{'3}+\frac{3}{2}a'_1a'_2+\frac{3}{2}a'_2a'_1+a'_3\\
    \hB_4&=a^{'4}_1+2a'_2a^{'2}_1+2a'_1a'_2a'_1+2a_1^{'2}a'_2+3a^{'2}_2+2a'_1a'_3+2a'_3a'_1+a'_4.
\end{aligned}
\end{equation}
The general expression for $\hB_n$ is given in terms of the the {\em noncommutative partial Bell polynomials} $\hB_{n,k}=\hB_{n,k}(a'_1,\cdots,a'_{n-k+1})$, where
\begin{align}\label{rec3_non_Bell_poly_type2}
    \hB_{n,k}:=\sum_{\substack{w\in\mathbb{A}^{*},\\ |w|=n,\#(w)=k}}\frac{1}{k!}\stir{n}{\|w\|}w.
\end{align}
Here $\stir{n}{\|w\|}$ is the multinormial coefficient defined as \eqref{multi-binomial}. $\hB_{n,k}$ only contains words $w$ of length $\#(w)=k\leq n$. The full Bell polynomials $\hB_n$ and the partial ones $\hB_{n,k}$ satisfy:
\begin{align*}
    \hB_n(a'_1,\cdots,a'_n)=\sum_{k=1}^n\hB_{n,k}(a'_1,\cdots,a'_{n-k+1}),
\end{align*}
which is similar to the commutative Bell polynomials. 
Type-II (partial) Bell polynomials also degenerate to the commutative (partial) Bell polynomials if the product of the associative algebra $\Dda$ comes commutative, i.e. $a'_ia'_j=a'_ja'_i$ for any $i,j\geq 1$. One can easily verify this by summing up terms with the same {\em set} of subscript indices in \eqref{type2_bell_example} and compare them with \eqref{f_k_complete_Bell_poly}. 

\paragraph{Non-commutative bipartition polynomials}
To discuss the Taylor series expansion of the LHS of \eqref{Com_expanion_time_depen}, we consider a different infinite alphabet $\mathbb{B}=\{b_j\}_{j=1}^{\infty}$ and define a new series of noncommutative polynomials of words in $\Ddb$, called the {\em noncommutative bipartition polynomial}. Specifically, this set of polynomials is defined using the recurrence relation:
\begin{equation}\label{rec_non_bipart_poly}
\begin{aligned}
\begin{dcases}
    \tP_0&=\I\\
    \tP_1&=b_1\\
    \tP_n&=\sum_{k=0}^{n-1}(-1)^{\delta_{k0}+1}b_{n-k}\tP_{k}(b_1,\cdots,b_{n-k}).
\end{dcases}
\end{aligned}
\end{equation}
The first few $\tP_n$s are:
\begin{align*}
    \tP_0&=\I\\
    \tP_1&=b_1\\
    \tP_2&=-b_1^2+b_2\\
    \tP_3&=b_1^3-b_1b_2-b_2b_1+b_3\\
    \tP_4&=-b_1^4+b_2b_1^2+b_1b_2b_1+b_1^2b_2-b_2^2-b_1b_3-b_3b_1+b_4.
\end{align*}
If we further define the {\em noncommutative partial bipartition polynomial} $\tP_{n,k}$ as 
\begin{align}\label{rec1.5_non_bipart_poly}
    \tP_{n,k}:=\sum_{\substack{w\in\mathbb{B}^{*},\\ |w|=n,\#(w)=k}}(-1)^{k+1}w,
\end{align}
where the summation goes over for all possible ordered partitions of the integer $k$. Then we get a similar relation:
\begin{align}\label{rec3_non_bipart_poly}
    \tP_n(b_1,\cdots,b_n)=\sum_{k=1}^n\tP_{n,k}(b_1,\cdots,b_{n-k+1}).
\end{align}
Similarly, if $\Ddb$ is commutative, we have the {\em commutative bipartition polynomials} $P_n$. The first few $P_n$s are:
\begin{align*}
    P_0&=\I\\
    P_1&=b_1\\
    P_2&=-b_1^2+b_2\\
    P_3&=b_1^3-2b_1b_2+b_3\\
    P_4&=-b_1^4+3b_1^2b_2-2b_1b_3-b_2^2+b_4.
\end{align*}
$P_n$s can also be written as the summation of the {\em commutative partial bipartition polynomials} $P_{n,k}$ as:
\begin{align}\label{rec2_bipart_poly_decomp}
    P_n(b_1,\cdots,b_n)=\sum_{k=1}^nP_{n,k}(b_1,\cdots,b_{n-k+1}).
\end{align}
Using a little bit of combinatorics, we can determine that 
\begin{align}\label{rec2_bipart_poly}
    P_{n,k}(b_1,\cdots,b_k)=\sum_{\substack{1\leq j_i,1\leq i\leq k,\\j_1+\cdots+j_k=n,\\j_1\leq j_2\leq\cdots \leq j_k}}(-1)^{k+1}{n \brack j_1,\cdots,j_k} b_{j_1}b_{j_2}\cdots b_{j_k},
\end{align}
where ${n \brack j_1,\cdots j_k}$ is the number of ways to partition a set of $n$ unlabeled objects into $k$ labeled bins, with each bin has at least one object, and one of the bins has $j_1$ objects, one of the bins has $j_2$ objects, and so on. One can also calculate $P_n$ recursively as:
\begin{align*}
    P_n(b_1,b_2,\cdots,b_n)=b_n-\sum_{k=1}^{n-1}b_{k} P_{n-k}(b_1,b_2,\cdots,b_{n-k}).
\end{align*}
We notice that this recurrence relation is exactly the same as \eqref{recurrence_relation}. 
\subsection{Symbolic algebraic equation of words}
\label{sec:alge_eqn_words}
Now suppose we have a system of algebraic equations which connects the noncommutative (commutative) polynomials $Q_n^a=(a_1,\cdots,a_n)$ in $\Dd$ and another noncommutative (commutative) polynomials $Q_n^b=(b_1,\cdots,b_n)$ in $\Ddb$ by
\begin{equation}\label{alge_eqn_word}
\begin{aligned}
\begin{dcases}
    Q^b_{m}(b_1,\cdots,b_{m})&=f_0,\\
    Q^b_{n+m}(b_1,\cdots,b_{n+m})&=\sum_{k=1}^nf_kQ^a_{n,k}(a_1,\cdots,a_{n-k+1}),\qquad \text{for all $n\geq 1$}.
\end{dcases}
\end{aligned}
\end{equation}
Equation \eqref{alge_eqn_word} will be called a {\em symbolic algebraic equation of words} and it is formulated in a rather general way. Here $\Dd$ and $\Ddb$ can be general alphabets, not necessarily the same as the ones used in the previous section. $Q^b_{n}$ and $Q^a_{n}$ can be any commutative or noncommutative polynomials of words introduced in Section \ref{sec:bell_poly}. $m$ can be any non-negative integer. In this section, we discuss the existence and uniqueness of the solution to \eqref{alge_eqn_word}, regardless of the specific setting of $\Dd$, $\Ddb$ and $Q^b_{n}$ and $Q^a_{n}$. In the follow-up sections, the noncommutative combinatorial expansion for the MZE can be obtained by solving \eqref{alge_eqn_word}, after we introduce algebra homomorphisms which translate the alphabets $\Dd$ and $\Ddb$ into operators. Hence, the discussion in this section will establish the theoretical foundation for the existence and uniqueness of the noncommutative combinatorial expansion theory, which we now detail. The following lemma provides two sufficient conditions to guarantee the solvability of Eqn \eqref{alge_eqn_word}.
\begin{lemma}\label{lemma1}
The system of algebraic equation of words \eqref{alge_eqn_word} has an unique solution $\{f_k\}_{k=0}^{\infty}$ for the following two cases:
\begin{enumerate}
    \item The letter $a_1\in\mathbb{A}$ is invertible and has inverse $a_1^{-1}$ in $\Dd^+$. Here, we considered  an extended associative algebra $\Dd^+$ which includes $a_1^{-1}$ into the original alphabet $\mathbb{A}$. $a_1$ in invertible implies $a_1a_1^{-1}=a_1^{-1}a_1=\I$. 
    \item All the odd letters in $\Aa$ and $\Bb$ are zero, i.e. $a_1,a_3\cdots=0$ and $b_1,b_3,\cdots=0$. $a_2$ is invertible in $\Dd$ and the inverse is $a_2^{-1}$. Moreover, $m$ in Eqn \eqref{alge_eqn_word} is an even integer. The invertibility of $a_2$ has the exact same meaning as in Case 1.  
\end{enumerate}
In each case, $\{f_k\}_{k=0}^{\infty}$ are polynomials in $\Dd^+\cup\Ddb$, where $\Dd^+$ corresponds to extended alphabets $\{a_1^{-1},a_1,a_2,\cdots\}$ and $\{a_2^{-1},a_2,a_4,\cdots\}$ respectively. Moreover, $\{f_k\}_{k=0}^{\infty}$ can be constructed recursively.
\end{lemma}
\begin{proof}
The proof can be obtained using direct computation and combinatorics. We first consider case 1. For $n=1$, $Q_{1,1}^a=a_1$ for all possible polynomials given in the previous section, hence the first equation of \eqref{alge_eqn_word} is given by: $Q^b_{m+1}=f_1a_1$, which implies $f_1=Q^b_{m+1}a_1^{-1}$ when $a_1^{-1}$ exists. For $n\geq 2$, $Q^a_{n}$ can be any one of the Bell polynomials $\tB_{n}$, $\hB_{n}$, $B_n$ or the bipartition polynomials $\tP_{n}$, $P_n$. For all these cases, we always have $Q^a_{n,n}(a_1)=(\pm a_1)^n=(\pm 1)^{n+1}a_1^n$, where $+$ is for the Bell polynomials case and $-$ is for the bipartition polynomial case. As a result, we have 
\begin{align*}
    Q^b_{n+m}(b_1,\cdots,b_{n+m})&=\sum_{k=1}^{n-1}f_kQ^a_{n,k}(a_1,\cdots,a_{n-k+1})+f_nQ^a_{n,k}(a_1)\\
    &=\sum_{k=1}^{n-1}f_kQ^a_{n,k}(a_1,\cdots,a_{n-k+1})+(\pm 1)^{n+1}f_na_1^n.
\end{align*}
Since $a_1$ is invertible, $f_n$ can be solved recursively by:
\begin{align*}
   (\pm 1)^{n+1} f_na_1^n&=Q^b_{n+m}(b_1,\cdots,b_{n+m})-\sum_{k=1}^{n-1}f_kQ^a_{n,k}(a_1,\cdots,a_{n-k+1})\\
 \Rightarrow\quad   (\pm 1)^{n+1}f_n&=Q^b_{n+m}(b_1,\cdots,b_{n+m})(a_1^{-1})^n-\sum_{k=1}^{n-1}f_kQ^a_{n,k}(a_1,\cdots,a_{n-k+1})(a_1^{-1})^n,\qquad n\geq 2.
\end{align*}
As an example, if $Q^b_{n+m}=\tP_n$ and $Q^a_{n}=\tB_n$, then the first few $f_k$s are:
\begin{align*}
    f_1&=b_1a_1^{-1}\\
    f_2&=-b_1^2a_1^{-2}+b_2a_1^{-2}-b_1a_1^{-1}a_2a_1^{-2}\\
    f_3&=b_3a_1^{-3}-b_1b_2a_1^{-3}-b_2b_1a_1^{-3}+b_3a_1^{-3}+2b_1^2a_1^{-1}a_2a_1^{-3}-2b_2a_1^{-1}a_2a_1^{-3}+2b_1a_1^{-1}a_2a_1^{-1}a_2a_1^{-3}\\
    &\quad +b_1^2a_1^{-2}a_2a_1^{-2}-b_2a_1^{-2}a_2a_1^{-2}+b_1a_1^{-1}a_2a_1^{-2}a_2a_1^{-2}-b_1a_1^{-1}a_3a_1^{-3}.
\end{align*}
For the second case, we first note that the commutative Bell/bipartition polynomials can be viewed as a special case of the general noncommutative Bell/bipartition polynomials. Moreover, for each $n$, the noncommutative Bell polynomials $\tB_n$, $\hB_n$ and bipartition polynomials $\tP_n$ only differ by the coefficients in front of the words. As a result, we can simply get a proof for a special case when $Q^b_{n+m}=\tP_{n+m}$ and $Q^a_n=\tB_n$. When the proof uses no information regarding the coefficients, then without loss of generality, we can conclude that the result is valid for all other cases. 

Bearing this in mind, now we assume $Q^b_{n+m}=\tP_{n+m}$ and $Q^a_n=\tB_n$. 
For any $w=a_{j_1}a_{j_2}\cdots a_{j_k}\in\tB_{n},\tP_{n+m}$ for any odd $n$, $n+m$, there must exists a letter with odd $j_i$, otherwise it contradict with the fact that $j_1+j_2+\cdots j_k$ is odd. Since all odd alphabet $a_j,b_j$ vanished, this implies all odd order $\tP_{n+m}$ and $\tB_{n}$ vanished. Hence the first equation of Eqn \eqref{alge_eqn_word} becomes $\tP_{m+1}=\tB_1=0$ and whole equation becomes a system of equations for $\{b_2,b_4,\cdots\} $ and $\{a_2,a_4,\cdots\}$. Now we only need to verify that such an equation can be solved recursively assuming only $a_2^{-1}$ is invertible. 
To this end, we note that for any word belonging to $\tB_{n}$, the index runs through all possible partitions of the integer $n$. From the above analysis, we know that only the partition that involves even alphabets $a_{2},a_{4},\cdots$ survives. Hence for any $w\in \tB_n$, only the word satisfying $\#(w)2\leq n$ are nonzero, which implies $\#(w)\leq \lfloor n/2\rfloor$, where $\lfloor\cdot\rfloor$ is the floor function. 
Naturally, we found that $\tB_n$ is a polynomial with $f_{k_1}a_2^{k_1}$, $k_1\leq \lfloor n/2\rfloor$, $f_{k_2}a_4^{k_2}$, $k_2\leq \lfloor n/4\rfloor$ $\cdots$ with $k_1>k_2>\cdots$. Hence $\tP_{m+2}=\tB_2$ is an equation for $a_2$ and $b$s, $\tP_{m+4}=\tB_4$ is an equation for $a_2,a_4$ and $b$s $\cdots$, and new $f_k$s always first appear with $f_{k_1}a_2^{k_1}$, $k_1\leq \lfloor n/2\rfloor$. As a result, one can solve for $f_k$s assuming only $a_2^{-1}$ is invertible. 
The above arguments are based on simple combinatorial rules that are similar to the pigeonhole principle. Note that we have not used any information regarding the word coefficients, Naturally, the results hold for any $Q^b_{n+m}$ and $Q^a_n$ for even $m$. As an example, if $Q^b_{n+m}=\tP_n$ and $Q^a_{n}=\tB_n$, for the second case, the first few $f_k$s are:
\begin{align*}
    f_1&=b_2a_2^{-1}\\
    f_2&=\frac{1}{3}b_2^2a_2^{-2}+\frac{1}{3}b_4a_2^{-2}-\frac{1}{3}b_2a_2^{-1}a_4a_2^{-2}.
\end{align*}

\end{proof}

Lemma \ref{lemma1} provides two sufficient conditions to get the unique solution for Eqn \eqref{alge_eqn_word}. These two cases happen to be the most common ones we will encounter when using CMZE in applications, as we have seen in Section \ref{sec:CCMZE}. More specifically, when $\L$ is a non-skew-Hermitian operator, we will have Case 1. Otherwise, we will get Case 2. In principle, we can also consider other situations 
\footnote{
For instance, if we only have $a_1=\cdots a_{i-1}=0$ and $a_i^{-1}$ exists, the words in $\tB_{n}$ which contains any $a_j$, $j\leq i-1$ will be $0$. Combinatorially, this means only the partition that involves $a_{i},a_{i+1},\cdots$ survives. Hence for any $w\in \tB_n$, only the word satisfying $\#(w)i\leq n$ are nonzero, which implies $\#(w)\leq \lfloor n/i\rfloor$. Naturally, we found that: (I) The second nonzero $\tB_n$ is when $n=2i$, which contains words of length 1 and 2. (2) For any $2i\leq n< 3i$, $\tB_n$ only contains  words of length 1 and 2. (3) For any $ki\leq n< (k+1)i$, $\tB_n$ only contains words of length $1,2\cdots,k$. Under this setting, we find that the equation between $ki\leq n< (k+1)i$ becomes {\em overdetermined} equation for $f_1,\cdots f_k$. Generally speaking, this indicates \eqref{alge_eqn_word} has no solution.}
, but the existence and uniqueness of the solution may not be guaranteed. Luckily, we probably will not deal with any of these in applications.

\subsection{Time-independent CMZE}\label{sec:main_thm}
With the noncommutative Bell and bipartition polynomials and the symbolic algebraic equation \eqref{alge_eqn_word}, now we can derive the combinatorial expansion for the operator-form MZEs \eqref{eqn:time_ind_PMZE_POP} and \eqref{eqn:time_d_MZE_POPP} and prove the main theorems of this paper. The road map is as follows. We first prove that the Taylor series expansion at the two sides of the combinatorial expansion ansatz:
\begin{align}\label{type1_expan}
   \P\L e^{t\Q\L}\Q\L\P=\sum_{n=0}^{\infty}\frac{1}{n!}\F_n(\P e^{t\L}\P-\P)^n
\end{align}
is given by the noncommutative polynomials introduced in Section \ref{sec:bell_poly}, after choosing suitable algebra homomorphisms to map $\Dd$ and $\Ddb$ into the free, unital, and associative algebra of operators. Naturally, \eqref{type1_expan} can be reformulated as the algebraic equation for words \eqref{alge_eqn_word}, which can be solved recursively using the result of Lemma \ref{lemma1}. We begin our discussion with the combinatorial expansion for the Mori-Zwanzig equation of time-independent systems. To this end, two identical alphabets $\Aa'=\Bb$ are considered, along with the homomorphisms induced by the map $a'_i=b_i\rightarrow \P\L^{i}\P$. As a result, the solution to Eqn \eqref{alge_eqn_word} with $Q_{n+m}^b=\tP_{n+2}$ and $Q_{n}^a=\hB_n$ yields the combinatorial expansion of the Mori-Zwanzig equation for the time-independent system. The result can be summarized as follows.
\begin{theorem}
\label{thm_combi_expansion1}
(Time-independent CMZE) Suppose the identity operator $\I$ in Banach space $X$ can be orthogonally partitioned by two projection operators $\P,\Q$ as $\I=\P+\Q$ with $\P\Q=0$. If the closure of the operators $\L:X\rightarrow X$ and $\Q\L:X\rightarrow X$ generate two one-parameter semigroups $e^{t\L}$ and $e^{t\Q\L}$ respectively, moreover, there exists a subspace $V\subset X$ such that $\text{Ran}(\P)\subset V$ and the restriction of $\P\L\P$ in V, i.e. $\P\L\P|_V$, is invertible, then we have the following combinatorial expansion of the differential-form Dyson's identity \eqref{eqn:time_ind_PMZE_POP}:
\begin{align}\label{op_id_1_nonc}
    \frac{d}{dt}\P e^{t\L}\P =\P e^{t\L}\P\L\P+\sum_{n=0}^{\infty}\sum_{k=0}^n\frac{(-1)^{n-k}}{(n-k)!k!}\int_0^t\P e^{(t-s)\L}\F_n\left(\P e^{s\L}\P\right)^kds,
\end{align}
where in \eqref{op_id_1_nonc}, $\F_n$s are the polynomials of operators which can be constructed recursively by solving the operator algebraic equation in $V$:
\begin{equation}\label{eqn:rec_op_sol}
\begin{aligned}
\begin{dcases}
\F_0&=\P\L\P,\\
\F_n&=-\tP_{n+2}(\P\L\P,\cdots,\P\L^{n+2}\P)(\P\L\P|_{V})^{-n}\\
&\ \ \ \ \ \ \ +\sum_{k=1}^{n-1}\F_{k}\hB_{n,k}(\P\L\P,\cdots,\P\L^{(n+1-k)}\P)(\P\L\P|_{V})^{-n}, \qquad n\geq 1.
\end{dcases}
\end{aligned}
\end{equation}
Here $(\P\L\P|_{V})^{-1}$ is the inverse of operator $\P\L\P$ in $V$. 
\end{theorem}
\begin{proof}
We follow the aforementioned road map to get the proof. Under the map $b_i\rightarrow\P\L^i\P$, one easily finds that the recurrent relation \eqref{recurrence_relation_op} is exactly the same as \eqref{rec_non_bipart_poly}, hence $\tP_n(\P\L\P,\cdots,\P\L^n\P)=\P\L(\Q\L)^{n-1}\P$. As a result, the Taylor series expansion of the LHS of \eqref{type1_expan} can be written as
\begin{align}\label{type1_LHS}
    \P\L e^{t\Q\L}\Q\L\P=\sum_{n=0}^{\infty}\frac{t^n}{n!}\tP_{n+2}(\P\L\P,\cdots,\P\L^{n+2}\P).
\end{align}
On the other hand, under the map $a_i\rightarrow\P\L^{i}\P$, we have the following expansion for $(\P e^{t\L}\P-\P)^n$: 
\begin{align*}
(\P e^{t\L}\P-\P)^n
&=\left(\sum_{j_1=1}^{\infty}\frac{t^{j_1}}{j_1!}\P\L^{j_1}\P\right)\left(\sum_{j_2=1}^{\infty}\frac{t^{j_2}}{j_2!}\P\L^{j_1}\P\right)\cdots \left(\sum_{j_n=1}^{\infty}\frac{t^{j_n}}{j_n!}\P\L^{j_n}\P\right)\\
&=\sum_{m=n}^{\infty}\frac{t^m}{m!}\sum_{\substack{j_1+j_1+\cdots j_n=m,\\ j_i\geq 1}}\frac{m!}{j_1!j_2!\cdots j_n!}(\P\L^{j_1}\P)(\P\L^{j_2}\P)\cdots(\P\L^{j_n}\P)
\\
&=\sum_{m=n}^{\infty}\frac{t^m}{m!}n!\hB_{m,n}(\P\L\P,\cdots,\P\L^{m+1-n}\P).
\end{align*}
The above formula is just the noncommutative version of Fais's derivation of Fa\`a di Bruno formula (See Page 227 in \cite{johnson2002curious}). Substituting this expansion into the RHS of \eqref{type1_expan}, we get
\begin{align}
\sum_{n=0}^{\infty}\frac{1}{n!}\F_n(\P e^{t\L}\P-\P)^n
&=\sum_{n=0}^{\infty}\frac{1}{n!}\F_n\sum_{m=n}^{\infty}
\frac{t^m}{m!}n!\hB_{m,n}(\P\L\P,\cdots,\P\L^{m+1-n}\P)
\nonumber\\
&=
\sum_{m=0}^{\infty}\frac{t^m}{m!}\sum_{n=0}^m\F_n\hB_{m,n}(\P\L\P,\cdots,\P\L^{m+1-n}\P).
\label{type1_RHS}
\end{align}
Matching different order $t^n$ terms on the RHS of \eqref{type1_LHS} and \eqref{type1_RHS}, we get the recursive algebraic equation:
\begin{equation}\label{eqn:rec_op_nonc}
\begin{aligned}
\begin{dcases}
&\F_0=\P\L^2\P-(\P\L\P)^2,\\
&\tP_{n+2}(\P\L\P,\cdots,\P\L^{n+2}\P)=\sum_{k=1}^{n}\F_{k}\hB_{n,k}(\P\L\P,\cdots,\P\L^{(n+1-k)}\P)\qquad n\geq 1.
\end{dcases}
\end{aligned}
\end{equation}
Here we used the definition $\hB_{n,k}=0$ for $n>0,k=0$. Under the assumption that $\P\L\P$ in invertible in subspace $V\subset X$, with the Case 1 result of Lemma \ref{lemma1}, we confirm that algebraic equation \eqref{eqn:rec_op_nonc} has an unique solution $\{\F_k\}_{k=1}^{\infty}$. In fact, according to the definition of the Type-II noncommutative polynomial $\hB_{n,k}$, the second equation in \eqref{eqn:rec_op_nonc} can be decomposed as:
\begin{align*}
    \tP_{n+2}(\P\L\P,\cdots,\P\L^{n+2}\P)=\sum_{k=1}^{n-1}\F_{k}\hB_{n,k}(\P\L\P,\cdots,\P\L^{(n+1-k)}\P)+\F_n(\P\L\P)^n.
\end{align*}
Then we have 
\begin{align*}
\F_n&=-\tP_{n+2}(\P\L\P,\cdots,\P\L^{n+2}\P)(\P\L\P|_{V})^{-n}\\
&\ \ \ \ \ \ \ +\sum_{k=1}^{n-1}\F_{k}\hB_{n,k}(\P\L\P,\cdots,\P\L^{(n+1-k)}\P)(\P\L\P|_{V})^{-n}, \qquad n\geq 1.
\end{align*}
Following this procedure, we can get the explicit expression for the first few $\F_k$s:
\begin{equation}\label{first_few_f_n}
\begin{aligned}
    \F_0&=\P\L^2\P-(\P\L\P)^2\\
    \F_1&=(\P\L\P)^2-\P\L\P\L^2\P(\P\L\P|_V)^{-1}-\P\L^2\P+\P\L^3\P(\P\L\P|_V)^{-1}\\
    \F_2&=\P\L\P\L^2\P(\P\L\P|_V)^{-1}\P\L^2\P(\P\L\P|_V)^{-2}
    -\P\L^3\P(\P\L\P|_V)^{-1}\P\L^2\P(\P\L\P|_V)^{-2}\\
    &\ \ -(\P\L\P)^2+\P\L\P\L^2\P(\P\L\P|_V)^{-1}-\P\L\P\L^3\P(\P\L\P|_V)^{-2}+\P\L^2\P\\
    &\ \ -\P\L^3\P(\P\L\P|_V)^{-1}+\P\L^4\P(\P\L\P|_V)^{-2}.
\end{aligned}
\end{equation}
In the derivation, we have consistently used the idempotence $\P^2=\P$. Moreover, with $\P^2=\P$, operator $\P$ and $\P e^{t\L}\P$ commute. Naturally, we can directly use the binomial theorem and obtain the following expansion:
\begin{align}\label{futher-type1_expansion}
   \P\L e^{t\Q\L}\Q\L\P=\sum_{n=0}^{\infty}\frac{1}{n!}\F_n(\P e^{t\L}\P-\P)^n=\sum_{n=0}^{\infty}\sum_{k=0}^n\frac{(-1)^{n-k}}{(n-k)!k!}\F_n(\P e^{t\L}\P)^k.
\end{align}
Here we emphasize that the above expansion is valid since $\text{Ran}(\P)\subset V$ and $\F_n$ always act after $\P$. This guarantees the invertibility of $\P\L\P$ in the expression of $\F_n$s. Substituting operator identity \eqref{futher-type1_expansion} into the MZ memory integral in \eqref{eqn:time_ind_PMZE_POP}, we obtain the combinatorial expansion \eqref{op_id_1_nonc} for the time-differential Dyson's identity. Further applying this operator equation into any observable function $v\in V\subset X$, we get the noncommutative CMZE for the projected quantity $\P e^{t\L}\P v$, which normally corresponds the time-autocorrelation function, or the Green's function for observable $v$. 
\end{proof}
\paragraph{Remark 1} In the whole derivation, we implicitly assumed the time-differentiability of semigroups up to an arbitrary order. We did not discuss the convergence problem for the semigroup Taylor series expansion. Hence the resulting combinatorial series expansion \eqref{op_id_1_nonc} is {\em not} guaranteed to be convergent for $v\in V\subset X$. Generally speaking, the convergence problem is rather difficult to address on a general basis since for most situations infinitesimal generators $\L$ and $\Q\L$ are unbounded operators with different analytical properties. In particular, the convergence of the Taylor series expansion for the effective semigroups $\P e^{t\L}\P$ and $\P\L e^{t\Q\L}\Q\L\P$ is related to the convergence of the Neumann series for $\P(\I-\L)^{-1}\P$ and $\P(\I-\Q\L)^{-1}\P$ in $\text{Ran}(\P)\subset X$, where $(\I-\L)^{-1}$ and $(\I-\Q\L)^{-1}$ are the resolvent operators. We will leave this as an open topic and await further investigation. Numerically, applications of the MZE in \cite{zhu2021effective,zhu2020generalized} have shown the convergence of these Neumann series for many Hamiltonian and stochastic systems. 
\paragraph{Remark 2} There are other ways to get a combinatorial expansion for the MZE. In fact, the following variants of the expansion ansatz:
\begin{align*}
\P\L e^{t\Q\L}\Q\L\P&=\partial_t\P\L e^{t\Q\L}\P=\partial_t\sum_{n=0}^{\infty}\frac{1}{n!}\F_n(\P e^{t\L}\P-\P)^n,\\
\P\L e^{t\Q\L}\Q\L\P&=\partial_t\P\L e^{t\Q\L}\P=\partial_t\sum_{n=0}^{\infty}\frac{1}{n!}\F_n(\P e^{t\L}\P)^n,\\
\P\L e^{t\Q\L}\Q\L\P&=\sum_{n=0}^{\infty}\frac{1}{n!}\F_n(\P e^{t\L}\P)^n
\end{align*}
also lead to self-consistent expansions of the MZE. Moreover, one can also place operator coefficients $\F_n$ to the RHS of $(\P e^{t\L}-\P)$ as:
\begin{align}\label{type1_variant_expan}
\P\L e^{t\Q\L}\Q\L\P=\sum_{n=0}^{\infty}\frac{1}{n!}(\P e^{t\L}\P-\P)^n\F_n.
\end{align}
For all these variations, interested readers may follow our road map and get the explicit form of the resulting combinatorial expansion series. Eventually, one is led to define a different version of the Type-II noncommutative Bell polynomials and then solve the corresponding symbolic equation of words. 
\begin{corollary}
\label{type1_co1}
If $\P\L^{2k+1}\P=0$ for $k\geq 0$ and $\P\L^2\P$ is invertible in $V\subset X$ and has inverse $\P\L^2\P|_V^{-1}$, assuming all other conditions listed in Theorem \ref{thm_combi_expansion1} are met, then we have a different set of operators $\{\F'_K\}_{k=1}^{\infty}$ such that 
the following differential-form Dyson's identity hold:
\begin{align}\label{op_id_1_nonc_co1}
    \frac{d}{dt}\P e^{t\L}\P =\P e^{t\L}\P\L\P+\sum_{n=0}^{\infty}\sum_{k=0}^n\frac{(-1)^{n-k}}{(n-k)!k!}\int_0^t\P e^{(t-s)\L}\F'_n\left(\P e^{s\L}\P\right)^kds,
\end{align}
where $\{\F'_n\}_{n=1}^{\infty}$ can be constructed recursively.
\end{corollary}
\begin{proof}
Using The Case 2 result of Lemma \ref{lemma1}, it is easy to get the proof. In fact, $\{\F'_n\}_{n=1}^{\infty}$ can be constructed recursively by solving the even order of \eqref{eqn:rec_op_nonc}. 
\end{proof}

In the rest of this paper, we will no longer consider the second case in Lemma \ref{lemma1} since it is direct to get the specific form of the combinatorial expansion based on the Case 1 result. 
%
%
%
%
\paragraph{Case study:}
We reconsider the simple example introduced in Section \ref{sec:CCMZE} to explain when the noncommutative CMZE \eqref{op_id_1_nonc} reduces to the commutative CMZE \eqref{eqn:Com_MEZ_time} and when it does not. Choose a Hilbert space $H=X$ and define one-dimensional projection operator $\P(\cdot)=\langle\cdot,u(0)\rangle u(0)$, $\|u(0)\|^2=1$. It is easy to verify that $\text{Ran}(\P)=V=\R u(0)\subset H$ is a closed subspace under the action of operators which are chosen from the alphabet $\{\P\L\P,\P\L^2\P ,\cdots\}$. Moreover, we also have $[\P\L^i\P,\P\L^j\P]=0$, $i,j\geq 0$ in $\R u(0)$. Hence $\Dda$, $\Ddb$ becomes commutative algebra in $\R u(0)$ after the homomorphism. As a result, the noncommutative polynomials $\tP_{n+2}$ and $\tB_{n}$ in algebraic equation \eqref{eqn:rec_op_nonc} becomes commutative polynomials $P_{n+2}$ and $B_{n}$. Further denoting $\mu_nu(0)=\P\L(\Q\L)^n\P u(0)$, $\gamma_n u(0)=\P\L^{n+1}\P u(0)$, and we can get the result of Theorem \ref{thm_combi_expansion} from Theorem \ref{thm_combi_expansion1}. Naturally, the noncommutative MZE \eqref{op_id_1_nonc} (acting on $u(0)$) becomes the commutative MZE \eqref{eqn:Com_MEZ_time}. For more general Mori-type projection operator defined as $\P(\cdot)=\sum_{j=1}^N\langle\cdot,u_j(0)\rangle u_j(0)$, where $\langle u_i(0),u_j(0)\rangle=\delta_{ij}$. The action of the elements chosen from the alphabet $\{\P\L\P,\P\L^2\P,\cdots\}$ are closed in multi-dimensional linear space $\text{Ran}(\P)=V=\text{Span}\{u_i(0)\}_{i=1}^N$, but operators such as $\P\L\P,\P\L^2\P$ do not commute in $V$ due to the non-commutativity of matrices. As a result, we will get a closed noncommutative CMZE. A specific example for this case will be considered in Section \ref{sec:app_mod_coup}.   
%
%

%
%
\subsection{Time-dependent CMZE}
\label{sec:Time-d-CMZE}
The derivation of the combinatorial MZE for time-dependent systems is more difficult due to the time-dependence of the infinitesimal generator $\L(t)$ and $\Q\L(t)$. Assuming these two operators are also time differentiable up to any order. Using the evolution operator decomposition $\U_{\Q}(t,s)=\U_{\Q}(t,0)\U_{\Q}(0,s)$, we propose to use the following ansatz to get the series expansion of the memory integral:
\begin{align}\label{CMZE-type2_op}
    \P\L(s)\U_{\Q}(t,s)\Q\L(t)\P&=\P\L(s)\U_{\Q}(t,0)\U_{\Q}(0,s)\Q\L(t)\P\nonumber\\
    &=\sum_{n=0}^{\infty}\sum_{m=0}^{\infty}\frac{1}{n!m!}(\P\U(t,0)\P-\P)^n\F_n(s)\G_m(t)(\P\U(s,0)\P-\P)^m.
\end{align}
In order to determine the form of the time-dependent operators $\F_n(s)$ and $\G_m(t)$, we need to introduce {\em several} different homomorphisms and a {\em set} of algebraic equations of words \eqref{alge_eqn_word}. To this end, we use $\Oo_1,\Oo_2,\Oo_3,\Oo_4$ to represent the associative algebras for different (infinite) sets of operators. $\text{Hom}(\Dd,\mathbb{O}_1)$, $\text{Hom}(\Dd,\mathbb{O}_2)$,$\text{Hom}(\Dd,\mathbb{O}_3)$ are homomorphisms induced by different translation of the same alphabet $\Aa$. $\text{Hom}(\Dda,\mathbb{O}_4)$ is a homomorphisms induced by the translation of the alphabet $\Aa'$. The specific translation rule of these four homomorphisms are:
\begin{equation}\label{homomorphism}
\begin{aligned}
    \text{Hom}(\Dd,\mathbb{O}_1):&\quad\{a_1,a_2,\cdots\}\rightarrow \{\L_1,\L_2,\cdots\}\\
    \text{Hom}(\Dd,\mathbb{O}_2):&\quad\{a_1,a_2,\cdots\}\rightarrow \{\Q\L_1,\Q\L_2,\cdots\}\\
    \text{Hom}(\Dd,\mathbb{O}_3):&\quad\{a_1,a_2,\cdots\}\rightarrow \{-\Q\L_1,-\Q\L_2,\cdots\}\\
    \text{Hom}(\Dda,\mathbb{O}_4):&\quad\{a'_1,a'_2,\cdots\}\rightarrow \{\P\tB_1,\P\tB_2,\cdots\}.
\end{aligned}
\end{equation}
Here  $\L_n=\frac{d^{n-1}}{dt^{n-1}}\L(t)|_{t=0}$ and $\Q\L_n=\frac{d^{n-1}}{dt^{n-1}}\Q\L(t)|_{t=0}$.  $\P\tB_n=\P\tB_n(\L_1,\L_2,\cdots,\L_n)$, i.e. The $n$-th order Type-I noncommutative Bell polynomial over alphabet $\{\L_1,\L_2,\cdots\}$ becomes a new alphabet $\{\P\tB_1,\P\tB_2,\cdots\}$. With these homomorphisms, we will be able to get a (decoupled) system of algebraic equations for $\F_n(s)$ and $\G_m(t)$, which can be solved recursively. The result is summarized in the following theorem.
\begin{theorem}
\label{thm_combi_expansion2}
(Time-dependent CMZE) 
Suppose the identity operator $\I$ in Banach space $X$ can be orthogonally partitioned by two projection operators $\P,\Q$ as $\I=\P+\Q$ with $\P\Q=0$. If the closure of the time-dependent operators $\L(\tau):X\rightarrow X$ and $\Q\L(\tau):X\rightarrow X$ generate two evolution operators $\U(t,s)$ and $\U_{\Q}(t,s)$ respectively, moreover, the operator $\P\L_1\P$ is invertible in a subspace $V\subset X$, then we have the following combinatorial expansion of the time-dependent differential-form Dyson's identity \eqref{eqn:time_d_MZE_POPP}:
\begin{equation}\label{op_id_1_nonc_type2}
\begin{aligned}
\frac{d}{dt}\P&\U(t,0)\P
=\P \U(t,0)\P\L(t)\P\\
&
+\sum_{n=0}^{\infty}\sum_{m=0}^{\infty}\sum_{k_1=0}^{n}\sum_{k_2=0}^{m}
\frac{(-1)^{n+m-k_1-k_2}}{(n-k_1)!(m-k_1)!k_1!k_2!}\int_0^t\P \U(s,0)(\P\U(t,0)\P)^{k_1}\F_n(s)
\G_m(t)(\P\U(s,0)\P)^{k_2}\P ds,
\end{aligned}
\end{equation}
where $\F_n(s), \G_m(t)$ are time-independent polynomials of operators. Further denoting $(\P\L_1\P|_{V})^{-1}$ as the inverse of operator $\P\L_1\P$ in $V\subset X$. If $\text{Ran}(\P)\subset V$, then there exists explicit expressions for the operator pairs $\F_n(s), \G_m(t)$ which are given by the following recursive operator algebraic equation in $V$:
\begin{equation}\label{eqn:rec_op_nonc_type22}
\begin{aligned}
\begin{dcases}
\F_0(s)&=\P\L(s)\\
\F_n(s)&=-(\P\L_1\P|_{V})^{-n}\P\L(s)\tB_{n}(\Q\L_1,\Q\L_2,\cdots,\Q\L_{n})\\
&\ \ \ \ \ \ +(\P\L_1\P|_{V})^{-n}\sum_{k=1}^{n-1}\hB_{n,k}(\P\tB_{1}\P,\P\tB_{2}\P,\cdots,\P\tB_{n-k}\P)\F_k(s),\qquad n\geq 1
\end{dcases}
\end{aligned}
\end{equation}
\begin{equation}\label{eqn:rec_op_nonc_type21}
\begin{aligned}
\begin{dcases}
\G_0(s)&=\Q\L(s),\\
\G_m(t)&=-\tB_{m}(-\Q\L_1,-\Q\L_2,\cdots,-\Q\L_{m})\Q\L(t)(\P\L_1\P|_{V})^{-m}\\
&\ \ \ \ \ +\sum_{k=1}^{m-1} \G_k(t)\hB_{m,k}(\P\tB_{1}\P,\P\tB_{2}\P,\cdots,\P\tB_{m-k}\P)(\P\L_1\P|_{V})^{-m},\qquad m\geq 1
\end{dcases}
\end{aligned}
\end{equation}
\end{theorem}
\begin{proof}
Assuming the validity of the ansatz \eqref{CMZE-type2_op}, we further assume this series expansion can be decoupled and we have
\begin{align}
\P\L(s)\U_{\Q}(t,0)&=\sum_{n=0}^{\infty}\frac{1}{n!}(\P\U(t,0)\P-\P)^n\F_n(s) \label{eqn:id1}\\ 
\U_{\Q}(0,s)\Q\L(t)\P&=\left[\sum_{m=0}^{\infty}\frac{1}{m!}\G_m(t)(\P\U(s,0)\P-\P)^m\right]\P.
\label{eqn:id22}
\end{align}
The decoupling assumption is {\em not} necessarily imposed, which means that possible choices of $\F(s)$, $\G(t)$ such that \eqref{CMZE-type2_op} holds may not be unique. Here we only need to confirm the existence by explicitly constructing an operator series pair $\{\F_n(s),\G_m(t)\}_{n,m=0}^{\infty}$ that satisfies \eqref{CMZE-type2_op}. The decoupling assumption enables us to treat two expansions separately. We first consider the Taylor series expansion for $\P\L(s)\U_{\Q}(t,0)$ at $t=0$. Since $\U_{\Q}(t,0)$ is generated by the time-dependent operator $\Q\L(t)$, its Taylor series expansion is given by the Lie-Butcher series \eqref{exp_map_lie} using the Type-I noncommutative Bell polynomials:
\begin{align}\label{deri02}
\P\L(s)\U_{\Q}(t,0)=\sum_{n=0}^{\infty}\frac{t^n}{n!}\P\L(s)\tB_{n}(\Q\L_1,\Q\L_2,\cdots,\Q\L_{n}).
\end{align}
This is because semigroup $\U_{\Q}(t,s)$ acts on the observable function $f$ as a Koopman operator: $f(x_{\Q}(t))=f\circ(x_{\Q}(t))=[\U_{\Q}(t,0)f](x_{\Q}(0))$, which essentially is a pullback map $\Phi_{t,x_{\Q}(t)}^*f=f\circ\Phi_{t,x_{\Q}(t)}$. Here $x_{\Q}(t)$ is the flow generated by the operator $\Q\L(t)$. On the other hand, since the dynamics we considered is in flat manifolds such as the Euclidean space $\R^n$, the Lie-derivative in \eqref{exp_map_lie} degenerates to the directional derivative, which yields the final expression \eqref{deri02}. For the same reason, the Taylor series expansion for the operator $\P\U(t,0)\P$ at $t=0$ can also be written as:
\begin{align}\label{PUP_Taylor}
\P\U(t,0)\P=\sum_{n=0}^{\infty}\frac{t^n}{n!}\P\tB_{n}(\L_1,\L_2,\cdots,\L_{n})\P.
\end{align}
As a result, the RHS term $(\P\U(t,0)\P-\P)^n$ in \eqref{eqn:id1} can be calculated as:
\begin{align*}
    (\P\U(t,0)\P-\P)^n&=\left(\sum_{j_1=1}^{\infty}\frac{t^{j_1}}{j_1!}\P\tB_{j_1}\P\right)\left(\sum_{j_2=1}^{\infty}\frac{t^{j_2}}{j_2!}\P\tB_{j_2}\P\right)
    \cdots\left(\sum_{j_n=1}^{\infty}\frac{t^{j_n}}{j_n!}\P\tB_{j_n}\P\right)\\
    &=\sum_{m=n}^{\infty}\frac{t^m}{m!}\sum_{\substack{j_1+j_1+\cdots j_n=m,\\ j_i\geq 1}}\frac{m!}{j_1!j_2!\cdots j_k!}(\P\tB_{j_1}\P)(\P\tB_{j_2}\P)\cdots(\P\tB_{j_n}\P)
\\
&=\sum_{m=n}^{\infty}\frac{t^m}{m!}n!\hB_{m,n}(\P\tB_{1}\P,\cdots,\P\tB_{m+1-n}\P)
\end{align*}
where $\P\tB_n\P=\P\tB_n(\L_1,\L_2,\cdots,\L_n)\P$. Then we have:
\begin{align}
\sum_{n=0}^{\infty}\frac{1}{n!}(\P \U(t,0)\P-\P)^n\F_n(s)
&=\sum_{n=0}^{\infty}\frac{1}{n!}\sum_{m=n}^{\infty}
\frac{t^m}{m!}n!\hB_{m,n}(\P\tB_{1}\P,\cdots,\P\tB_{m+1-n}\P)\F_n(s)
\nonumber\\
&=
\sum_{m=0}^{\infty}\frac{t^m}{m!}\sum_{n=0}^m\hB_{m,n}(\P\tB_{1}\P,\cdots,\P\tB_{m+1-n}\P)\F_n(s)
\label{type3_RHS}
\end{align}
To be noticed that the $\F_k(s)$ is placed at the {\em RHS} of $\hB_{m,n}$ which is similar to expansion \eqref{type1_variant_expan}. Matching the coefficients in different order of $t^n$ on the RHS of \eqref{deri02} and \eqref{type3_RHS}, we can get the explicit expression of the operator $\{\F_n(s)\}_{n=0}^{\infty}$ using recurrence relations. Specifically, it is easy to verify that $\F_0(s)=\P\L(s)$, other $\F_n(s)$s can be constructed recursively by solving the following equation:
\begin{equation}\label{47}
\begin{aligned}
\P\L(s)\tB_{n}(\Q\L_1,\Q\L_2,\cdots,\Q\L_{n})
   =\sum_{k=1}^n\hB_{n,k}(\P\tB_{1}\P,\P\tB_{2}\P,\cdots,\P\tB_{n-k+1}\P)\F_k(s),\qquad n\geq 1.
\end{aligned}
\end{equation}
For $n\geq 1$, we have 
\begin{equation}
\begin{aligned}
\P\L(s)\tB_{n}(\Q\L_1,\Q\L_2,\cdots,\Q\L_{n})
   &=\sum_{k=1}^n\hB_{n,k}(\P\tB_{1}\P,\P\tB_{2}\P,\cdots,\P\tB_{n-k+1}\P)\F_k(s)\\
   &=\sum_{k=1}^{n-1} \hB_{n,k}(\P\tB_{1}\P,\P\tB_{2}\P,\cdots,\P\tB_{n-k}\P)\F_k(s)+ (\P\L_1\P)^n\F_n(s).
\end{aligned}
\end{equation}
If $\P\L_1\P|_V^{-1}$ exists , we can construct $\F_n(s)$ recursively as:
\begin{equation}
\begin{aligned}
\F_n(s)=(\P\L_1\P|_{V})^{-n}\P\L(s)\tB_{n}&(-\Q\L_1,-\Q\L_2,\cdots,-\Q\L_{n})\\
&\ +\sum_{k=1}^{n-1}(\P\L_1\P|_{V})^{-n}\hB_{n,k}(\P\tB_{1}\P,\P\tB_{2}\P,\cdots,\P\hB_{n-k}\P)\F_k(s),\qquad n\geq 1
\end{aligned}
\end{equation}
To be noticed that the above formula is valid since we have  $\P\L(s)\tB_{n}(-\Q\L_1,-\Q\L_2,\cdots,-\Q\L_{n})\in \text{Ran}(\P)$ and $\hB_{n-1,k}(\P\tB_{1}\P,\P\tB_{2}\P,\cdots,\P\tB_{n-k}\P)\F_k(s)\in \text{Ran}(\P)$. These two relations guarantee that operator $(\P\L_1\P|_{V})^{-n}$ is placed on $\text{Ran}(\P)\subset V$. The first few $\F_n(s)$s are:
\begin{equation}\label{first_few_F_s}
\begin{aligned}
\F_0(s)&=\P\L(s)\\
\F_1(s)&=(\P\L_1\P|_V)^{-1}\P\L(s)\Q\L_1\\
\F_2(s)&=(\P\L_1\P|_V)^{-2}\P\L(s)[(\Q\L_1)^2+\Q\L_2]+(\P\L_1\P|_V)^{-2}(\P\L_1^2\P+\P\L_2\P)(\P\L_1\P|_V)^{-1}\P\L(s)\Q\L_1.
\end{aligned}
\end{equation}
On the other hand, it is easy to verify that the Taylor series expansion of $\U_{\Q}(0,s)$ at $s=0$ is given by:
\begin{align}\label{deri01}
\U_{\Q}(0,s)\Q\L(t)=\sum_{m=0}^{\infty}\frac{s^m}{m!}\tB_{m}(-\Q\L_1,-\Q\L_2,\cdots,-\Q\L_{m})\Q\L(t).
\end{align}
Similarly, for the second expansion \eqref{eqn:id22}, we find that:
\begin{align}
\sum_{m=0}^{\infty}\frac{1}{m!}\G_m(t)(\P \U(s,0)\P-\P)^m
&=\sum_{m=0}^{\infty}\frac{1}{m!}\sum_{k=m}^{\infty}
\frac{s^k}{k!}m!\G_m(t)\hB_{m,k}(\P\tB_{1}\P,\cdots,\P\tB_{m+1-k}\P)
\nonumber\\
&=
\sum_{k=0}^{\infty}\frac{s^k}{k!}\sum_{m=0}^k\G_m(t)\hB_{m,k}(\P\tB_{1}\P,\cdots,\P\tB_{m+1-k}\P).
\label{type3_RHS_2}
\end{align}
Now we match the coefficient in different order of $s^m$ on the RHS of \eqref{deri01} and \eqref{type3_RHS_2}. When $n=0$, we have $\G_0(t)=\Q\L(t)$. For any $n\geq1$, we have 
\begin{equation}\label{52}
\begin{aligned}
\tB_{m}(-\Q\L_1,-\Q\L_2,\cdots,-\Q\L_{m})\Q\L(t)
   &=\sum_{k=1}^m\G_k(t)\hB_{m,k}(\P\tB_{1}\P,\P\tB_{2}\P,\cdots,\P\tB_{m-k+1}\P)\\
   &=\sum_{k=1}^{m-1} \G_k(t)\hB_{m,k}(\P\tB_{1}\P,\P\tB_{2}\P,\cdots,\P\tB_{m-k}\P)+\G_m(t)(\P\L_1\P)^m.
\end{aligned}
\end{equation}
If $\P\L_1\P|_V^{-1}$ exists, we can construct $\G_m(t)$ recursively as
\begin{equation}\label{G_M_recursive}
\begin{aligned}
\G_m(t)=\tB_{m}&(-\Q\L_1,-\Q\L_2,\cdots,-\Q\L_{m})\Q\L(t)(\P\L_1\P|_{V})^{-m}\\
&\ +\sum_{k=1}^{m-1} \G_k(t)\hB_{m,k}(\P\tB_{1}\P,\P\tB_{2}\P,\cdots,\P\tB_{m-k}\P)(\P\L_1\P|_{V})^{-m},\qquad m\geq 1
\end{aligned}
\end{equation}
The first few $\G_m(t)$ are:
\begin{equation}\label{first_few_G_t}
\begin{aligned}
\G_0(t)&=\Q\L(t)\\
\G_1(t)&=-\Q\L_1\Q\L(t)(\P\L_1\P|_V)^{-1}\\
\G_2(t)&=[(\Q\L_1)^2-\Q\L_2]\Q\L(t)(\P\L_1\P|_V)^{-2}+\Q\L_1\Q\L(t)(\P\L_1\P|_V)^{-1}(\P\L_1^2\P+\P\L_2\P)(\P\L_1\P|_V)^{-2}.
\end{aligned}
\end{equation}
Note that in \eqref{eqn:id22}, $\G_m(t)$ always acts after $\P$, which guarantees the invertibility of $\P\L\P$ for the first term on the RHS of \eqref{G_M_recursive}. In addition to this,
$\tB_{m-1,k}(\P\tB_{1}\P,\P\tB_{2}\P,\cdots,\P\tB_{m-k}\P)(\P\L_1\P|_{V})^{-m}\in\text{Ran}(\P)$. This ensures the invertibility of $\P\L\P$ for the second term on the RHS of \eqref{G_M_recursive}. In conclusion, we confirmed the existence of $\F_n(s)$ and $\G_m(t)$ such that \eqref{CMZE-type2_op} holds. In the derivation, we used four homomorphisms \eqref{homomorphism} and two algebraic equations of words (after the homomorphism): \eqref{deri02}=\eqref{type3_RHS} and \eqref{deri01}=\eqref{type3_RHS_2}.
Using the fact that $P$ commutes with $\P\U(t,0)\P$ and $\P\U(s,0)\P$, we simply use the binomial theorem to expand \eqref{CMZE-type2_op} and get 
\begin{align}\label{type2_LHS}
\P\L(s)\U_{\Q}(t,s)\Q\L(t)\P=\sum_{n=0}^{\infty}\sum_{m=0}^{\infty}\frac{(-1)^{n+m-k_1-k_2}}{(n-k_1)!(m-k_2)!k_1!k_2!}(\P\U(t,0)\P)^{k_1}\F_n(s)
\G_m(t)(\P\U(s,0)\P)^{k_2}\P.
\end{align}
Substituting expansion \eqref{type2_LHS} into the time-dependent differential-form Dyson's identity \eqref{eqn:time_d_MZE_POPP}, we obtain \eqref{op_id_1_nonc_type2}.
\end{proof}
The combinatorial expansion \eqref{CMZE-type2_op} for the time-dependent system memory integrand naturally implies that for the time-independent case, we also have:
\begin{align*}
    \P\L e^{(t-s)\Q\L}\Q\L\P=\P\L e^{-s\Q\L}e^{t\Q\L}\Q\L\P=
    \sum_{n=0}^{\infty}\sum_{m=0}^{\infty}\frac{1}{n!m!}(\P e^{t\L}\P-\P)^n \F_n\G_m(\P e^{s\L}\P-\P)^m.
\end{align*}
However, this does not bring us any obvious computational benefit. For a time-independent system, using the change of variable $s=t-s$, it is more convenient to develop the combinatorial expansion for $\P\L e^{s\Q\L}\Q\L$. This is exactly what we did in the previous section. On the other hand, we particularly emphasize that the coefficient operators $\F_n(s),\G_m(t)$ in ansatz \eqref{CMZE-type2_op} {\em has to be} placed in that order in order to guarantee the existence of $\F_n(s),\G_m(t)$. This means that other ansatz settings say the following one:
\begin{equation}\label{type_2_alter_ansatz11}
\begin{aligned}
    \P\L(s)\U_{\Q}(t,s)\Q\L(t)\P&=\P\L(s)\U_{\Q}(t,0)\U_{\Q}(0,s)\Q\L(t)\P\\
    &=\sum_{n=0}^{\infty}\sum_{m=0}^{\infty}\frac{1}{n!m!}(\P\U(t,0)\P-\P)^n\F_n(s)(\P\U(s,0)\P-\P)^m\G_m(t)
\end{aligned}
\end{equation}
would not work. For example, if assuming \eqref{type_2_alter_ansatz11}, when $n=1$, Eqn \eqref{52} will become:
\begin{align*}
-\Q\L_1\Q\L(t)=(\P\L_1\P)\G_m(t),
\end{align*}
which has no solution whenever $-\Q\L_1\Q\L(t)\neq 0$. This is because $-\Q\L_1\Q\L(t): X\rightarrow \text{Ran}(\Q)$ and $(\P\L_1\P)\G_m(t):X \rightarrow \text{Ran}(\P)$, and $\text{Ran}(\P)\perp\text{Ran}(\Q)$ due to the orthogonality $\P\Q=0$ between projection operators. Other types of rearrangement of $\F_n(s),\G_m(t)$ would fail for the same reason. 

%
%
\paragraph{Case study} As an example, we again consider the simplest case when the projection operator is the one-dimensional Mori's projection: $\P=\langle\cdot,u(0)\rangle u(0)$, where $u(0)\in V\subset H$ is an observable in a Hilbert space $H$ with $\|u(0)\|^2=1$. According to this  definition, it is easy to verify that $\text{Ran}(\P)=V=\R u(0)$ is an invariant subspace of the operator $\P\L_1\P$, i.e. $\P\L_1\P:V\rightarrow V$. $\P\L_1\P$ can be invertible in $V$ since it merely is a linear scaling operator $\P\L_1\P u(0)=au(0)$, hence invertible as long as $a=\langle \L_1u(0),u(0)\rangle\neq 0$. Now using Theorem \ref{thm_combi_expansion2}, the memory integral in Eqn \eqref{op_id_1_nonc_type2} can be peeled like an onion from right to left. Acting the operator identity \eqref{op_id_1_nonc_type2} on observable $u(0)$, it is easy to verify that 
\begin{align*}
    (\P\U(s,0)\P)^{k_2}\P u(0)=\langle \U(s,0)u(0),u(0)\rangle(\P\U(s,0))^{k_2-1} u(0)=\langle \U(s,0)u(0),u(0)\rangle^{k_2} u(0)=C^{k_2}(s)u(0),
\end{align*}
where $C(s)$ is a scalar function of $s$. When $n=0$, $\F_0(s)=\P\L(s)$, we have 
\begin{align*}
    \F_0(s)\G_m(t)u(0)=\langle\L(s)\G_m(t)u(0),u(0)\rangle u(0)=k_{0,m}(t,s)u(0)
\end{align*}
where $k_{0,m}(t,s)$ is a scalar function of time $t,s$. For $n\geq 1$, we have 
\begin{align*}
    (\P\U(t,0)\P)^{k_1}\F_n(s)\G_m(t)u(0)
    &=(\P\U(t,0))^{k_1}\P\F_n(s)\G_m(t)u(0)\\
    &=\langle\F_n(s)\G_m(t)u(0),u(0)\rangle(\P\U(t,0))^{k_1}u(0)\\
    &=\langle\F_n(s)\G_m(t)u(0),u(0)\rangle\langle\U(t,0)u(0),u(0)\rangle^{k_1}u(0)\\
    &=k_{n,m}(t,s)C^{k_1}(t)u(0).
\end{align*}
Lastly, we have 
\begin{align*}
   \P\U(s,0)u(0)=C(s)u(0).
\end{align*}
Combining all these terms, we can see that when applying \eqref{op_id_1_nonc_type2} to $u(0)$, it simplifies to the following nonlinear {\em self-consistent} integro-differential equation for $C(t)$:
\begin{align*}
\frac{d}{dt}C(t)
&=\Omega(t)C(t)+\sum_{n=0}^{\infty}\sum_{m=0}^{\infty}\sum_{k_1=0}^{n}\sum_{k_2=0}^{m}
\frac{(-1)^{n+m-k_1-k_2}}{(n-k_1)!(m-k_1)!k_1!k_2!}C^{k_1}(t)\int_{0}^tk_{n,m}(t,s)C^{k_2+1}(s)ds.
\end{align*}
where $\Omega(t)=\langle \L(t)u(0),u(0)\rangle$. The first few coefficient functions $k_{n,m}(t,s)$ are explicitly calculated and summarized in Table \ref{Tab:1}.
\paragraph{Resummation of the combinatorial expansion} We can also consider the resummation of the combinatorial expansion by formally replacing the power series sequence with other polynomials. In this paper, we do not provide detailed explorations in this regard, and only show how the resummation can be done for the commutative CMZE \eqref{eqn:Com_MEZ_time}. For this case, we have combinatorial expansion $K(t)=f(\hat C(t))=\sum_kf_k/k!\hat C^k(t)$. We can use families of orthogonal polynomials to reorganize terms in the Taylor series expansion as $K(t)=f(\hat C(t))=\sum_kf_k/k!\hat C^k(t)=\sum_nw_n(f_1,\cdots,f_n)\phi_n(\hat C(t))$, where $\phi_n(x)$ is a $n$-th order orthogonal polynomial of $x$, and $w_n$ is the corresponding expansion coefficient. This leads to resummed CMZE:
\begin{align}\label{eqn:Com_MEZ_resum1}
\frac{d}{dt}C(t)=\Omega C(t)+\sum_{n}w_n(f_1,\cdots,f_n)\int_0^t\phi_n(\hat C(s))C(t-s)ds.
\end{align}
Since the normalized correlation function $\hat C(t)\in [-1,1]$, Suitable choices of $\phi_n(x)$ can be Chebyshev polynomials, Legendre polynomials or any kind of Jacobi polynomials. Pad\'e resummation of the order $[m/n]$ can also be used, correspondingly we have 
\begin{align}\label{eqn:Com_MEZ_resum2}
\frac{d}{dt}C(t)=\Omega C(t)+\int_0^t\frac{P_m(\hat C(s))}{Q_n(\hat C(s))}C(t-s)ds,
\end{align}
where $P_n(\hat C(t))=\sum_{k=0}^ma_k(f_1,\cdots,f_k)\hat C^k(t)$ and $Q_n(\hat C(t))=1+\sum_{k=1}^nb_k(f_1,\cdots,f_k)\hat C^k(t)$.
\section{Tree representation of the combinatorial Mori-Zwanzig equation}
\label{sec:NCC-MZE_tree}
From a mathematical point of view, the word representation is sufficient for proving the existence of the combinatorial series expansion and deriving the explicit expression of the CMZE. In this section, we introduce the planar tree representation for the noncommutative Bell and bipartition polynomials and use them to derive an equivalent representation of the CMZE using (decorated) trees. The tree representation provides interesting connections between the CMZE and the combinatorial Dyson-Schwinger equation (DSE) \cite{kreimer2006etude} that was recently discussed heavily in combinatorial quantum field theory \cite{foissy2008faa,yeats2008growth,yeats2017combinatorial}. Applying the CMZE to quantum many-body problems, we further point out the connection between the tree diagram and the Feynman diagram, and clearly demonstrate that CMZE leads to self-consistent expansions of the self-energy (under the MZ framework) which is comparable with the skeleton expansion of the self-energy in renormalized Dyson's equation \cite{stefanucci2013nonequilibrium}. This is the connection we emphasized in Section \ref{sec:into_main_result}.
\subsection{Tree representation of the non-commutative Bell and bipartition polynomials}
\label{sec:tree_intro}
A planar tree is an undirected graph of nodes (or vertices) in which any two vertices are connected by exactly one path and can be embedded in a plane. Planar trees can be graphically represented by the dot plots as shown in \eqref{Tree_graphs}. The planar tree we defined can grow upwards by connecting existing nodes in the tree with new nodes. Specifically, the tree can grow in two ways: I. By connecting apical nodes with a new node:
{\fontsize{1}{1}\selectfont
\begin{forest}
 for tree={grow=90, l=1mm}
[]
 \path[fill=black]  (!.parent anchor) circle[radius=2pt];
\end{forest}
}
$
\rightarrow
$
{\fontsize{3}{1}\selectfont
\begin{forest}
 for tree={grow=90, l=1mm}
[[]]
 \path[fill=black]  (!.parent anchor) circle[radius=2pt];
 \path[fill=black] (!1.child anchor) circle[radius=2pt];
\end{forest}
}
.
This will be called {\em germinate}.
II. By adding a new branch to non-apical nodes:
{\fontsize{3}{1}\selectfont
\begin{forest}
 for tree={grow=90, l=1mm}
[[]]
 \path[fill=black]  (!.parent anchor) circle[radius=2pt];
 \path[fill=black] (!1.child anchor) circle[radius=2pt];
\end{forest}
}
$
\rightarrow
$
{\fontsize{3}{1}\selectfont
\begin{forest}
 for tree={grow=90, l=1mm}
[[][]]
 \path[fill=black]  (!.parent anchor) circle[radius=2pt];
 \path[fill=black] (!1.child anchor) circle[radius=2pt]
                   (!2.child anchor) circle[radius=2pt];
\end{forest}
}.
This will be called {\em fork}. Planar trees with $n+1$ nodes is denoted by $T_n$. If it has height $k$, i.e. the highest depth of all branches, it will be further denoted as $T_{n,k}$. Gathering different trees together we get a forest $F$. With a slight abuse of notation, we use the same $F$ to represent the set $F$ and the summation of trees within $F$, such as the expressions in \eqref{Tree_graphs}. A forest that collects planar trees with $n+1$ nodes will be denoted as $F_n$. Any $F_n$, $n\neq 0$ can be decomposed as $F_n=\sum_{k=1}^nF_{n,k}$, where $F_{n,k}$ is the subset of $F_n$ containing trees with height $k$. Here the summation sign $\sum$ means the set operation $\cup$ or the regular summation, depending on the context. We now define a tree (forest) sprouting operator $\S^+:T_n\rightarrow T_{n+1}$ as follows:
\begin{equation}
\begin{aligned}
    \S^+T_n=\sum_{i=0}^{\#f}T_{n+1}^i,\qquad \text{where} \quad
    \begin{dcases}
    T_{n+1}^0 \text{is a new tree when $T_n$'s rightmost apical nodes germinates}\\
    T_{n+1}^i \text{is a new tree when $T_n$'s $i$-th non-apical node forks, $1\leq i\leq\# f$}
    \end{dcases}
\end{aligned},
\end{equation}
where $\#f$ denotes the total number of non-apical nodes in $T_n$. As an example, we have 
\begin{align*}
    \S^+
    {\fontsize{4}{1}\selectfont
\begin{forest}
 for tree={grow=90, l=1mm}
[[[]][[][]]]
\path[fill=black]  (!.parent anchor) circle[radius=2pt];
\path[fill=black]  (!1.child anchor) circle[radius=2pt];
\path[draw]    (!11.child anchor) circle[radius=2pt];
\path[fill=black]  (!2.child anchor) circle[radius=2pt];
\path[draw]   (!21.child anchor) circle[radius=2pt];
\path[draw]   (!22.child anchor) circle[radius=2pt];
\end{forest}
}
=
{\fontsize{4}{1}\selectfont
\begin{forest}
 for tree={grow=90, l=1mm}
[ [[[]]] [[][]]]
\path[fill=black]  (!.parent anchor) circle[radius=2pt];
\path[fill=black]  (!1.child anchor) circle[radius=2pt];
\path[fill=black]  (!11.child anchor) circle[radius=2pt];
\path[draw]  (!111.child anchor) circle[radius=2pt];
\path[fill=black]   (!2.child anchor) circle[radius=2pt];
\path[draw]   (!21.child anchor) circle[radius=2pt];
\path[draw]   (!22.child anchor) circle[radius=2pt];
\end{forest}
}
+
{\fontsize{4}{1}\selectfont
\begin{forest}
 for tree={grow=90, l=1mm}
[[[][]][[][]]]
\path[fill=black]  (!.parent anchor) circle[radius=2pt];
\path[fill=black]   (!1.child anchor) circle[radius=2pt];
\path[draw] (!11.child anchor) circle[radius=2pt];
\path[draw]  (!12.child anchor) circle[radius=2pt];
\path[fill=black]  (!2.child anchor) circle[radius=2pt];
\path[draw]  (!21.child anchor) circle[radius=2pt];
\path[draw] (!22.child anchor) circle[radius=2pt];
\end{forest}
}
+
{\fontsize{4}{1}\selectfont
\begin{forest}
 for tree={grow=90, l=1mm}
[[[]][[][][]]]
 \path[fill=black]  (!.parent anchor) circle[radius=2pt];
 \path[fill=black] (!1.child anchor) circle[radius=2pt];
 \path[draw]  (!11.child anchor) circle[radius=2pt];
\path[fill=black]   (!2.child anchor) circle[radius=2pt];
\path[draw]   (!21.child anchor) circle[radius=2pt];
\path[draw]   (!22.child anchor) circle[radius=2pt];
\path[draw]   (!23.child anchor) circle[radius=2pt];
\end{forest}
}
+
{\fontsize{4}{1}\selectfont
\begin{forest}
 for tree={grow=90, l=1mm}
[[[]][][[][]]]
 \path[fill=black]  (!.parent anchor) circle[radius=2pt];
 \path[fill=black]  (!1.child anchor) circle[radius=2pt];
 \path[draw]  (!11.child anchor) circle[radius=2pt];
 \path[draw]  (!2.child anchor) circle[radius=2pt];
 \path[fill=black] (!3.child anchor) circle[radius=2pt];
 \path[draw]  (!31.child anchor) circle[radius=2pt];
 \path[draw]  (!32.child anchor) circle[radius=2pt];
\end{forest}
},
\end{align*}
where ${\fontsize{1}{1}\selectfont
\begin{forest}
 for tree={grow=90, l=1mm}
[]
 \path[draw]  (!.parent anchor) circle[radius=2pt];
\end{forest}
}
$ represents the apical node and 
${\fontsize{1}{1}\selectfont
\begin{forest}
 for tree={parent anchor=east, child anchor=east, grow=90, l=1mm}
[]
 \path[fill=black]  (!.parent anchor) circle[radius=2pt];
\end{forest}
}$
represents the non-apical node. When applying $\S^+$ to a forest $F$, it applies to all trees within $F$. There are two ways to get a new forest which leads to two types of sprouting operators. The first one is still denoted as $\S^+$ and may be called the sprouting operator without counting. Basically $\S^+$ acts on $F$ as a set operation and returns non-repeatedly a new set of trees that form a forest. The second sprouting operator includes the counting factor and will be re-denoted as $\S^+_{C}$. It is easier to understand the action of $\S^+,\S^+_{C}$ using the following example:
\begin{align}
    \S^+\big(
{\fontsize{4}{1}\selectfont
\begin{forest}
 for tree={grow=90, l=1mm}
[[[]]]
 \path[fill=black]  (!.parent anchor) circle[radius=2pt];
 \path[fill=black] (!1.child anchor) circle[radius=2pt]
                   (!11.child anchor) circle[radius=2pt];
\end{forest}
}
+
{\fontsize{4}{1}\selectfont
\begin{forest}
for tree={grow=90, l=1mm}
[[][]]
\path[fill=black]  (!.parent anchor) circle[radius=2pt];
\path[fill=black] (!1.child anchor) circle[radius=2pt]
                 (!2.child anchor) circle[radius=2pt];
\end{forest}
}
\big)
&=
{\fontsize{4}{1}\selectfont
\begin{forest}
 for tree={grow=90, l=1mm}
[[[[]]]]
 \path[fill=black]  (!.parent anchor) circle[radius=2pt];
 \path[fill=black] (!1.child anchor) circle[radius=2pt]
                    (!11.child anchor) circle[radius=2pt]
                   (!111.child anchor) circle[radius=2pt];
\end{forest}
}
+
{\fontsize{4}{1}\selectfont
\begin{forest}
 for tree={grow=90, l=1mm}
[[[]][]]
 \path[fill=black]  (!.parent anchor) circle[radius=2pt];
 \path[fill=black] (!1.child anchor) circle[radius=2pt]
                  (!11.child anchor) circle[radius=2pt]
                  (!2.child anchor) circle[radius=2pt];
\end{forest}
}
+
{\fontsize{4}{1}\selectfont
\begin{forest}
 for tree={parent anchor=east, child anchor=east, grow=90, l=1mm}
[[[][]]]
 \path[fill=black]  (!.parent anchor) circle[radius=2pt];
 \path[fill=black] (!1.child anchor) circle[radius=2pt]
                   (!11.child anchor) circle[radius=2pt]
                   (!12.child anchor) circle[radius=2pt];
\end{forest}
}
+
{\fontsize{4}{1}\selectfont
\begin{forest}
 for tree={grow=90, l=1mm}
[[][][]]
 \path[fill=black]  (!.parent anchor) circle[radius=2pt];
 \path[fill=black] (!1.child anchor) circle[radius=2pt]
                   (!2.child anchor) circle[radius=2pt]
                   (!3.child anchor) circle[radius=2pt];
\end{forest}
}\\
\S^+_C\big(
{\fontsize{4}{1}\selectfont
\begin{forest}
 for tree={grow=90, l=1mm}
[[[]]]
 \path[fill=black]  (!.parent anchor) circle[radius=2pt];
 \path[fill=black] (!1.child anchor) circle[radius=2pt]
                   (!11.child anchor) circle[radius=2pt];
\end{forest}
}
+
{\fontsize{4}{1}\selectfont
\begin{forest}
for tree={grow=90, l=1mm}
[[][]]
\path[fill=black]  (!.parent anchor) circle[radius=2pt];
\path[fill=black] (!1.child anchor) circle[radius=2pt]
                 (!2.child anchor) circle[radius=2pt];
\end{forest}
}
\big)
&=
{\fontsize{4}{1}\selectfont
\begin{forest}
 for tree={grow=90, l=1mm}
[[[[]]]]
 \path[fill=black]  (!.parent anchor) circle[radius=2pt];
 \path[fill=black] (!1.child anchor) circle[radius=2pt]
                    (!11.child anchor) circle[radius=2pt]
                   (!111.child anchor) circle[radius=2pt];
\end{forest}
}
+2
{\fontsize{4}{1}\selectfont
\begin{forest}
 for tree={grow=90, l=1mm}
[[[]][]]
 \path[fill=black]  (!.parent anchor) circle[radius=2pt];
 \path[fill=black] (!1.child anchor) circle[radius=2pt]
                  (!11.child anchor) circle[radius=2pt]
                  (!2.child anchor) circle[radius=2pt];
\end{forest}
}
+
{\fontsize{4}{1}\selectfont
\begin{forest}
 for tree={parent anchor=east, child anchor=east, grow=90, l=1mm}
[[[][]]]
 \path[fill=black]  (!.parent anchor) circle[radius=2pt];
 \path[fill=black] (!1.child anchor) circle[radius=2pt]
                   (!11.child anchor) circle[radius=2pt]
                   (!12.child anchor) circle[radius=2pt];
\end{forest}
}
+
{\fontsize{4}{1}\selectfont
\begin{forest}
 for tree={grow=90, l=1mm}
[[][][]]
 \path[fill=black]  (!.parent anchor) circle[radius=2pt];
 \path[fill=black] (!1.child anchor) circle[radius=2pt]
                   (!2.child anchor) circle[radius=2pt]
                   (!3.child anchor) circle[radius=2pt];
\end{forest}
}.
\end{align}
Here the counting factor 2 appears since the sprouting operator generates two identical trees as shown above. Obviously, we have $\S^+,\S^+_C:F_n\rightarrow F_{n+1}$. With the sprouting operators, we can construct an infinite sequence of forests $\{F_n\}_{n=1}^{\infty}$ from a seed $T_1=F_1={\fontsize{1}{1}\selectfont
\begin{forest}
 for tree={parent anchor=east, child anchor=east, grow=90, l=1mm}
[]
 \path[fill=black]  (!.parent anchor) circle[radius=2pt];
\end{forest}
}$ as $F_{n+1}=\S^+(F_n)$ or $F_{n+1}=\S_C^+(F_n)$ for $n\geq 1$. The noncommutative Bell polynomials can be represented using such an infinite sequence of forests. In particular, it is generated by $\S^+_C$ as $F^{\tB}_{n+1}=\S_C^+(F^{\tB}_n)$, where the subscript $\tB$ indicates that this is for the Type-I noncommutative Bell polynomials. The first few forests can be explicitly expressed as:
\begin{equation}\label{Tree_graphs}
\begin{aligned}
F_0^{\tB}&=
{\fontsize{4}{1}\selectfont
\begin{forest}
 for tree={grow=90, l=1mm}
[]
 \path[fill=black]  (!.parent anchor) circle[radius=2pt];
\end{forest}
}
\\
F_1^{\tB}&=
{\fontsize{4}{1}\selectfont
\begin{forest}
 for tree={grow=90, l=1mm}
[[]]
 \path[fill=black]  (!.parent anchor) circle[radius=2pt];
 \path[fill=black] (!1.child anchor) circle[radius=2pt];
\end{forest}
}
\\
F_2^{\tB}&=
{\fontsize{4}{1}\selectfont
\begin{forest}
 for tree={grow=90, l=1mm}
[[[]]]
 \path[fill=black]  (!.parent anchor) circle[radius=2pt];
 \path[fill=black] (!1.child anchor) circle[radius=2pt]
                   (!11.child anchor) circle[radius=2pt];
\end{forest}
}
+
{\fontsize{4}{1}\selectfont
\begin{forest}
for tree={grow=90, l=1mm}
[[][]]
\path[fill=black]  (!.parent anchor) circle[radius=2pt];
\path[fill=black] (!1.child anchor) circle[radius=2pt]
                 (!2.child anchor) circle[radius=2pt];
\end{forest}
}
\\
F_3^{\tB}&=
{\fontsize{4}{1}\selectfont
\begin{forest}
 for tree={grow=90, l=1mm}
[[[[]]]]
 \path[fill=black]  (!.parent anchor) circle[radius=2pt];
 \path[fill=black] (!1.child anchor) circle[radius=2pt]
                    (!11.child anchor) circle[radius=2pt]
                   (!111.child anchor) circle[radius=2pt];
\end{forest}
}
+2
{\fontsize{4}{1}\selectfont
\begin{forest}
 for tree={grow=90, l=1mm}
[[[]][]]
 \path[fill=black]  (!.parent anchor) circle[radius=2pt];
 \path[fill=black] (!1.child anchor) circle[radius=2pt]
                  (!11.child anchor) circle[radius=2pt]
                  (!2.child anchor) circle[radius=2pt];
\end{forest}
}
+
{\fontsize{4}{1}\selectfont
\begin{forest}
 for tree={parent anchor=east, child anchor=east, grow=90, l=1mm}
[[[][]]]
 \path[fill=black]  (!.parent anchor) circle[radius=2pt];
 \path[fill=black] (!1.child anchor) circle[radius=2pt]
                   (!11.child anchor) circle[radius=2pt]
                   (!12.child anchor) circle[radius=2pt];
\end{forest}
}
+
{\fontsize{4}{1}\selectfont
\begin{forest}
 for tree={grow=90, l=1mm}
[[][][]]
 \path[fill=black]  (!.parent anchor) circle[radius=2pt];
 \path[fill=black] (!1.child anchor) circle[radius=2pt]
                   (!2.child anchor) circle[radius=2pt]
                   (!3.child anchor) circle[radius=2pt];
\end{forest}
}
\\
F_{4}^{\tB}&=
{\fontsize{4}{1}\selectfont
\begin{forest}
 for tree={grow=90, l=1mm}
 [[[[[]]]]]
 \path[fill=black]  (!.parent anchor) circle[radius=2pt];
 \path[fill=black] (!1.child anchor) circle[radius=2pt]
                  (!11.child anchor) circle[radius=2pt]
                  (!111.child anchor) circle[radius=2pt]
                  (!1111.child anchor)  circle[radius=2pt];
\end{forest}
}
+
{\fontsize{4}{1}\selectfont
\begin{forest}
 for tree={grow=90, l=1mm}
 [[[[][]]]]
 \path[fill=black]  (!.parent anchor) circle[radius=2pt];
 \path[fill=black] (!1.child anchor) circle[radius=2pt]
                  (!11.child anchor) circle[radius=2pt]
                  (!111.child anchor) circle[radius=2pt]
                  (!112.child anchor)  circle[radius=2pt];
\end{forest}
}
+2
{\fontsize{4}{1}\selectfont
\begin{forest}
 for tree={grow=90, l=1mm}
[[[[]][]]]
 \path[fill=black]  (!.parent anchor) circle[radius=2pt];
 \path[fill=black] (!1.child anchor) circle[radius=2pt]
                  (!11.child anchor) circle[radius=2pt]
                  (!111.child anchor) circle[radius=2pt]
                  (!12.child anchor) circle[radius=2pt];
\end{forest}
}
+3
{\fontsize{4}{1}\selectfont
\begin{forest}
 for tree={grow=90, l=1mm}
[[[[]]][]]
 \path[fill=black]  (!.parent anchor) circle[radius=2pt];
 \path[fill=black] (!1.child anchor) circle[radius=2pt]
                  (!11.child anchor) circle[radius=2pt]
                  (!111.child anchor) circle[radius=2pt]
                  (!2.child anchor) circle[radius=2pt];
\end{forest}
}
+3
{\fontsize{4}{1}\selectfont
\begin{forest}
 for tree={grow=90, l=1mm}
 [ [[][]] []]
 \path[fill=black]  (!.parent anchor) circle[radius=2pt];
 \path[fill=black] (!1.child anchor) circle[radius=2pt]
                  (!11.child anchor) circle[radius=2pt]
                  (!12.child anchor) circle[radius=2pt]
                  (!2.child anchor)  circle[radius=2pt];
\end{forest}
}
+3
{\fontsize{4}{1}\selectfont
\begin{forest}
 for tree={grow=90, l=1mm}
 [[[]] [] [] ]
 \path[fill=black]  (!.parent anchor) circle[radius=2pt];
 \path[fill=black] (!1.child anchor) circle[radius=2pt]
                  (!11.child anchor) circle[radius=2pt]
                  (!2.child anchor) circle[radius=2pt]
                  (!3.child anchor)  circle[radius=2pt];
\end{forest}
}
+
{\fontsize{4}{1}\selectfont
\begin{forest}
 for tree={grow=90, l=1mm}
 [ [[][][]] ]
 \path[fill=black]  (!.parent anchor) circle[radius=2pt];
 \path[fill=black] (!1.child anchor) circle[radius=2pt]
                  (!11.child anchor) circle[radius=2pt]
                  (!12.child anchor) circle[radius=2pt]
                  (!13.child anchor)  circle[radius=2pt];
\end{forest}
}
+
{\fontsize{4}{1}\selectfont
\begin{forest}
 for tree={grow=90, l=1mm}
 [ [][][] []]
 \path[fill=black]  (!.parent anchor) circle[radius=2pt];
 \path[fill=black] (!1.child anchor) circle[radius=2pt]
                  (!2.child anchor) circle[radius=2pt]
                  (!3.child anchor) circle[radius=2pt]
                  (!4.child anchor)  circle[radius=2pt];
\end{forest}
}
\\
&\ \ \vdots
\end{aligned}
\end{equation}
The tree representation of the noncommutative Bell polynomials is equivalent to its word representation. In fact, there is a one-to-one correspondence between the trees and the words. One simply needs to translate $F_n^{\tB}$ to words using the following dictionary:
\begin{equation}\label{trans_rule}
\begin{aligned}
    {\fontsize{4}{1}\selectfont
\begin{forest}
 for tree={grow=90, l=1mm}
[]
 \path[fill=black]  (!.parent anchor) circle[radius=2pt];
\end{forest}
}
&
\Rightarrow \I\\
{\fontsize{4}{1}\selectfont
\begin{forest}
 for tree={grow=90, l=1mm}
[[]]
 \path[fill=black]  (!.parent anchor) circle[radius=2pt]
                    (!1.child anchor) circle[radius=2pt];
\end{forest}
}
&\Rightarrow a_1,
\qquad
{\fontsize{4}{1}\selectfont
\begin{forest}
 for tree={grow=90, l=1mm}
[[][]]
 \path[fill=black]  (!.parent anchor) circle[radius=2pt]
                    (!1.child anchor) circle[radius=2pt]
                    (!2.child anchor) circle[radius=2pt];
\end{forest}
}
\Rightarrow a_2,\qquad
{\fontsize{4}{1}\selectfont
\begin{forest}
 for tree={grow=90, l=1mm}
[[][][]]
 \path[fill=black]  (!.parent anchor) circle[radius=2pt]
                    (!1.child anchor) circle[radius=2pt]
                    (!2.child anchor) circle[radius=2pt]
                    (!3.child anchor) circle[radius=2pt];
\end{forest}
}
\Rightarrow a_3,\qquad
{\fontsize{4}{1}\selectfont
\begin{forest}
 for tree={grow=90, l=1mm}
[[][][][]]
 \path[fill=black]  (!.parent anchor) circle[radius=2pt]
                    (!1.child anchor) circle[radius=2pt]
                    (!2.child anchor) circle[radius=2pt]
                    (!3.child anchor) circle[radius=2pt]
                    (!4.child anchor) circle[radius=2pt];
\end{forest}
}
\Rightarrow a_4,\qquad \cdots\\
{\fontsize{4}{1}\selectfont
\begin{forest}
 for tree={grow=90, l=1mm}
[[[]]]
 \path[fill=black]  (!.parent anchor) circle[radius=2pt]
                    (!1.child anchor) circle[radius=2pt]
                    (!11.child anchor) circle[radius=2pt];
\end{forest}
}
&\Rightarrow a_1^2,
\qquad
{\fontsize{4}{1}\selectfont
\begin{forest}
 for tree={grow=90, l=1mm}
[[[[]]]]
 \path[fill=black]  (!.parent anchor) circle[radius=2pt]
                    (!1.child anchor) circle[radius=2pt]
                    (!11.child anchor) circle[radius=2pt]
                    (!111.child anchor) circle[radius=2pt];
\end{forest}
}
\Rightarrow a_1^3,\qquad \cdots.
\end{aligned}
\end{equation}
For any tree $T_{n,k}\in F_n^{\tB}$, counting from the bottom branch $i=1$ to the top branch $i=k$, using the dictionary \eqref{trans_rule} to translate different branches into words and then putting them from right to left, then we get all the words in $\tB_n$. For instance, we have
\begin{align*}
    {\fontsize{4}{1}\selectfont
\begin{forest}
 for tree={grow=90, l=1mm}
 [[[[]]][]]
 \path[fill=black]  (!.parent anchor) circle[radius=2pt];
 \path[fill=black] (!1.child anchor) circle[radius=2pt]
                   (!11.child anchor) circle[radius=2pt]
                  (!111.child anchor) circle[radius=2pt]
                  (!2.child anchor) circle[radius=2pt];
\end{forest}
}\Rightarrow a_1^2a_2,\qquad
{\fontsize{4}{1}\selectfont
\begin{forest}
 for tree={grow=90, l=1mm}
 [ [[][]] []]
 \path[fill=black]  (!.parent anchor) circle[radius=2pt];
 \path[fill=black] (!1.child anchor) circle[radius=2pt]
                  (!11.child anchor) circle[radius=2pt]
                  (!12.child anchor) circle[radius=2pt]
                  (!2.child anchor)  circle[radius=2pt];
\end{forest}
}\Rightarrow a_2^2,
\qquad
{\fontsize{4}{1}\selectfont
\begin{forest}
 for tree={grow=90, l=1mm}
 [[[]] [] [] ]
 \path[fill=black]  (!.parent anchor) circle[radius=2pt];
 \path[fill=black] (!1.child anchor) circle[radius=2pt]
                  (!11.child anchor) circle[radius=2pt]
                  (!2.child anchor) circle[radius=2pt]
                  (!3.child anchor)  circle[radius=2pt];
\end{forest}
}
\Rightarrow a_1a_3.
\end{align*}
This translation rule directly stems from the fact that any word is spelled from left to right non-interchangeably. Equivalently, any tree $T_{n,k}$ can be decomposed as $b_1\lgraf b_2\lgraf\cdots\lgraf b_k=T_{n,k}$, where $b_j$ is the $j$-th branch and $\lgraf$ is the {\em right grafting product} which grafts the LHS branch on the top right of the RHS branch. To be noticed that this product $\lgraf$ is non-commutative but distributive, i.e. $a\lgraf b\neq b\lgraf a$ and $a\lgraf(b+c)=a\lgraf b+a\lgraf c$. Since the sprouting operator always germinates on the rightmost apical nodes, we will not encounter the order problem when translating branches in each layer. All the recurrence relations, decomposition, and combinatorial coefficients that appeared in Section \ref{sec:bell_poly} can find their tree correspondences as well. For instance: the sprouting operation $F^{\tB}_{n+1}=\S_C^+(F^{\tB}_n)$ is the equivalent of the $(a_1+\partial)$ operator, where the multiplication of $a_1$ to different words corresponds to the germinating part of $\S_C^+$ and the action of $\partial$ corresponds to the fork operation of $\S_C^+$. The noncommutative Bell polynomials can be decomposed by the word length as $\tB_n=\sum_{k=1}^n\tB_{n,k}$. Correspondingly, we can decompose a forest by the height $k$ as $F^{\tB}_n=\sum_{k=1}^nF^{\tB}_{n,k}$. For instance, $F_3^{\tB}$ can be decomposed as:
\begin{align*}
F^{\tB}_{3}=\sum_{k=1}^3F^{\tB}_{3,k}=
\underbrace{
{\fontsize{4}{1}\selectfont
\begin{forest}
 for tree={grow=90, l=1mm}
[[[[]]]]
 \path[fill=black]  (!.parent anchor) circle[radius=2pt];
 \path[fill=black] (!1.child anchor) circle[radius=2pt]
                    (!11.child anchor) circle[radius=2pt]
                   (!111.child anchor) circle[radius=2pt];
\end{forest}
}
}_{F_{3,3}}
+
\underbrace{
2
{\fontsize{4}{1}\selectfont
\begin{forest}
 for tree={grow=90, l=1mm}
[[[]][]]
 \path[fill=black]  (!.parent anchor) circle[radius=2pt];
 \path[fill=black] (!1.child anchor) circle[radius=2pt]
                  (!11.child anchor) circle[radius=2pt]
                  (!2.child anchor) circle[radius=2pt];
\end{forest}
}
+
{\fontsize{4}{1}\selectfont
\begin{forest}
 for tree={grow=90, l=1mm}
[[[][]]]
 \path[fill=black]  (!.parent anchor) circle[radius=2pt];
 \path[fill=black] (!1.child anchor) circle[radius=2pt]
                   (!11.child anchor) circle[radius=2pt]
                   (!12.child anchor) circle[radius=2pt];
\end{forest}
}
}_{F_{3,2}}
+
\underbrace{
{\fontsize{4}{1}\selectfont
\begin{forest}
 for tree={grow=90, l=1mm}
[[][][]]
 \path[fill=black]  (!.parent anchor) circle[radius=2pt];
 \path[fill=black] (!1.child anchor) circle[radius=2pt]
                   (!2.child anchor) circle[radius=2pt]
                   (!3.child anchor) circle[radius=2pt];
\end{forest}
}
}_{F_{3,1}}.
\end{align*}
The combinatorial coefficients in \eqref{rec3_non_Bell_poly} can be represented as:
\begin{align}\label{weight_B}
    \frac{n}{\|b_1\|!\cdots\|b_k\|!}
   \frac{\|b_1\|\cdots\|b_k\|}{\|b_1\|(\|b_1+b_2\|)\cdots(\|b_1+\cdots+b_k\|)}.
\end{align}
where $b_j$ is the $j$-th branch of a tree $T_{n,k}$, $\|b_k\|$ gives the number of children nodes of $b_k$. We can similarly define the tree representation for the Type-II Bells polynomials $\hB_n$ and the bipartition polynomials $\tP_n$. Comparing the explicit expression for $\tB_n$, $\hB_n$ and $\tP_n$, we find that these noncommutative polynomials consist of the same set of words and only differ in coefficients. Hence, we can define the tree representation of $\hB_n$ and $\tP_n$ using the non-counted sprouting operator $\S^+$ and two weight scaling operators $\W_{\hB/\tP}$:
\begin{align}\label{inverse_fact_or}
    \W_{\hB}(T_{n,k})= \frac{n}{k!\|b_1\|!\cdots\|b_k\|!}T_{n,k},\quad \W_{\tP}(T_{n,k})=(-1)^{\prod_{j}\|b_j\|}T_{n,k}.
\end{align}
As a result, we find $F^{\hB}_{n+1}=\W_{\hB}\circ\S^+(F^{\hB}_n)$ and $F^{\tP}_{n+1}=\W_{\tP}\circ\S^+(F^{\tP}_n)$, where $\circ$ is the composition. In fact, using the non-counted sprouting operator $\S^+$, $F^{\tB}_{n+1}$ can be rewritten as $F^{\tB}_{n+1}=\S^+_C(F^{\tB}_{n})=\W_{\tB}\circ\S^+(F^{\tB}_n)$, where the weight scaling operator $\W_{\tB}$ is defined similarly using the tree weight \eqref{weight_B}. The tree representation for the first few terms of the Type-II Bell polynomials and bipartition polynomials are provided in \ref{Appendix_tree_repre}.

Algebraically, the associative algebra for words can be endowed with a coproduct $\Delta_{\Dd}$ and hence becomes a bialgebra, which is known as the Dynkin-Fa\`a di Bruno bialgebra \cite{munthe2008hopf}. Similar things can be done for planar trees, where the product is just a way to graft subtrees to form bigger trees and coproduct is a way to decompose a tree into subtrees. In fact, both the associative rule and coproduct for planar trees are not unique, which leads to various algebraic structures for the planar trees. Some of them can be further equipped with a Hopf algebraic structure after introducing an additional antipode operation $S_N$. These algebraic structures for planar trees have found various applications in different branches of mathematics and physics such as the renormalization theory for quantum fields \cite{connes1999hopf,connes2000renormalization} and stochastic partial differential equations \cite{bruned2019algebraic}, the Runge-Kutta method on the differential manifold \cite{munthe1995lie,munthe1998runge,lundervold2011hopf}, and the index theory in noncommutative geometry \cite{connes1995local}. In this work, we will not further discuss the algebraic structures associated with the combinatorial class (see the definition in the following subsection) of the noncommutative Bell/bipartition polynomials. The Hopf algebra theory for the Type-I noncommutative Bell polynomials has been discussed in  \cite{munthe1998runge,lundervold2011hopf}. We leave other parts as an open topic and provide in-depth discussions should we need them in the future. 

\subsection{Tree representation of the combinatorial MZE}
The tree representation of the noncommutative polynomials enables a tree representation for the CMZE, which leads us to find a functional equation for the tree-generating function that is similar to the combinatorial Dyson-Schwinger equation (DSE). To begin, we briefly review some combinatorial concepts introduced by Yeats in \cite{yeats2017combinatorial} to facilitate further discussion. For planar rooted trees defined in Section \ref{sec:tree_intro}, we define a combinatorial class $\mathfrak{C}$, which is made of a (probably infinite) set of trees graded by a size function $|\cdot|:\mathfrak{C}\rightarrow\Z_{\geq 0}$. A common choice of $|\cdot|$ is the total number of nodes (or vertices) of the tree. Given a combinatorial class $\mathfrak{C}$, the {\em argumented exponential generating function} $C(\cdot):\R\rightarrow\mathfrak{C}$ is given by the formal Taylor series:
\begin{align}\label{Com_class_GF}
    C(t)=\sum_{c\in\mathfrak{C}}c\frac{t^{|c|}}{|c|!}=\sum_{k=0}f_{k}\frac{t^k}{k!},
\end{align}
where $\sum_{k=0}f_k\frac{t^k}{k!}$ is the resummation of series with respect to the grade $k$, and $f_{k}\subset\mathfrak{C}$ is a forest of trees of the size $k$. Exponential generating functions written in terms of $f_k$ will be called the {\em graded exponential generating function}. A combinatorial Dyson-Schwinger equation is a functional equation for the generating function $C(t)$. For the simplest case \cite{foissy2008faa}, the equation can be written as:
\begin{align}\label{DSE}
    C(t)=t\B^+(g(C(t))).
\end{align}
Here $g=g(u)$ is a scalar function and admits a formal series expansion $g(u)=\sum_{n=0}^{\infty}g_ku^k/k!$ with $g_0=1$. When $u=u(t):\R\rightarrow \mathfrak{C}$ is chosen to be an exponential generating function, we still assume that this formal series expansion is valid and naturally we have $g(\cdot):\mathfrak{C}\rightarrow\mathfrak{C}$.  $\B^+=\B^+(c_1,c_2,\cdots,c_n)$ in Eqn \eqref{DSE} is known as the {\em grafting operator} which connects subtrees $c_1,c_2,\cdots,c_n$ to a common root. For example, 
\begin{align*}
  \B^+(
{\fontsize{4}{1}\selectfont
\begin{forest}
 for tree={grow=90, l=1mm}
[]
 \path[fill=black]  (!.parent anchor) circle[radius=2pt];
\end{forest}
} 
{\fontsize{4}{1}\selectfont
\begin{forest}
 for tree={grow=90, l=1mm}
[[]]
 \path[fill=black]  (!.parent anchor) circle[radius=2pt];
 \path[fill=black] (!1.child anchor) circle[radius=2pt];
\end{forest}
})
=
{\fontsize{4}{1}\selectfont
\begin{forest}
 for tree={grow=90, l=1mm}
[[][[]]]
 \path[fill=black]  (!.parent anchor) circle[radius=2pt];
 \path[fill=black] (!1.child anchor) circle[radius=2pt];
 \path[fill=black] (!2.child anchor) circle[radius=2pt];
 \path[fill=black] (!21.child anchor) circle[radius=2pt];
\end{forest}
}.
\end{align*}
To be noticed that different from the previously defined right grafting product $\lgraf$, the grafting defined in $\B^+$ has no order and introduces a {\em new} vertex as the common root. In \eqref{DSE}, after taking the formal series expansion of $g(C(t))$, the grafting operator $\B^+$ will combine trees in the expansion series to form bigger trees, which are then matched with trees of different sizes corresponding to the exponential generating function $C(t)$ on the RHS of \eqref{DSE}. This leads to the recursive way to construct the unique solution to \eqref{DSE}
\footnote{Solution \eqref{com_class_eqn_sol} is given in \cite{yeats2017combinatorial} and \cite{foissy2008faa}. We present it with slight modifications adapted to our definition of the generating function $C(t)$ and $\{g_k\}_{k=0}^{\infty}$.}
:
\begin{equation}\label{com_class_eqn_sol}
\begin{aligned}
\begin{dcases}
\bm c_1&=
{\fontsize{4}{1}\selectfont
\begin{forest}
 for tree={grow=90, l=1mm}
[]
 \path[fill=black]  (!.parent anchor) circle[radius=2pt];
\end{forest}
}
,\\
\bm c_{n+1}&=\sum_{k=1}^n\sum_{p_1+p_2+\cdots+p_k=n}\frac{g_k}{k!}\B^+(\bm c_{p_1},\bm c_{p_1},\cdots,\bm c_{p_k})
\end{dcases}
\end{aligned}
\end{equation}
where $\bm c_j=\{c_j^1/|c_j^1|,c_j^2/|c_j^2|,\cdots\}\subset\mathfrak{C}$ is a collection of (normalized) trees with {\em different} sizes, which distinguishes it from the same-size forest $f_k$ defined in \eqref{Com_class_GF}. Given a function $g$, or equivalently all its expansion coefficients $\{g_k\}_{k=0}^{\infty}$, using \eqref{com_class_eqn_sol}, we can get the exact exponential generating function $C(t)$ and hence the combinatorial class $\mathfrak{C}$. 
The combinatorial DSE can be generalized to the multidimensional case and we get a system of combinatorial DSEs. Moreover, the grafting operator $\B^+$ can also be multi-degree. All these problems were discussed extensively in the work of Foissy \cite{foissy2008faa,foissy2010classification,foissy2014general,foissy2020multigraded}. The application to quantum field theory can be found in the groundbreaking work of Kreimer, Connes and Yeats \cite{connes1999hopf,connes2000renormalization,kreimer2006etude,yeats2008growth,yeats2017combinatorial}. CMZE can also be derived using the functional equation for exponential generating functions. This new set of equations corresponds to the combinatorial expansion for the memory kernel of the CMZE, and its form is similar to but different from the combinatorial DSE \eqref{DSE} due to the non-commutativity. Now we detail its construction.  

\paragraph{Time-independent CMZE}
We first focus on the tree representation of time-independent CMZE and show how its construction boils down to the reformulation of a functional equation for an exponential generating $C(t)$. To this end, we introduce two combinatorial classes $\mathfrak{C}$ and $\mathfrak{C}_{\tP}$. The size of trees within $\mathfrak{C}$ and $\mathfrak{C}_{\tP}$ is chosen to be $|\cdot|=\#(nodes)-1$, i.e. the total number of vertices (nodes) minus one. $\mathfrak{C}$ and $\mathfrak{C}_{\tP}$ are generated by the following graded exponential generating functions respectively:
\begin{align}
C(t)&=\sum_{k=0}^{\infty}f_k\frac{t^k}{k!}=
t{\fontsize{4}{1}\selectfont
\begin{forest}
 for tree={grow=90, l=1mm}
[[]]
 \path[fill=black]  (!.parent anchor) circle[radius=2pt]
                    (!1.child anchor) circle[radius=2pt];
\end{forest}
}
+\frac{t^2}{2!}
{\fontsize{4}{1}\selectfont
\begin{forest}
 for tree={grow=90, l=1mm}
[[][]]
 \path[fill=black]  (!.parent anchor) circle[radius=2pt]
                    (!1.child anchor) circle[radius=2pt]
                    (!2.child anchor) circle[radius=2pt];
\end{forest}
}
+\frac{t^3}{3!}
{\fontsize{4}{1}\selectfont
\begin{forest}
 for tree={grow=90, l=1mm}
[[][][]]
 \path[fill=black]  (!.parent anchor) circle[radius=2pt]
                    (!1.child anchor) circle[radius=2pt]
                    (!2.child anchor) circle[radius=2pt]
                    (!3.child anchor) circle[radius=2pt];
\end{forest}
}+\cdots\\
C_{\tP}(t)&=\sum_{k=0}^{\infty}f^{\tP}_k\frac{t^k}{k!}
=-{\fontsize{4}{1}\selectfont
\begin{forest}
 for tree={grow=90, l=1mm}
[[[]]]
 \path[fill=black]  (!.parent anchor) circle[radius=2pt];
 \path[fill=black] (!1.child anchor) circle[radius=2pt]
                   (!11.child anchor) circle[radius=2pt];
\end{forest}
}
+
{\fontsize{4}{1}\selectfont
\begin{forest}
for tree={grow=90, l=1mm}
[[][]]
\path[fill=black]  (!.parent anchor) circle[radius=2pt];
\path[fill=black] (!1.child anchor) circle[radius=2pt]
                 (!2.child anchor) circle[radius=2pt];
\end{forest}
}
+t\biggl[
{\fontsize{4}{1}\selectfont
\begin{forest}
 for tree={grow=90, l=1mm}
[[[[]]]]
 \path[fill=black]  (!.parent anchor) circle[radius=2pt];
 \path[fill=black] (!1.child anchor) circle[radius=2pt]
                    (!11.child anchor) circle[radius=2pt]
                   (!111.child anchor) circle[radius=2pt];
\end{forest}
}
-
{\fontsize{4}{1}\selectfont
\begin{forest}
 for tree={grow=90, l=1mm}
[[[]][]]
 \path[fill=black]  (!.parent anchor) circle[radius=2pt];
 \path[fill=black] (!1.child anchor) circle[radius=2pt]
                  (!11.child anchor) circle[radius=2pt]
                  (!2.child anchor) circle[radius=2pt];
\end{forest}
}
-
{\fontsize{4}{1}\selectfont
\begin{forest}
 for tree={parent anchor=east, child anchor=east, grow=90, l=1mm}
[[[][]]]
 \path[fill=black]  (!.parent anchor) circle[radius=2pt];
 \path[fill=black] (!1.child anchor) circle[radius=2pt]
                   (!11.child anchor) circle[radius=2pt]
                   (!12.child anchor) circle[radius=2pt];
\end{forest}
}
+
{\fontsize{4}{1}\selectfont
\begin{forest}
 for tree={grow=90, l=1mm}
[[][][]]
 \path[fill=black]  (!.parent anchor) circle[radius=2pt];
 \path[fill=black] (!1.child anchor) circle[radius=2pt]
                   (!2.child anchor) circle[radius=2pt]
                   (!3.child anchor) circle[radius=2pt];
\end{forest}
}
\biggr]
+\cdots
\label{CP(t)}
\end{align}
In fact, $f_k^{\tP}$ is $F_{k+1}^{\tP}$ in \eqref{Tree_graphs_bipart}. The subscript index $\tP$ implies that this generation function corresponds to the noncommutative bipartition polynomials. Now consider a function $F(\cdot):\mathfrak{C}\rightarrow\mathfrak{C}_{\tP}$ which admits a formal series expansion $F(u(t))=\sum_{k=0}^{\infty}F_k\lgraf\frac{(u(t))^k}{k!}$, where $F_k$ is a forest, not necessarily graded by $k$, and $u(t):\R\rightarrow\mathfrak{C}$ is an exponential generating function. Then the combinatorial expansion for the memory kernel of the time-independent CMZE can be formally induced by solving a functional equation:
\begin{align}\label{CMZE-DSE}
   C_{\tP}(t)=F(C(t)). 
\end{align}
This functional equation is similar to the combinatorial DSE \eqref{DSE} and can be solved recursively. However, unlike the DSE \eqref{DSE} in which one knows $g_k$ a prior and solves for the forests $\{\bm c_j\}_{j=1}^{\infty}$. For CMZE \eqref{CMZE-DSE}, it is just the opposite. i.e. we know explicitly the generating functions $C(t)$ and $C_{\tP}(t)$ and solve for the forests $\{F_k\}_{k=0}^{\infty}$. Specifically, following a similar procedure we use the formal Taylor series expansion to get the explicit expression of $F(C(t))$ and then match $t^n$ terms on the LHS of \eqref{CMZE-DSE}. Eventually, this leads to the tree solution:
\begin{equation}\label{tree_sol_CMZE}
\begin{aligned}
\begin{dcases}
    f_0^{\tP}&=F^{\tP}_1\\
    f_k^{\tP}&=\sum_{k=1}^n\sum_{p_1+p_2+\cdots+p_k=n}\frac{n!}{p_1!p_2!\cdots p_k!}F_k\lgraf\B^{\lgraf}(f_{p_1},f_{p_2},\cdots,f_{p_k}),\qquad k\geq 1.
\end{dcases}
\end{aligned}
\end{equation}
Here $\B^{\lgraf}$ is a right-grafting operator defined in terms of the right-grafting product $\lgraf$, which acts on trees as $\B^{\lgraf}(f_1,f_2,\cdots,f_m)=f_1\lgraf f_2\lgraf\cdots\lgraf f_m$. Note that the distributive property of the product $\lgraf$ enables further decomposition of the forest into trees and $\lgraf$ will combine them into new trees. As a result, we find that the tree representation of the first few $F_k$s is given by
\begin{equation}\label{first_few_tree_sol}
\begin{aligned}
F_0&={\fontsize{1}{1}\selectfont
\begin{forest}
 for tree={grow=90, l=1mm}
[[][]]
 \path[fill=black]  (!.parent anchor) circle[radius=2pt]
                    (!1.child anchor) circle[radius=2pt]
                    (!2.child anchor) circle[radius=2pt];
\end{forest}
}
-
{\fontsize{1}{1}\selectfont
\begin{forest}
 for tree={grow=90, l=1mm}
[[[]]]
 \path[fill=black]  (!.parent anchor) circle[radius=2pt]
                    (!1.child anchor) circle[radius=2pt]
                    (!11.child anchor) circle[radius=2pt];
\end{forest}
}\\
F_1&=
{\fontsize{1}{1}\selectfont
\begin{forest}
 for tree={grow=90, l=1mm}
[[[]]]
 \path[fill=black]  (!.parent anchor) circle[radius=2pt]
                    (!1.child anchor) circle[radius=2pt]
                    (!11.child anchor) circle[radius=2pt];
\end{forest}
}
-
{\fontsize{1}{1}\selectfont
\begin{forest}
 for tree={grow=90, l=1mm}
[[[]][]]
 \path[fill=black]  (!.parent anchor) circle[radius=2pt]
                    (!1.child anchor) circle[radius=2pt]
                    (!11.child anchor) circle[radius=2pt]
                    (!2.child anchor) circle[radius=2pt];
\end{forest}
}
\lgraf
{\fontsize{1}{1}\selectfont
\begin{forest}
 for tree={grow=90, l=1mm}
[[]]
 \path[draw]  (!.parent anchor) circle[radius=2pt]
              (!1.child anchor) circle[radius=2pt];
\end{forest}
}
-{\fontsize{1}{1}\selectfont
\begin{forest}
 for tree={grow=90, l=1mm}
[[][]]
 \path[fill=black]  (!.parent anchor) circle[radius=2pt]
                    (!1.child anchor) circle[radius=2pt]
                    (!2.child anchor) circle[radius=2pt];
\end{forest}
}
+
{\fontsize{1}{1}\selectfont
\begin{forest}
 for tree={grow=90, l=1mm}
[[][][]]
 \path[fill=black]  (!.parent anchor) circle[radius=2pt]
                    (!1.child anchor) circle[radius=2pt]
                    (!2.child anchor) circle[radius=2pt]
                    (!3.child anchor) circle[radius=2pt];
\end{forest}
}
\lgraf
{\fontsize{1}{1}\selectfont
\begin{forest}
 for tree={grow=90, l=1mm}
[[[]]]
 \path[draw]  (!.parent anchor) circle[radius=2pt]
              (!1.child anchor) circle[radius=2pt]
              (!11.child anchor) circle[radius=2pt];
\end{forest}
}\\
F_2&=
{\fontsize{1}{1}\selectfont
\begin{forest}
 for tree={grow=90, l=1mm}
[[[]][]]
 \path[fill=black]  (!.parent anchor) circle[radius=2pt]
                    (!1.child anchor) circle[radius=2pt]
                    (!11.child anchor) circle[radius=2pt]
                    (!2.child anchor) circle[radius=2pt];
\end{forest}
}
\lgraf
{\fontsize{1}{1}\selectfont
\begin{forest}
 for tree={grow=90, l=1mm}
[[]]
 \path[draw]  (!.parent anchor) circle[radius=2pt]
              (!1.child anchor) circle[radius=2pt];
\end{forest}
}
\lgraf
{\fontsize{1}{1}\selectfont
\begin{forest}
 for tree={grow=90, l=1mm}
[[][]]
 \path[fill=black]  (!.parent anchor) circle[radius=2pt]
                    (!1.child anchor) circle[radius=2pt]
                    (!2.child anchor) circle[radius=2pt];
\end{forest}
}
\lgraf
{\fontsize{1}{1}\selectfont
\begin{forest}
 for tree={grow=90, l=1mm}
[[[]]]
 \path[draw]  (!.parent anchor) circle[radius=2pt]
              (!1.child anchor) circle[radius=2pt]
              (!11.child anchor) circle[radius=2pt];
\end{forest}
}
+
{\fontsize{1}{1}\selectfont
\begin{forest}
 for tree={grow=90, l=1mm}
[[][]]
 \path[fill=black]  (!.parent anchor) circle[radius=2pt]
                    (!1.child anchor) circle[radius=2pt]
                    (!2.child anchor) circle[radius=2pt];
\end{forest}
}
-
{\fontsize{1}{1}\selectfont
\begin{forest}
 for tree={grow=90, l=1mm}
[[][][]]
 \path[fill=black]  (!.parent anchor) circle[radius=2pt]
                    (!1.child anchor) circle[radius=2pt]
                    (!2.child anchor) circle[radius=2pt]
                    (!3.child anchor) circle[radius=2pt];
\end{forest}
}
\lgraf
{\fontsize{1}{1}\selectfont
\begin{forest}
 for tree={grow=90, l=1mm}
[[]]
 \path[draw]  (!.parent anchor) circle[radius=2pt]
              (!1.child anchor) circle[radius=2pt];
\end{forest}
}
\lgraf
{\fontsize{1}{1}\selectfont
\begin{forest}
 for tree={grow=90, l=1mm}
[[][]]
 \path[fill=black]  (!.parent anchor) circle[radius=2pt]
                    (!1.child anchor) circle[radius=2pt]
                    (!2.child anchor) circle[radius=2pt];
\end{forest}
}
\lgraf
{\fontsize{1}{1}\selectfont
\begin{forest}
 for tree={grow=90, l=1mm}
[[[]]]
 \path[draw]  (!.parent anchor) circle[radius=2pt]
              (!1.child anchor) circle[radius=2pt]
              (!11.child anchor) circle[radius=2pt];
\end{forest}
}\\
&\ \ \ 
-
{\fontsize{1}{1}\selectfont
\begin{forest}
 for tree={grow=90, l=1mm}
[[[]]]
 \path[fill=black]  (!.parent anchor) circle[radius=2pt]
                    (!1.child anchor) circle[radius=2pt]
                    (!11.child anchor) circle[radius=2pt];
\end{forest}
}
+
{\fontsize{1}{1}\selectfont
\begin{forest}
 for tree={grow=90, l=1mm}
[[[]][]]
 \path[fill=black]  (!.parent anchor) circle[radius=2pt]
                    (!1.child anchor) circle[radius=2pt]
                    (!11.child anchor) circle[radius=2pt]
                    (!2.child anchor) circle[radius=2pt];
\end{forest}
}
\lgraf
{\fontsize{1}{1}\selectfont
\begin{forest}
 for tree={grow=90, l=1mm}
[[]]
 \path[draw]  (!.parent anchor) circle[radius=2pt]
              (!1.child anchor) circle[radius=2pt];
\end{forest}
}
-
{\fontsize{1}{1}\selectfont
\begin{forest}
 for tree={grow=90, l=1mm}
[[[]][][]]
 \path[fill=black]  (!.parent anchor) circle[radius=2pt]
                    (!1.child anchor) circle[radius=2pt]
                    (!11.child anchor) circle[radius=2pt]
                    (!2.child anchor) circle[radius=2pt]
                    (!3.child anchor) circle[radius=2pt];
\end{forest}
}
\lgraf
{\fontsize{1}{1}\selectfont
\begin{forest}
 for tree={grow=90, l=1mm}
[[[]]]
 \path[draw]  (!.parent anchor) circle[radius=2pt]
              (!1.child anchor) circle[radius=2pt]
              (!11.child anchor) circle[radius=2pt];
\end{forest}
}
-
{\fontsize{1}{1}\selectfont
\begin{forest}
 for tree={grow=90, l=1mm}
[[][][]]
 \path[fill=black]  (!.parent anchor) circle[radius=2pt]
                    (!1.child anchor) circle[radius=2pt]
                    (!2.child anchor) circle[radius=2pt]
                    (!3.child anchor) circle[radius=2pt];
\end{forest}
}
\lgraf
{\fontsize{1}{1}\selectfont
\begin{forest}
 for tree={grow=90, l=1mm}
[[]]
 \path[draw]  (!.parent anchor) circle[radius=2pt]
              (!1.child anchor) circle[radius=2pt];
\end{forest}
}
+
{\fontsize{1}{1}\selectfont
\begin{forest}
 for tree={grow=90, l=1mm}
[[][][][]]
 \path[fill=black]  (!.parent anchor) circle[radius=2pt]
                    (!1.child anchor) circle[radius=2pt]
                    (!2.child anchor) circle[radius=2pt]
                    (!3.child anchor) circle[radius=2pt]
                    (!4.child anchor) circle[radius=2pt];
\end{forest}
}
\lgraf
{\fontsize{1}{1}\selectfont
\begin{forest}
 for tree={grow=90, l=1mm}
[[[]]]
 \path[draw]  (!.parent anchor) circle[radius=2pt]
              (!1.child anchor) circle[radius=2pt]
              (!11.child anchor) circle[radius=2pt];
\end{forest}
}
\end{aligned}
\end{equation}
Here we have introduced the {\em inverse} trees:
\begin{align}
{\fontsize{1}{1}\selectfont
\begin{forest}
 for tree={grow=90, l=1mm}
[[]]
 \path[draw]  (!.parent anchor) circle[radius=2pt]
              (!1.child anchor) circle[radius=2pt];
\end{forest}
},\quad
{\fontsize{1}{1}\selectfont
\begin{forest}
 for tree={grow=90, l=1mm}
[[[]]]
 \path[draw]  (!.parent anchor) circle[radius=2pt]
              (!1.child anchor) circle[radius=2pt]
              (!11.child anchor) circle[radius=2pt];
\end{forest}
}
,\quad
{\fontsize{1}{1}\selectfont
\begin{forest}
 for tree={grow=90, l=1mm}
[[[[]]]]
 \path[draw]  (!.parent anchor) circle[radius=2pt]
              (!1.child anchor) circle[radius=2pt]
              (!11.child anchor) circle[radius=2pt]
              (!111.child anchor) circle[radius=2pt];
\end{forest}
}
\quad\cdots,
\end{align}
which acts on the regular trees as:
\begin{align*}
 {\fontsize{1}{1}\selectfont
\begin{forest}
 for tree={grow=90, l=1mm}
[[]]
 \path[draw]  (!.parent anchor) circle[radius=2pt]
              (!1.child anchor) circle[radius=2pt];
\end{forest}
}\lgraf
{\fontsize{1}{1}\selectfont
\begin{forest}
 for tree={grow=90, l=1mm}
[[]]
 \path[fill=black]  (!.parent anchor) circle[radius=2pt]
              (!1.child anchor) circle[radius=2pt];
\end{forest}
}
={\fontsize{1}{1}\selectfont
\begin{forest}
 for tree={grow=90, l=1mm}
[[]]
 \path[fill=black]  (!.parent anchor) circle[radius=2pt]
              (!1.child anchor) circle[radius=2pt];
\end{forest}
}\lgraf
 {\fontsize{1}{1}\selectfont
\begin{forest}
 for tree={grow=90, l=1mm}
[[]]
 \path[draw]  (!.parent anchor) circle[radius=2pt]
              (!1.child anchor) circle[radius=2pt];
\end{forest}
}=\I=\blkdbig,\quad 
{\fontsize{1}{1}\selectfont
\begin{forest}
 for tree={grow=90, l=1mm}
[[]]
 \path[draw]  (!.parent anchor) circle[radius=2pt]
              (!1.child anchor) circle[radius=2pt];
\end{forest}
}\lgraf
{\fontsize{1}{1}\selectfont
\begin{forest}
 for tree={grow=90, l=1mm}
[[[]]]
 \path[fill=black]  (!.parent anchor) circle[radius=2pt]
              (!1.child anchor) circle[radius=2pt]
              (!11.child anchor) circle[radius=2pt];
\end{forest}
}
={\fontsize{1}{1}\selectfont
\begin{forest}
 for tree={grow=90, l=1mm}
[[[]]]
 \path[fill=black]  (!.parent anchor) circle[radius=2pt]
              (!1.child anchor) circle[radius=2pt]
              (!11.child anchor) circle[radius=2pt];
\end{forest}
}\lgraf
 {\fontsize{1}{1}\selectfont
\begin{forest}
 for tree={grow=90, l=1mm}
[[]]
 \path[draw]  (!.parent anchor) circle[radius=2pt]
              (!1.child anchor) circle[radius=2pt];
\end{forest}
}=
{\fontsize{1}{1}\selectfont
\begin{forest}
 for tree={grow=90, l=1mm}
[[]]
 \path[fill=black]  (!.parent anchor) circle[radius=2pt]
              (!1.child anchor) circle[radius=2pt];
\end{forest}
}.
\end{align*}
Now, it is fairly simple to use the tree solution \eqref{tree_sol_CMZE} to get the combinatorial MZE. One only needs to use the following dictionary to translate trees into operators:
\begin{equation}\label{trans_rule_2}
\begin{aligned}
{\fontsize{4}{1}\selectfont
\begin{forest}
 for tree={grow=90, l=1mm}
[[]]
 \path[fill=black]  (!.parent anchor) circle[radius=2pt]
                    (!1.child anchor) circle[radius=2pt];
\end{forest}
}
&
\Rightarrow \P\L\P,
\qquad
{\fontsize{4}{1}\selectfont
\begin{forest}
 for tree={grow=90, l=1mm}
[[]]
 \path[draw]  (!.parent anchor) circle[radius=2pt]
              (!1.child anchor) circle[radius=2pt];
\end{forest}
}
\Rightarrow (\P\L\P|_V)^{-1},
\\
{\fontsize{4}{1}\selectfont
\begin{forest}
 for tree={grow=90, l=1mm}
[[][]]
 \path[fill=black]  (!.parent anchor) circle[radius=2pt]
                    (!1.child anchor) circle[radius=2pt]
                    (!2.child anchor) circle[radius=2pt];
\end{forest}
}
&\Rightarrow \P\L^2\P,\qquad
{\fontsize{4}{1}\selectfont
\begin{forest}
 for tree={grow=90, l=1mm}
[[][][]]
 \path[fill=black]  (!.parent anchor) circle[radius=2pt]
                    (!1.child anchor) circle[radius=2pt]
                    (!2.child anchor) circle[radius=2pt]
                    (!3.child anchor) circle[radius=2pt];
\end{forest}
}
\Rightarrow \P\L^3\P,\qquad
{\fontsize{4}{1}\selectfont
\begin{forest}
 for tree={grow=90, l=1mm}
[[][][][]]
 \path[fill=black]  (!.parent anchor) circle[radius=2pt]
                    (!1.child anchor) circle[radius=2pt]
                    (!2.child anchor) circle[radius=2pt]
                    (!3.child anchor) circle[radius=2pt]
                    (!4.child anchor) circle[radius=2pt];
\end{forest}
}
\Rightarrow \P\L^4\P,\qquad \cdots\\
{\fontsize{4}{1}\selectfont
\begin{forest}
 for tree={grow=90, l=1mm}
[[[]]]
 \path[fill=black]  (!.parent anchor) circle[radius=2pt]
                    (!1.child anchor) circle[radius=2pt]
                    (!11.child anchor) circle[radius=2pt];
\end{forest}
}
&\Rightarrow (\P\L\P)^2,
\qquad
{\fontsize{4}{1}\selectfont
\begin{forest}
 for tree={grow=90, l=1mm}
[[[[]]]]
 \path[fill=black]  (!.parent anchor) circle[radius=2pt]
                    (!1.child anchor) circle[radius=2pt]
                    (!11.child anchor) circle[radius=2pt]
                    (!111.child anchor) circle[radius=2pt];
\end{forest}
}
\Rightarrow (\P\L\P)^3,\qquad \cdots\\
{\fontsize{4}{1}\selectfont
\begin{forest}
 for tree={grow=90, l=1mm}
[[[]]]
 \path[draw]  (!.parent anchor) circle[radius=2pt]
                    (!1.child anchor) circle[radius=2pt]
                    (!11.child anchor) circle[radius=2pt];
\end{forest}
}
&\Rightarrow (\P\L\P|_V)^{-2},
\qquad
{\fontsize{4}{1}\selectfont
\begin{forest}
 for tree={grow=90, l=1mm}
[[[[]]]]
 \path[draw]  (!.parent anchor) circle[radius=2pt]
                    (!1.child anchor) circle[radius=2pt]
                    (!11.child anchor) circle[radius=2pt]
                    (!111.child anchor) circle[radius=2pt];
\end{forest}
}
\Rightarrow (\P\L\P|_V)^{-3},
\qquad \cdots.
\end{aligned}
\end{equation}
As a result of the homomorphism induced by the translation, we find that the exponential generating function $C(t)$ and $C_{\tP}(t)$ are mapped into evolution operators as:
\begin{align}
    C(t) \Rightarrow \P e^{t\L}\P -\P,\qquad C_{\tP}(t)\Rightarrow \P\L e^{t\Q\L}\Q\L\P.
\end{align}
Moreover, the recursive equation \eqref{tree_sol_CMZE} is mapped into \eqref{eqn:rec_op_nonc}, and the first few tree solution \eqref{first_few_tree_sol} is merely a translation of the operator series \eqref{first_few_f_n}. At this point, since we know all $F_k$s and the translation rule between trees and operators, it is easy to reformulate the time-independent CMZE \eqref{op_id_1_nonc} as the differential equation for the generating function $C(t)$: 
\begin{align}\label{CMZE-tree}
   \frac{d}{dt}C(t)=(C(t)+C(0))\lgraf F_0+\int_0^t(C(t-s)+C(0))\lgraf F(C(s))ds
\end{align}
Eqn \eqref{CMZE-tree} is equivalent to \eqref{op_id_1_nonc} but the form of it is slightly different since we defined $C(t)$ as $\P e^{t\L}\P -\P$ instead of $\P e^{t\L}\P$. As we have seen from the case study of Section \ref{sec:main_thm}, $C(t)+C(0)$ corresponds to the time correlation function or Green's function of an interactive many-body system. Hence \eqref{CMZE-tree} is just a tree representation of the self-consistent EOM for the correlation (Green's) function.
\paragraph{Time-dependent CMZE} In the derivation of the time-dependent CMZE, we have used different alphabets to derive the expansion series. Correspondingly, the tree representation of the time-dependent CMZE requires the usage of decorated planar trees, i.e. trees with different colors. With the decorated trees, the combinatorial expansion for the time-dependent CMZE can be reformulated into a system of combinatorial DSEs, which can be solved recursively. To begin, we introduce the following graded exponential generating function for decorated trees.
\begin{equation}\label{NMZE-EGF_noneq}
\begin{aligned}
C_{\tB,\blkdbig}^{\ydbig}(t)
&=\sum_{k=0}f_k^{\tB}\frac{t^k}{k!}
=
\ydbig
\lgraf
\biggl[
\blkdbig
+
t
{\fontsize{1}{1}\selectfont
\begin{forest}
 for tree={grow=90, l=1mm}
[[]]
 \path[fill=black]  (!.parent anchor) circle[radius=2pt]
                    (!1.child anchor) circle[radius=2pt];
\end{forest}
}
+
\frac{t^2}{2}
\biggl(
{\fontsize{1}{1}\selectfont
\begin{forest}
 for tree={grow=90, l=1mm}
[[[]]]
 \path[fill=black]  (!.parent anchor) circle[radius=2pt];
 \path[fill=black] (!1.child anchor) circle[radius=2pt]
                   (!11.child anchor) circle[radius=2pt];
\end{forest}
}
+
{\fontsize{1}{1}\selectfont
\begin{forest}
for tree={grow=90, l=1mm}
[[][]]
\path[fill=black]  (!.parent anchor) circle[radius=2pt];
\path[fill=black] (!1.child anchor) circle[radius=2pt]
                 (!2.child anchor) circle[radius=2pt];
\end{forest}
}
\biggr)
+\cdots
\biggr]
\lgraf
\ydbig\\
C_{\tB_1,\rdbig}^{\gdbig}(t,s)&=\sum_{k=0}f_k^{\tB_1}(s)\frac{t^k}{k!}
=
\gdbig\lgraf
\biggl[
\rdbig+
t
{\fontsize{1}{1}\selectfont
\begin{forest}
 for tree={grow=90, l=1mm}
[[]]
 \path[fill=red]  (!.parent anchor) circle[radius=2pt]
                  (!1.child anchor) circle[radius=2pt];
\end{forest}
}
+\frac{t^2}{2}
\biggl(
{\fontsize{1}{1}\selectfont
\begin{forest}
 for tree={grow=90, l=1mm}
[[[]]]
 \path[fill=red]  (!.parent anchor) circle[radius=2pt];
 \path[fill=red] (!1.child anchor) circle[radius=2pt]
                   (!11.child anchor) circle[radius=2pt];
\end{forest}
}
+
{\fontsize{1}{1}\selectfont
\begin{forest}
for tree={grow=90, l=1mm}
[[][]]
\path[fill=red]  (!.parent anchor) circle[radius=2pt];
\path[fill=red] (!1.child anchor) circle[radius=2pt]
                 (!2.child anchor) circle[radius=2pt];
\end{forest}
}
\biggr)
+\cdots\biggr]\\
C_{\tB_2,\rdbig}^{\pdbig}(s,t)&=\sum_{k=0}f_k^{\tB_1}(t)\frac{s^k}{k!}
=
\biggl[
\rdbig-
s
{\fontsize{1}{1}\selectfont
\begin{forest}
 for tree={grow=90, l=1mm}
[[]]
 \path[fill=red]  (!.parent anchor) circle[radius=2pt]
                  (!1.child anchor) circle[radius=2pt];
\end{forest}
}
+\frac{s^2}{2}
\biggl(
{\fontsize{1}{1}\selectfont
\begin{forest}
 for tree={grow=90, l=1mm}
[[[]]]
 \path[fill=red]  (!.parent anchor) circle[radius=2pt];
 \path[fill=red] (!1.child anchor) circle[radius=2pt]
                   (!11.child anchor) circle[radius=2pt];
\end{forest}
}
-
{\fontsize{1}{1}\selectfont
\begin{forest}
for tree={grow=90, l=1mm}
[[][]]
\path[fill=red]  (!.parent anchor) circle[radius=2pt];
\path[fill=red] (!1.child anchor) circle[radius=2pt]
                 (!2.child anchor) circle[radius=2pt];
\end{forest}
}
\biggr)
+\cdots\biggr]\lgraf\pdbig
\end{aligned}
\end{equation}
Specifically, $C_{\tB,\blkdbig}^{\ydbig}(t)$, $C_{\tB_1,\rdbig}^{\gdbig}(t,s)$ and $C_{\tB_2,\rdbig}^{\pdbig}(s,t)$ are generating functions for combinatorial classes $\mathfrak{C}$, $\mathfrak{C}_{\tB_1}$ and $\mathfrak{C}_{\tB_2}$. In $\mathfrak{C}$, $\mathfrak{C}_{\tB_1}$ and $\mathfrak{C}_{\tB_2}$, the size functions are defined as $|\cdot|=\#(nodes)-3$, $|\cdot|=\#(nodes)-2$ and $|\cdot|=\#(nodes)-2$ respectively. These three combinatorial classes all correspond to the Type-I noncommutative Bell polynomials with additional decorations of one or two vertices with different colors. This fact is reflected in the subscript indices $\tB,\tB_1,\tB_2$ and the additional colored dots in expression of $C_{\tB,\blkdbig}^{\ydbig}(t)$, $C_{\tB_1,\rdbig}^{\gdbig}(t,s)$ and $C_{\tB_2,\rdbig}^{\pdbig}(s,t)$. Now consider a function $F(s,\cdot):\R\times\mathfrak{C}\times\mathfrak{C}_{F}\rightarrow\mathfrak{C}_{\tB_1}$ which admits a formal series expansion $F(s,u(t))=\sum_{k=0}^{\infty}\frac{(u(t))^k}{k!}\lgraf F_k(s)$. Here $u(t):\R\rightarrow\mathfrak{C}$ is a graded exponential generating function. $F_k(s)$ is a forest depending on $s$, and for any fixed $s$, $\{F_k(s)\}_{k=1}^{\infty}$ forms another combinatorial class $\mathfrak{C}_{F}$ (currently unknown). Similarly, we define 
$G(t,\cdot):\R\times\mathfrak{C}\times\mathfrak{C}_G\rightarrow\mathfrak{C}_{\tB_2}$, which is induced by a formal series expansion $G(t,u(s))=\sum_{k=0}^{\infty}G_k(t)\lgraf\frac{(u(s))^k}{k!}$. 
With all these definitions, the combinatorial expansion for the memory kernel of the time-dependent CMZE can be obtained by solving the following decoupled system of functional equations for decorated trees: 
\begin{equation}\label{CMZE-DSE_noneq}
    \begin{aligned}
    \begin{dcases}
    C_{\tB_1,\rdbig}^{\gdbig}(t,s)&=F(s,C_{\tB,\blkdbig}^{\ydbig}(t))\\
    C_{\tB_2,\rdbig}^{\pdbig}(s,t)&=G(t,C_{\tB,\blkdbig}^{\ydbig}(s))
    \end{dcases}
    \end{aligned}
\end{equation}
Here, one needs to solve for $F_k(s)$ and $G_k(t)$ separately. which can be obtained using the recurrence relations:
\begin{equation}\label{tree_sol_CMZE_FG}
\begin{aligned}
&
\begin{dcases}
    f_0^{\tB_1}(s)&=\gdbig\\
    f_k^{\tB_1}(s)&=\sum_{k=1}^n\sum_{p_1+p_2+\cdots+p_k=n}\frac{n!}{p_1!p_2!\cdots p_k!}\B^{\lgraf}(f^{\tB}_{p_1},f^{\tB}_{p_2},\cdots,f^{\tB}_{p_k})\lgraf F_k(s),\qquad k\geq 1
\end{dcases}
\\
&
\begin{dcases}
    f_0^{\tB_2}(t)&=\pdbig\\
    f_0^{\tB_2}(t)&=\sum_{k=1}^n\sum_{p_1+p_2+\cdots+p_k=n}\frac{n!}{p_1!p_2!\cdots p_k!}G_k(t)\lgraf\B^{\lgraf}(f^{\tB}_{p_1},f^{\tB}_{p_2},\cdots,f^{\tB}_{p_k}),\qquad k\geq 1
\end{dcases}
\end{aligned}
\end{equation}
where $f_k^{\tB}$, $f_k^{\tB_1}$ and $f_k^{\tB_2}$ are forest of the size $k$ defined in \eqref{NMZE-EGF_noneq}. In the course of solving for $F_k(s)$ and $G_k(t)$, we need to introduce inverse trees: $\bledbig,
{\fontsize{1}{1}\selectfont
\begin{forest}
 for tree={grow=90, l=1mm}
[[]]
 \path[fill=blue]  (!.parent anchor) circle[radius=2pt];
 \path[fill=blue] (!1.child anchor) circle[radius=2pt];
\end{forest}
}
,\cdots$. As a result, we find that the first few tree solutions are given by:
\begin{equation}\label{tree_Fs}
\begin{aligned}
    F_0(s)&=\gdbig\\
    F_1(s)&=\gdbig\lgraf
{\fontsize{1}{1}\selectfont
\begin{forest}
 for tree={grow=90, l=1mm}
[[]]
 \path[fill=red]  (!.parent anchor) circle[radius=2pt];
 \path[fill=red] (!1.child anchor) circle[radius=2pt];
\end{forest}
}
\lgraf\bledbig\\
F_2(s)&=\gdbig\lgraf
{\fontsize{1}{1}\selectfont
\begin{forest}
 for tree={grow=90, l=1mm}
[[[]]]
 \path[fill=red]  (!.parent anchor) circle[radius=2pt];
 \path[fill=red] (!1.child anchor) circle[radius=2pt];
  \path[fill=red] (!11.child anchor) circle[radius=2pt];
\end{forest}
}\lgraf
{\fontsize{1}{1}\selectfont
\begin{forest}
 for tree={grow=90, l=1mm}
[[]]
 \path[fill=blue]  (!.parent anchor) circle[radius=2pt];
 \path[fill=blue] (!1.child anchor) circle[radius=2pt];
\end{forest}
}
+
\gdbig\lgraf
{\fontsize{1}{1}\selectfont
\begin{forest}
 for tree={grow=90, l=1mm}
[[][]]
 \path[fill=red]  (!.parent anchor) circle[radius=2pt];
 \path[fill=red] (!1.child anchor) circle[radius=2pt];
  \path[fill=red] (!2.child anchor) circle[radius=2pt];
\end{forest}
}\lgraf
{\fontsize{1}{1}\selectfont
\begin{forest}
 for tree={grow=90, l=1mm}
[[]]
 \path[fill=blue]  (!.parent anchor) circle[radius=2pt];
 \path[fill=blue] (!1.child anchor) circle[radius=2pt];
\end{forest}
}
-
\gdbig\lgraf
{\fontsize{1}{1}\selectfont
\begin{forest}
 for tree={grow=90, l=1mm}
[[]]
 \path[fill=red]  (!.parent anchor) circle[radius=2pt];
 \path[fill=red] (!1.child anchor) circle[radius=2pt];
\end{forest}
}
\lgraf
\bledbig
\lgraf
\ydbig
\lgraf
{\fontsize{1}{1}\selectfont
\begin{forest}
 for tree={grow=90, l=1mm}
[[[]]]
 \path[fill=black]  (!.parent anchor) circle[radius=2pt];
 \path[fill=black] (!1.child anchor) circle[radius=2pt];
 \path[fill=black] (!11.child anchor) circle[radius=2pt];
\end{forest}
}
\lgraf
{\fontsize{1}{1}\selectfont
\begin{forest}
 for tree={grow=90, l=1mm}
[[]]
 \path[fill=blue]  (!.parent anchor) circle[radius=2pt];
 \path[fill=blue] (!1.child anchor) circle[radius=2pt];
\end{forest}
}
-
\gdbig\lgraf
{\fontsize{1}{1}\selectfont
\begin{forest}
 for tree={grow=90, l=1mm}
[[]]
 \path[fill=red]  (!.parent anchor) circle[radius=2pt];
 \path[fill=red] (!1.child anchor) circle[radius=2pt];
\end{forest}
}
\lgraf
\bledbig
\lgraf
\ydbig
\lgraf
{\fontsize{1}{1}\selectfont
\begin{forest}
 for tree={grow=90, l=1mm}
[[][]]
 \path[fill=black]  (!.parent anchor) circle[radius=2pt];
 \path[fill=black] (!1.child anchor) circle[radius=2pt];
 \path[fill=black] (!2.child anchor) circle[radius=2pt];
\end{forest}
}
\lgraf
{\fontsize{1}{1}\selectfont
\begin{forest}
 for tree={grow=90, l=1mm}
[[]]
 \path[fill=blue]  (!.parent anchor) circle[radius=2pt];
 \path[fill=blue] (!1.child anchor) circle[radius=2pt];
\end{forest}
},
\end{aligned}
\end{equation}
and
\begin{equation}\label{tree_Gt}
\begin{aligned}
    G_0(t)&=\pdbig\\
    G_1(t)&=-
{\fontsize{1}{1}\selectfont
\begin{forest}
 for tree={grow=90, l=1mm}
[[]]
 \path[fill=red]  (!.parent anchor) circle[radius=2pt];
 \path[fill=red] (!1.child anchor) circle[radius=2pt];
\end{forest}
}
\lgraf\pdbig\lgraf\bledbig\\
G_2(t)&=
{\fontsize{1}{1}\selectfont
\begin{forest}
 for tree={grow=90, l=1mm}
[[[]]]
 \path[fill=red]  (!.parent anchor) circle[radius=2pt];
 \path[fill=red] (!1.child anchor) circle[radius=2pt];
  \path[fill=red] (!11.child anchor) circle[radius=2pt];
\end{forest}
}\lgraf\pdbig\lgraf
{\fontsize{1}{1}\selectfont
\begin{forest}
 for tree={grow=90, l=1mm}
[[]]
 \path[fill=blue]  (!.parent anchor) circle[radius=2pt];
 \path[fill=blue] (!1.child anchor) circle[radius=2pt];
\end{forest}
}
-
{\fontsize{1}{1}\selectfont
\begin{forest}
 for tree={grow=90, l=1mm}
[[][]]
 \path[fill=red]  (!.parent anchor) circle[radius=2pt];
 \path[fill=red] (!1.child anchor) circle[radius=2pt];
  \path[fill=red] (!2.child anchor) circle[radius=2pt];
\end{forest}
}\lgraf\pdbig\lgraf
{\fontsize{1}{1}\selectfont
\begin{forest}
 for tree={grow=90, l=1mm}
[[]]
 \path[fill=blue]  (!.parent anchor) circle[radius=2pt];
 \path[fill=blue] (!1.child anchor) circle[radius=2pt];
\end{forest}
}
+
{\fontsize{1}{1}\selectfont
\begin{forest}
 for tree={grow=90, l=1mm}
[[]]
 \path[fill=red]  (!.parent anchor) circle[radius=2pt];
 \path[fill=red] (!1.child anchor) circle[radius=2pt];
\end{forest}
}\lgraf
\pdbig
\lgraf
\bledbig
\lgraf
\ydbig
\lgraf
{\fontsize{1}{1}\selectfont
\begin{forest}
 for tree={grow=90, l=1mm}
[[[]]]
 \path[fill=black]  (!.parent anchor) circle[radius=2pt];
 \path[fill=black] (!1.child anchor) circle[radius=2pt];
 \path[fill=black] (!11.child anchor) circle[radius=2pt];
\end{forest}
}
\lgraf
{\fontsize{1}{1}\selectfont
\begin{forest}
 for tree={grow=90, l=1mm}
[[]]
 \path[fill=blue]  (!.parent anchor) circle[radius=2pt];
 \path[fill=blue] (!1.child anchor) circle[radius=2pt];
\end{forest}
}
+
{\fontsize{1}{1}\selectfont
\begin{forest}
 for tree={grow=90, l=1mm}
[[]]
 \path[fill=red]  (!.parent anchor) circle[radius=2pt];
 \path[fill=red] (!1.child anchor) circle[radius=2pt];
\end{forest}
}\lgraf
\pdbig
\lgraf
\bledbig
\lgraf
\ydbig
\lgraf
{\fontsize{1}{1}\selectfont
\begin{forest}
 for tree={grow=90, l=1mm}
[[][]]
 \path[fill=black]  (!.parent anchor) circle[radius=2pt];
 \path[fill=black] (!1.child anchor) circle[radius=2pt];
 \path[fill=black] (!2.child anchor) circle[radius=2pt];
\end{forest}
}
\lgraf
{\fontsize{1}{1}\selectfont
\begin{forest}
 for tree={grow=90, l=1mm}
[[]]
 \path[fill=blue]  (!.parent anchor) circle[radius=2pt];
 \path[fill=blue] (!1.child anchor) circle[radius=2pt];
\end{forest}
}.
\end{aligned}
\end{equation}
The operator-form time-dependent CMZE can be easily obtained using the dictionary of translation for decorated trees. Specifically, the simplest decorated trees are translated as follows:
\begin{equation}\label{trans_rule_3.0}
\begin{aligned}
    {\fontsize{4}{1}\selectfont
\begin{forest}
 for tree={grow=90, l=1mm}
[]
 \path[fill=yellow]  (!.parent anchor) circle[radius=2pt];
\end{forest}
}
\Rightarrow \P,
\qquad
{\fontsize{4}{1}\selectfont
\begin{forest}
 for tree={grow=90, l=1mm}
[]
 \path[fill=green]  (!.parent anchor) circle[radius=2pt];
\end{forest}
}
\Rightarrow \P\L(s),
\qquad
{\fontsize{4}{1}\selectfont
\begin{forest}
 for tree={grow=90, l=1mm}
[]
 \path[fill=pink]  (!.parent anchor) circle[radius=2pt];
\end{forest}
}
\Rightarrow \Q\L(t),
\qquad
{\fontsize{4}{1}\selectfont
\begin{forest}
 for tree={grow=90, l=1mm}
[]
 \path[fill=blue]  (!.parent anchor) circle[radius=2pt];
\end{forest}
}
\Rightarrow (\P\L_1\P|_V)^{-1}.
\end{aligned}
\end{equation}
The height-1 decorated trees are translated as:
\begin{equation}\label{trans_rule_3.1}
\begin{aligned}
    {\fontsize{1}{1}\selectfont
\begin{forest}
 for tree={grow=90, l=1mm}
[[]]
 \path[fill=black]  (!.parent anchor) circle[radius=2pt]
                    (!1.child anchor) circle[radius=2pt];
\end{forest}
}
\Rightarrow \L_1,
\qquad
{\fontsize{1}{1}\selectfont
\begin{forest}
 for tree={grow=90, l=1mm}
[[]]
 \path[fill=red]  (!.parent anchor) circle[radius=2pt]
                  (!1.child anchor) circle[radius=2pt];
\end{forest}
}
\Rightarrow \Q\L_1.
\end{aligned}
\end{equation}
Trees of different colors are translated in the same way. For instance, the black-decorated trees are translated to operators as:
\begin{equation}\label{trans_rule_3.2}
\begin{aligned}
{\fontsize{4}{1}\selectfont
\begin{forest}
 for tree={grow=90, l=1mm}
[[]]
 \path[fill=black]  (!.parent anchor) circle[radius=2pt]
                    (!1.child anchor) circle[radius=2pt];
\end{forest}
}
&\Rightarrow \L_1,
\qquad
{\fontsize{4}{1}\selectfont
\begin{forest}
 for tree={grow=90, l=1mm}
[[][]]
 \path[fill=black]  (!.parent anchor) circle[radius=2pt]
                    (!1.child anchor) circle[radius=2pt]
                    (!2.child anchor) circle[radius=2pt];
\end{forest}
}
\Rightarrow \L_2,\qquad
{\fontsize{4}{1}\selectfont
\begin{forest}
 for tree={grow=90, l=1mm}
[[][][]]
 \path[fill=black]  (!.parent anchor) circle[radius=2pt]
                    (!1.child anchor) circle[radius=2pt]
                    (!2.child anchor) circle[radius=2pt]
                    (!3.child anchor) circle[radius=2pt];
\end{forest}
}
\Rightarrow \L_3,\qquad
{\fontsize{4}{1}\selectfont
\begin{forest}
 for tree={grow=90, l=1mm}
[[][][][]]
 \path[fill=black]  (!.parent anchor) circle[radius=2pt]
                    (!1.child anchor) circle[radius=2pt]
                    (!2.child anchor) circle[radius=2pt]
                    (!3.child anchor) circle[radius=2pt]
                    (!4.child anchor) circle[radius=2pt];
\end{forest}
}
\Rightarrow \L_4,\qquad \cdots\\
{\fontsize{4}{1}\selectfont
\begin{forest}
 for tree={grow=90, l=1mm}
[[[]]]
 \path[fill=black]  (!.parent anchor) circle[radius=2pt]
                    (!1.child anchor) circle[radius=2pt]
                    (!11.child anchor) circle[radius=2pt];
\end{forest}
}
&\Rightarrow (\L_1)^2,
\qquad
{\fontsize{4}{1}\selectfont
\begin{forest}
 for tree={grow=90, l=1mm}
[[[[]]]]
 \path[fill=black]  (!.parent anchor) circle[radius=2pt]
                    (!1.child anchor) circle[radius=2pt]
                    (!11.child anchor) circle[radius=2pt]
                    (!111.child anchor) circle[radius=2pt];
\end{forest}
}
\Rightarrow (\L_1)^3,\qquad \cdots.
\end{aligned}
\end{equation}
We similarly translate ${\fontsize{1}{1}\selectfont
\begin{forest}
 for tree={grow=90, l=1mm}
[[][]]
 \path[fill=red]  (!.parent anchor) circle[radius=2pt]
                    (!1.child anchor) circle[radius=2pt]
                    (!2.child anchor) circle[radius=2pt];
\end{forest}
}\Rightarrow \Q\L_2$, ${\fontsize{1}{1}\selectfont
\begin{forest}
 for tree={grow=90, l=1mm}
[[][][]]
 \path[fill=red]  (!.parent anchor) circle[radius=2pt]
                    (!1.child anchor) circle[radius=2pt]
                    (!2.child anchor) circle[radius=2pt]
                    (!3.child anchor) circle[radius=2pt];;
\end{forest}
}
\Rightarrow \Q\L_3$, and so on. Correspondingly, the recursive equation \eqref{tree_sol_CMZE_FG} is mapped into \eqref{47} and \eqref{52}, and the first few tree solutions \eqref{tree_Fs} and \eqref{tree_Gt} are merely a translation of the operator series \eqref{first_few_F_s} and \eqref{first_few_G_t}. After solving \eqref{CMZE-DSE_noneq} for $F_k(s)$ and $G_k(t)$ and acknowledging the translation rules, the time-dependent CMZE \eqref{eqn:rec_op_nonc_type21} can be rewritten as the following evolution equation for the exponential generating function $C_{\tB,\blkdbig}^{\ydbig}(t)$:
\begin{align}
    \partial_tC_{\tB,\blkdbig}^{\ydbig}(t)
    =C_{\tB,\blkdbig}^{\ydbig}(t)\lgraf\gdbig(t)
    +\int_0^tC_{\tB,\blkdbig}^{\ydbig}(s)\lgraf F(s,C_{\tB,\blkdbig}^{\ydbig}(t))
    \lgraf G(t,C_{\tB,\blkdbig}^{\ydbig}(s))ds,
\end{align}
where $\gdbig(t)\Rightarrow\P\L(t)$.
\subsection{Feynman diagram versus Tree diagram}
\label{sec:tree_feynman_diag}
The tree representation of the CMZE and the associated functional equations for the CMZE memory kernel actually provide a diagrammatic method to get the evolution equation of the correlation (Green's) function for many-body systems. Hence, it is constructive to compare this new tree diagrammatic method with the well-established Feynman diagrammatic method for many-body systems. As we will see,  Both methods can be categorized as (renormalized) perturbation expansion. The difference is that the Feynman diagram is a bookkeeping device for tracking the series expansion of Dyson's equation with respect to the interaction strength $v$, while the tree diagram is for tracking the series expansion of the MZE with respect to the evolution time $t$.  The following discussion is based on the quantum many-body theory for condensed matters \cite{stefanucci2013nonequilibrium}. Hence, we will use the terminology of condensed matter physics and re-denote $C(t)$ as an interactive Green's function $G(t)$ and call the memory kernel $K(t)$ the self-energy $\hat{\Sigma}(t)$. The reader can also refer to Section \ref{sec:app_hubbard} for a specific example of how to apply the CMZE into quantum many-body systems to get the evolution equation for $G(t)$ and how the MZ memory kernel function $K(t)$ becomes the self-energy $\hat{\Sigma}(t)$ \footnote{We note that the self-energy $\hat{\Sigma}(t)$ in the Mori-Zwanzig framework has a similar function but different from the $\Sigma(t)$ used in the regular many-body perturbation theory \cite{mahan2013many,stefanucci2013nonequilibrium}.}. In this section, all the discussion will be quite formal. Our goal is to show the common feature between the tree diagram and the Feynman diagram, and the similarity between the combinatorial expansion for the CMZE self-energy $\hat{\Sigma}(t)$ and the skeleton expansion for the self-energy $\Sigma(t)$ used in the regular many-body perturbation theory. 

As we mentioned in Section \ref{sec:into_main_result}, the evolution equation for the interactive Green's function $G=G(p,i\omega)$ of a quantum many-body system is given by Dyson's equation:
\begin{equation}\label{Dyson_eqn}
G=G_0+G\Sigma G_0=G_0+G_0\Sigma G_0+G_0\Sigma G_0\Sigma G_0+\cdots
\end{equation}
This series admits a natural diagrammatic representation:
\begin{equation}\label{Greenfdiagram}
\begin{aligned}
\begin{fmffile}{Geenfgraph1}
\parbox{40pt}{
\begin{fmfgraph*}(40,40)
      \fmfleft{v1}
      \fmfright{v2}
      \fmf{dbl_plain_arrow}{v2,v1}
\end{fmfgraph*}}=\quad
\parbox{40pt}{
\begin{fmfgraph*}(40,40)
      \fmfleft{v1}
      \fmfright{v2}
      \fmf{fermion}{v2,v1}
\end{fmfgraph*}}
\quad +\quad\parbox{50pt}{
    \begin{fmfgraph*}(60,40)
      \fmfleft{v1}
      \fmfright{v2}
      \fmfblob{.32w}{g1}
      \fmf{fermion}{v2,g1}
      \fmf{fermion}{g1,v1}
    \end{fmfgraph*}}
\quad+\quad\parbox{50pt}{
    \begin{fmfgraph*}(100,40)
       \fmfleft{v1}
       \fmfright{v3}
       \fmfblob{.2w}{g1,g2}
       \fmf{fermion}{v3,g2}
       \fmf{fermion}{g2,v2}
       \fmf{fermion}{v2,g1}
       \fmf{fermion}{g1,v1}
    \end{fmfgraph*}}
\end{fmffile}
\qquad\qquad\qquad
+
\cdots
\end{aligned}
\end{equation}
where the first, double-line-arrow diagram represents $G$, single-line-arrow diagram represents $G_0$, and the blob diagram corresponds to the self-energy $\Sigma$. This diagram can translate back to Dyson's equation for Green's function using Feynman rule \cite{stefanucci2013nonequilibrium}. Dyson equation indicates that the interactive Green's function $G$ can be constructed by inserting the self-energy $\Sigma$ into $G_0$ recursively and summing the first-order insertion, the second-order insertion, up to infinity. Hence, to solve for $G$ (at least formally), we only need to know the self-energy $\Sigma$. In many-body perturbation theory, $\Sigma$ can be calculated using another diagrammatic perturbation expansion. For instance, if the quantum many-body system only involves electron interactions, say the Hubbard model, then the diagrammatic perturbation expansion for $\Sigma$ can be written as \cite{stefanucci2013nonequilibrium}:   
\begin{equation}\label{Sigma_manybody}
\begin{aligned}
\Sigma&=\Sigma^1[G_0,v]+\Sigma^2[G_0,v]+\cdots\\[10pt]
\vspace{2cm}
&=
\begin{fmffile}{feyngraph1}
\parbox{30pt}{
    \begin{fmfgraph*}(30,30)
      \fmftop{v1}
      \fmfbottom{v3}
      \fmf{plain,left,tension=0.3}{v2,v1}
      \fmf{plain_arrow,left,tension=0.3}{v1,v2}
      \fmf{photon}{v2,v3}
    \end{fmfgraph*}}
    +\quad\parbox{30pt}{
    \begin{fmfgraph*}(30,30)
      \fmfleft{v1}
      \fmfright{v2}
      \fmf{fermion}{v2,v1}
      \fmf{photon,left,tension=0.2}{v1,v2}
    \end{fmfgraph*}}
\end{fmffile}
\quad
\Biggr\}
\Rightarrow O(v)
\\
&
\begin{fmffile}{feyngraph2}
+\quad\parbox{30pt}{
    \begin{fmfgraph*}(40,40)
      \fmfleft{v1}
      \fmfright{v4}
      \fmf{fermion}{v1,v2}
      \fmf{fermion}{v2,v3}
      \fmf{fermion}{v3,v4}
      \fmf{photon,left,tension=0.2}{v1,v3}
      \fmf{photon,right,tension=0.2}{v2,v4}
    \end{fmfgraph*}}
    \quad\quad+\quad\parbox{30pt}{
    \begin{fmfgraph*}(30,25)
      \fmftop{v1,v3}
      \fmfbottom{v2,v4}
      \fmf{photon}{v1,v2}
      \fmf{photon}{v3,v4}
      \fmf{plain_arrow,left=0.5,tension=0.5}{v1,v3,v1}
      \fmf{fermion}{v2,v4}
    \end{fmfgraph*}}
    \quad +\quad\parbox{30pt}{
    \begin{fmfgraph*}(40,50)
      \fmftop{v1}
      \fmfbottom{v5}
      \fmf{plain,left,tension=0.3}{v2,v1}
      \fmf{plain_arrow,left,tension=0.3}{v1,v2}
      \fmf{photon}{v2,v3}
      \fmf{plain_arrow,left,tension=0.3}{v3,v4}
      \fmf{plain_arrow,left,tension=0.3}{v4,v3}
      \fmf{photon}{v4,v5}
    \end{fmfgraph*}}
    \quad +\quad\parbox{30pt}{
    \begin{fmfgraph*}(30,30)
      \fmfleft{v1}
      \fmfright{v4}
      \fmf{fermion}{v1,v2}
      \fmf{fermion}{v2,v3}
      \fmf{fermion}{v3,v4}
      \fmf{photon,left,tension=0.2}{v1,v4}
      \fmf{photon,left,tension=0.05}{v2,v3}
    \end{fmfgraph*}}
    \quad +\quad\parbox{30pt}{
    \begin{fmfgraph*}(30,30)
      \fmftop{v1}
      \fmfleft{v4}
      \fmfright{v5}
      \fmfbottom{v4,v5}
      \fmf{plain,left,tension=0.3}{v2,v1}
      \fmf{plain_arrow,left,tension=0.3}{v1,v2}
      \fmf{photon}{v2,v3}
      \fmf{fermion}{v4,v3}
      \fmf{fermion}{v3,v5}
      \fmf{photon,right=0.5,tension=0.5}{v4,v5}
    \end{fmfgraph*}}
    \quad +\quad\parbox{30pt}{
    \begin{fmfgraph*}(30,50)
      \fmfleft{v1}
      \fmfright{v2}
      \fmfbottom{v4}
      \fmf{photon,left=0.7,tension=0.5}{v1,v2}
      \fmf{fermion}{v1,v2}
      \fmf{fermion}{v1,v3}
      \fmf{fermion}{v3,v2}
      \fmf{photon}{v3,v4}
    \end{fmfgraph*}}
\end{fmffile}
\quad
\Biggr\}\Rightarrow O(v^2)\\
&\cdots
\end{aligned}
\end{equation}
where the arrow lines and wiggly lines have their normal meanings in the Feynman diagram, and we refer to \cite{stefanucci2013nonequilibrium,mahan2013many} for technical details. \eqref{Sigma_manybody} is a perturbation series with respect to the modeling parameter $v$. When $v$ is small, higher-order terms (diagrams) may be truncated and we obtain a good approximation to the self-energy $\Sigma$ and hence the Green's function $G$. In electron systems, $v$ is normally chosen to be the magnitude of the interaction Hamiltonian which indicates the strength of scattering interactions between electrons. As a perturbation series, when $v$ is large, it is expected that \eqref{Sigma_manybody} no longer serves as a good starting point of approximation. To this end, a {\em renormalized perturbation technique}, which is often called as the skeleton expansion \cite{stefanucci2013nonequilibrium}, is often used to reconstruct expansion series based on \eqref{Sigma_manybody}. In essence, the skeleton expansion is just a resummation technique that groups diagrams and insert them back into themselves. For examples, the second diagram in \eqref{Sigma_manybody} can be renormalized as:
\begin{equation}\label{feyngraphresum_ex}
\begin{aligned}
\begin{fmffile}{feyngraphresum}
\parbox{30pt}{
    \begin{fmfgraph*}(30,40)
      \fmfleft{v1}
      \fmfright{v2}
      \fmf{dbl_plain_arrow}{v2,v1}
      \fmf{photon,left,tension=0.1}{v1,v2}
    \end{fmfgraph*}}
\quad =
\quad\parbox{30pt}{
    \begin{fmfgraph*}(30,30)
      \fmfleft{v1}
      \fmfright{v2}
      \fmf{fermion}{v2,v1}
      \fmf{photon,left,tension=0.2}{v1,v2}
    \end{fmfgraph*}}
    \quad+
\quad\parbox{30pt}{
    \begin{fmfgraph*}(30,30)
      \fmftop{v1}
      \fmfleft{v4}
      \fmfright{v5}
      \fmfbottom{v4,v5}
      \fmf{plain,left,tension=0.3}{v2,v1}
      \fmf{plain_arrow,left,tension=0.3}{v1,v2}
      \fmf{photon}{v2,v3}
      \fmf{fermion}{v4,v3}
      \fmf{fermion}{v3,v5}
      \fmf{photon,right=0.5,tension=0.5}{v4,v5}
    \end{fmfgraph*}}
\quad +
\quad\parbox{30pt}{
    \begin{fmfgraph*}(30,30)
      \fmfleft{v1}
      \fmfright{v4}
      \fmf{fermion}{v1,v2}
      \fmf{fermion}{v2,v3}
      \fmf{fermion}{v3,v4}
      \fmf{photon,left,tension=0.2}{v1,v4}
      \fmf{photon,left,tension=0.05}{v2,v3}
\end{fmfgraph*}}
\quad +\quad\parbox{30pt}{
    \begin{fmfgraph*}(40,40)
      \fmfleft{p1}
      \fmfright{p2}
      \fmf{fermion}{p1,v1}
      \fmf{fermion}{v1,v2}
      \fmf{fermion}{v2,v3}
      \fmf{fermion}{v3,v4}
      \fmf{fermion}{v4,p2}
      \fmf{photon,left,tension=0.2}{v1,v3}
      \fmf{photon,right,tension=0.2}{v2,v4}
      \fmf{photon,left,tension=0.1}{p1,p2}
    \end{fmfgraph*}}
\end{fmffile}
\quad
+\quad \cdots,
\end{aligned}
\end{equation}
where on the RHS of \eqref{feyngraphresum_ex}, we only explicitly show the first three terms of the first-order self-energy insertion. The summation actually goes over all possible self-energy insertions, thereby yielding the interactive Green's function $G$ on the LHS of \eqref{feyngraphresum_ex}. After the renormalization, only the {\em skeleton} diagrams which cannot be further decomposed into the combination of bare self-energy diagrams retain. Eventually, we get the skeleton expansion of the self-energy:
\begin{equation}\label{renormalized_sigma}
\begin{aligned}
\Sigma&=\Sigma^1_s[G,v]+\Sigma^2_s[G,v]+\cdots\\[10pt]
&=
\underbrace{
\begin{fmffile}{feyngraph3}
\quad\parbox{30pt}{
    \begin{fmfgraph*}(30,40)
      \fmftop{v1}
      \fmfbottom{v3}
      \fmf{dbl_plain,left,tension=0.3}{v2,v1}
      \fmf{dbl_plain_arrow,left,tension=0.3}{v1,v2}
      \fmf{photon}{v2,v3}
    \end{fmfgraph*}}
    \quad +\quad\parbox{30pt}{
    \begin{fmfgraph*}(30,40)
      \fmfleft{v1}
      \fmfright{v2}
      \fmf{dbl_plain_arrow}{v2,v1}
      \fmf{photon,left,tension=0.1}{v1,v2}
    \end{fmfgraph*}}
\end{fmffile}
}_{O(v)}
\underbrace{
\begin{fmffile}{feyngraph4}
    \quad+\quad\parbox{30pt}{
    \begin{fmfgraph*}(40,50)
      \fmfleft{v1}
      \fmfright{v4}
      \fmf{dbl_plain_arrow}{v1,v2}
      \fmf{dbl_plain_arrow}{v2,v3}
      \fmf{dbl_plain_arrow}{v3,v4}
      \fmf{photon,left,tension=0.1}{v1,v3}
      \fmf{photon,right,tension=0.1}{v2,v4}
    \end{fmfgraph*}}
    \qquad +\quad\parbox{30pt}{
    \begin{fmfgraph*}(30,25)
      \fmftop{v1,v3}
      \fmfbottom{v2,v4}
      \fmf{photon}{v1,v2}
      \fmf{photon}{v3,v4}
      \fmf{dbl_plain_arrow,left=0.5,tension=0.3}{v1,v3,v1}
      \fmf{dbl_plain_arrow}{v2,v4}
    \end{fmfgraph*}}
    \end{fmffile}
\quad +\quad \cdots
}_{O(v^2)}.
\end{aligned}
\end{equation}
Renormalized self-energy expansion \eqref{renormalized_sigma} now becomes self-consistent since it contains diagrams that only involve the interactive Green's function $G$. Formally, it still is a series expansion with respect to $v$, while from its construction, we know that low-order terms after the dressing replacement $G_0\rightarrow G$ actually contain infinite diagrams. This is the reason why even the low-order truncation of this series often leads to meaningful physical predictions for strongly interactive systems with large $v$. In comparison, we recall what we have done in the previous section for the tree diagrams and one can easily see the similarities. First, we note that for time-independent equilibrium systems, the tree representation of the self-energy $\hat\Sigma$ in the Mori-Zwanzig framework is just the exponential generating function $C_{\tP}(t)$ \eqref{CP(t)}, Writing more terms explicitly in the expansion series \eqref{CP(t)}, we have the tree diagram representation:
\begin{equation}\label{bare_expansion_CMZE_sigma}
\begin{aligned}
\hat\Sigma
&=\hat\Sigma^0[G(0),t]+\hat\Sigma^1[G(0),t]+\hat\Sigma^2[G(0),t]+\cdots \\
&=
-{\fontsize{1}{1}\selectfont
\begin{forest}
 for tree={grow=90, l=1mm}
[[[]]]
 \path[fill=black]  (!.parent anchor) circle[radius=2pt];
 \path[fill=black] (!1.child anchor) circle[radius=2pt]
                  (!11.child anchor) circle[radius=2pt];
\end{forest}
}
+
{\fontsize{1}{1}\selectfont
\begin{forest}
for tree={grow=90, l=1mm}
[[][]]
\path[fill=black]  (!.parent anchor) circle[radius=2pt];
\path[fill=black] (!1.child anchor) circle[radius=2pt]
                 (!2.child anchor) circle[radius=2pt];
\end{forest}
}
\quad\bigr\}\Rightarrow O(1)\\
&\ \
+t\biggl[
{\fontsize{1}{1}\selectfont
\begin{forest}
 for tree={grow=90, l=1mm}
[[[[]]]]
 \path[fill=black]  (!.parent anchor) circle[radius=2pt];
 \path[fill=black] (!1.child anchor) circle[radius=2pt]
                    (!11.child anchor) circle[radius=2pt]
                  (!111.child anchor) circle[radius=2pt];
\end{forest}
}
-
{\fontsize{1}{1}\selectfont
\begin{forest}
 for tree={grow=90, l=1mm}
[[[]][]]
 \path[fill=black]  (!.parent anchor) circle[radius=2pt];
 \path[fill=black] (!1.child anchor) circle[radius=2pt]
                  (!11.child anchor) circle[radius=2pt]
                  (!2.child anchor) circle[radius=2pt];
\end{forest}
}
-
{\fontsize{1}{1}\selectfont
\begin{forest}
 for tree={parent anchor=east, child anchor=east, grow=90, l=1mm}
[[[][]]]
 \path[fill=black]  (!.parent anchor) circle[radius=2pt];
 \path[fill=black] (!1.child anchor) circle[radius=2pt]
                  (!11.child anchor) circle[radius=2pt]
                  (!12.child anchor) circle[radius=2pt];
\end{forest}
}
+
{\fontsize{1}{1}\selectfont
\begin{forest}
 for tree={grow=90, l=1mm}
[[][][]]
 \path[fill=black]  (!.parent anchor) circle[radius=2pt];
 \path[fill=black] (!1.child anchor) circle[radius=2pt]
                  (!2.child anchor) circle[radius=2pt]
                  (!3.child anchor) circle[radius=2pt];
\end{forest}
}
\biggr]
\quad\Biggr\}\Rightarrow O(t)
\\
&\ \ \ \ \ \ \ \ \ 
+\frac{t^2}{2}
\Biggl[
-
{\fontsize{1}{1}\selectfont
\begin{forest}
 for tree={grow=90, l=1mm}
 [[[[[]]]]]
 \path[fill=black]  (!.parent anchor) circle[radius=2pt];
 \path[fill=black] (!1.child anchor) circle[radius=2pt]
                  (!11.child anchor) circle[radius=2pt]
                  (!111.child anchor) circle[radius=2pt]
                  (!1111.child anchor)  circle[radius=2pt];
\end{forest}
}
+
{\fontsize{1}{1}\selectfont
\begin{forest}
for tree={grow=90, l=1mm}
 [[[[][]]]]
 \path[fill=black]  (!.parent anchor) circle[radius=2pt];
 \path[fill=black] (!1.child anchor) circle[radius=2pt]
                  (!11.child anchor) circle[radius=2pt]
                  (!111.child anchor) circle[radius=2pt]
                  (!112.child anchor)  circle[radius=2pt];
\end{forest}
}
+
{\fontsize{1}{1}\selectfont
\begin{forest}
 for tree={grow=90, l=1mm}
[[[[]][]]]
 \path[fill=black]  (!.parent anchor) circle[radius=2pt];
 \path[fill=black] (!1.child anchor) circle[radius=2pt]
                  (!11.child anchor) circle[radius=2pt]
                  (!111.child anchor) circle[radius=2pt]
                  (!12.child anchor) circle[radius=2pt];
\end{forest}
}
+
{\fontsize{1}{1}\selectfont
\begin{forest}
 for tree={grow=90, l=1mm}
[[[[]]][]]
 \path[fill=black]  (!.parent anchor) circle[radius=2pt];
 \path[fill=black] (!1.child anchor) circle[radius=2pt]
                  (!11.child anchor) circle[radius=2pt]
                  (!111.child anchor) circle[radius=2pt]
                  (!2.child anchor) circle[radius=2pt];
\end{forest}
}
+
{\fontsize{1}{1}\selectfont
\begin{forest}
 for tree={grow=90, l=1mm}
 [ [[][]] []]
 \path[fill=black]  (!.parent anchor) circle[radius=2pt];
 \path[fill=black] (!1.child anchor) circle[radius=2pt]
                  (!11.child anchor) circle[radius=2pt]
                  (!12.child anchor) circle[radius=2pt]
                  (!2.child anchor)  circle[radius=2pt];
\end{forest}
}
-
{\fontsize{1}{1}\selectfont
\begin{forest}
 for tree={grow=90, l=1mm}
 [[[]] [] [] ]
 \path[fill=black]  (!.parent anchor) circle[radius=2pt];
 \path[fill=black] (!1.child anchor) circle[radius=2pt]
                  (!11.child anchor) circle[radius=2pt]
                  (!2.child anchor) circle[radius=2pt]
                  (!3.child anchor)  circle[radius=2pt];
\end{forest}
}
-
{\fontsize{1}{1}\selectfont
\begin{forest}
 for tree={grow=90, l=1mm}
 [ [[][][]] ]
 \path[fill=black]  (!.parent anchor) circle[radius=2pt];
 \path[fill=black] (!1.child anchor) circle[radius=2pt]
                  (!11.child anchor) circle[radius=2pt]
                  (!12.child anchor) circle[radius=2pt]
                  (!13.child anchor)  circle[radius=2pt];
\end{forest}
}
+
{\fontsize{1}{1}\selectfont
\begin{forest}
 for tree={grow=90, l=1mm}
 [ [][][] []]
 \path[fill=black]  (!.parent anchor) circle[radius=2pt];
 \path[fill=black] (!1.child anchor) circle[radius=2pt]
                  (!2.child anchor) circle[radius=2pt]
                  (!3.child anchor) circle[radius=2pt]
                  (!4.child anchor)  circle[radius=2pt];
\end{forest}
}
\Biggr]
\quad\Biggr\}\Rightarrow O(t^2)
\\
&\quad \cdots.
\end{aligned}
\end{equation}
This is a bare expansion series with respect to time $t$, and each diagram is a function of the initial condition $G(0)$ \footnote{In fact, $\hat{\Sigma}^i[G(0),v]$ depends on the initial condition of {\em many-particle} Green's function even when the $G(t)$ under the investigation is only a one-particle Green's function, we will explain this in details in Section \ref{sec:app_hubbard}. } By introducing the combinatorial expansion for the self-energy and solving the functional equation \eqref{CMZE-DSE}, we find that the above diagrams can be decomposed and re-summarized as:
\begin{equation}\label{Sigma_tree_CMZE}
\begin{aligned}
\hat\Sigma
&=\hat\Sigma^0[G(0),t]+\hat\Sigma^1[G(0),t]+\hat\Sigma^2[G(0),t]+\cdots\\
&=
-{\fontsize{1}{1}\selectfont
\begin{forest}
 for tree={grow=90, l=1mm}
[[[]]]
 \path[fill=black]  (!.parent anchor) circle[radius=2pt];
 \path[fill=black] (!1.child anchor) circle[radius=2pt]
                  (!11.child anchor) circle[radius=2pt];
\end{forest}
}
+
{\fontsize{1}{1}\selectfont
\begin{forest}
for tree={grow=90, l=1mm}
[[][]]
\path[fill=black]  (!.parent anchor) circle[radius=2pt];
\path[fill=black] (!1.child anchor) circle[radius=2pt]
                 (!2.child anchor) circle[radius=2pt];
\end{forest}
}
\quad\bigr\}\Rightarrow O(1)\\
&\ 
+t\biggl[
{\fontsize{1}{1}\selectfont
\begin{forest}
 for tree={grow=90, l=1mm}
[[[]]]
 \path[fill=black]  (!.parent anchor) circle[radius=2pt]
                    (!1.child anchor) circle[radius=2pt]
                    (!11.child anchor) circle[radius=2pt];
\end{forest}
}
-
{\fontsize{1}{1}\selectfont
\begin{forest}
 for tree={grow=90, l=1mm}
[[[]][]]
 \path[fill=black]  (!.parent anchor) circle[radius=2pt]
                    (!1.child anchor) circle[radius=2pt]
                    (!11.child anchor) circle[radius=2pt]
                    (!2.child anchor) circle[radius=2pt];
\end{forest}
}
\lgraf
{\fontsize{1}{1}\selectfont
\begin{forest}
 for tree={grow=90, l=1mm}
[[]]
 \path[draw]  (!.parent anchor) circle[radius=2pt]
              (!1.child anchor) circle[radius=2pt];
\end{forest}
}
-{\fontsize{1}{1}\selectfont
\begin{forest}
 for tree={grow=90, l=1mm}
[[][]]
 \path[fill=black]  (!.parent anchor) circle[radius=2pt]
                    (!1.child anchor) circle[radius=2pt]
                    (!2.child anchor) circle[radius=2pt];
\end{forest}
}
+
{\fontsize{1}{1}\selectfont
\begin{forest}
 for tree={grow=90, l=1mm}
[[][][]]
 \path[fill=black]  (!.parent anchor) circle[radius=2pt]
                    (!1.child anchor) circle[radius=2pt]
                    (!2.child anchor) circle[radius=2pt]
                    (!3.child anchor) circle[radius=2pt];
\end{forest}
}
\lgraf
{\fontsize{1}{1}\selectfont
\begin{forest}
 for tree={grow=90, l=1mm}
[[]]
 \path[draw]  (!.parent anchor) circle[radius=2pt]
              (!1.child anchor) circle[radius=2pt];
\end{forest}
}
\biggr]\lgraf
{\fontsize{1}{1}\selectfont
\begin{forest}
 for tree={grow=90, l=1mm}
[[]]
 \path[fill=black]  (!.parent anchor) circle[radius=2pt]
              (!1.child anchor) circle[radius=2pt];
\end{forest}
}
\quad\Biggr\}\Rightarrow O(t)
\\
&\ \ \
\begin{rcases}
&\ \ \
+\frac{t^2}{2}\biggl[
{\fontsize{1}{1}\selectfont
\begin{forest}
 for tree={grow=90, l=1mm}
[[[]]]
 \path[fill=black]  (!.parent anchor) circle[radius=2pt]
                    (!1.child anchor) circle[radius=2pt]
                    (!11.child anchor) circle[radius=2pt];
\end{forest}
}
-
{\fontsize{1}{1}\selectfont
\begin{forest}
 for tree={grow=90, l=1mm}
[[[]][]]
 \path[fill=black]  (!.parent anchor) circle[radius=2pt]
                    (!1.child anchor) circle[radius=2pt]
                    (!11.child anchor) circle[radius=2pt]
                    (!2.child anchor) circle[radius=2pt];
\end{forest}
}
\lgraf
{\fontsize{1}{1}\selectfont
\begin{forest}
 for tree={grow=90, l=1mm}
[[]]
 \path[draw]  (!.parent anchor) circle[radius=2pt]
              (!1.child anchor) circle[radius=2pt];
\end{forest}
}
-{\fontsize{1}{1}\selectfont
\begin{forest}
 for tree={grow=90, l=1mm}
[[][]]
 \path[fill=black]  (!.parent anchor) circle[radius=2pt]
                    (!1.child anchor) circle[radius=2pt]
                    (!2.child anchor) circle[radius=2pt];
\end{forest}
}
+
{\fontsize{1}{1}\selectfont
\begin{forest}
 for tree={grow=90, l=1mm}
[[][][]]
 \path[fill=black]  (!.parent anchor) circle[radius=2pt]
                    (!1.child anchor) circle[radius=2pt]
                    (!2.child anchor) circle[radius=2pt]
                    (!3.child anchor) circle[radius=2pt];
\end{forest}
}
\lgraf
{\fontsize{1}{1}\selectfont
\begin{forest}
 for tree={grow=90, l=1mm}
[[]]
 \path[draw]  (!.parent anchor) circle[radius=2pt]
              (!1.child anchor) circle[radius=2pt];
\end{forest}
}
\biggr]
\lgraf
{\fontsize{1}{1}\selectfont
\begin{forest}
 for tree={grow=90, l=1mm}
[[][]]
 \path[fill=black]  (!.parent anchor) circle[radius=2pt]
                    (!1.child anchor) circle[radius=2pt]
                    (!2.child anchor) circle[radius=2pt];
\end{forest}
}
\\
&\ \ \ 
+\frac{t^2}{2}\biggl[{\fontsize{1}{1}\selectfont
\begin{forest}
 for tree={grow=90, l=1mm}
[[[]][]]
 \path[fill=black]  (!.parent anchor) circle[radius=2pt]
                    (!1.child anchor) circle[radius=2pt]
                    (!11.child anchor) circle[radius=2pt]
                    (!2.child anchor) circle[radius=2pt];
\end{forest}
}
\lgraf
{\fontsize{1}{1}\selectfont
\begin{forest}
 for tree={grow=90, l=1mm}
[[]]
 \path[draw]  (!.parent anchor) circle[radius=2pt]
              (!1.child anchor) circle[radius=2pt];
\end{forest}
}
\lgraf
{\fontsize{1}{1}\selectfont
\begin{forest}
 for tree={grow=90, l=1mm}
[[][]]
 \path[fill=black]  (!.parent anchor) circle[radius=2pt]
                    (!1.child anchor) circle[radius=2pt]
                    (!2.child anchor) circle[radius=2pt];
\end{forest}
}
\lgraf
{\fontsize{1}{1}\selectfont
\begin{forest}
 for tree={grow=90, l=1mm}
[[[]]]
 \path[draw]  (!.parent anchor) circle[radius=2pt]
              (!1.child anchor) circle[radius=2pt]
              (!11.child anchor) circle[radius=2pt];
\end{forest}
}
+
{\fontsize{1}{1}\selectfont
\begin{forest}
 for tree={grow=90, l=1mm}
[[][]]
 \path[fill=black]  (!.parent anchor) circle[radius=2pt]
                    (!1.child anchor) circle[radius=2pt]
                    (!2.child anchor) circle[radius=2pt];
\end{forest}
}
-
{\fontsize{1}{1}\selectfont
\begin{forest}
 for tree={grow=90, l=1mm}
[[][][]]
 \path[fill=black]  (!.parent anchor) circle[radius=2pt]
                    (!1.child anchor) circle[radius=2pt]
                    (!2.child anchor) circle[radius=2pt]
                    (!3.child anchor) circle[radius=2pt];
\end{forest}
}
\lgraf
{\fontsize{1}{1}\selectfont
\begin{forest}
 for tree={grow=90, l=1mm}
[[]]
 \path[draw]  (!.parent anchor) circle[radius=2pt]
              (!1.child anchor) circle[radius=2pt];
\end{forest}
}
\lgraf
{\fontsize{1}{1}\selectfont
\begin{forest}
 for tree={grow=90, l=1mm}
[[][]]
 \path[fill=black]  (!.parent anchor) circle[radius=2pt]
                    (!1.child anchor) circle[radius=2pt]
                    (!2.child anchor) circle[radius=2pt];
\end{forest}
}
\lgraf
{\fontsize{1}{1}\selectfont
\begin{forest}
 for tree={grow=90, l=1mm}
[[[]]]
 \path[draw]  (!.parent anchor) circle[radius=2pt]
              (!1.child anchor) circle[radius=2pt]
              (!11.child anchor) circle[radius=2pt];
\end{forest}
}\\
&\ \ \ \ \ \ \ \ \ \ \ \ \ \
-{\fontsize{1}{1}\selectfont
\begin{forest}
 for tree={grow=90, l=1mm}
[[[]]]
 \path[fill=black]  (!.parent anchor) circle[radius=2pt]
                    (!1.child anchor) circle[radius=2pt]
                    (!11.child anchor) circle[radius=2pt];
\end{forest}
}
+
{\fontsize{1}{1}\selectfont
\begin{forest}
 for tree={grow=90, l=1mm}
[[[]][]]
 \path[fill=black]  (!.parent anchor) circle[radius=2pt]
                    (!1.child anchor) circle[radius=2pt]
                    (!11.child anchor) circle[radius=2pt]
                    (!2.child anchor) circle[radius=2pt];
\end{forest}
}
\lgraf
{\fontsize{1}{1}\selectfont
\begin{forest}
 for tree={grow=90, l=1mm}
[[]]
 \path[draw]  (!.parent anchor) circle[radius=2pt]
              (!1.child anchor) circle[radius=2pt];
\end{forest}
}
-
{\fontsize{1}{1}\selectfont
\begin{forest}
 for tree={grow=90, l=1mm}
[[[]][][]]
 \path[fill=black]  (!.parent anchor) circle[radius=2pt]
                    (!1.child anchor) circle[radius=2pt]
                    (!11.child anchor) circle[radius=2pt]
                    (!2.child anchor) circle[radius=2pt]
                    (!3.child anchor) circle[radius=2pt];
\end{forest}
}
\lgraf
{\fontsize{1}{1}\selectfont
\begin{forest}
 for tree={grow=90, l=1mm}
[[[]]]
 \path[draw]  (!.parent anchor) circle[radius=2pt]
              (!1.child anchor) circle[radius=2pt]
              (!11.child anchor) circle[radius=2pt];
\end{forest}
}
-
{\fontsize{1}{1}\selectfont
\begin{forest}
 for tree={grow=90, l=1mm}
[[][][]]
 \path[fill=black]  (!.parent anchor) circle[radius=2pt]
                    (!1.child anchor) circle[radius=2pt]
                    (!2.child anchor) circle[radius=2pt]
                    (!3.child anchor) circle[radius=2pt];
\end{forest}
}
\lgraf
{\fontsize{1}{1}\selectfont
\begin{forest}
 for tree={grow=90, l=1mm}
[[]]
 \path[draw]  (!.parent anchor) circle[radius=2pt]
              (!1.child anchor) circle[radius=2pt];
\end{forest}
}
+
{\fontsize{1}{1}\selectfont
\begin{forest}
 for tree={grow=90, l=1mm}
[[][][][]]
 \path[fill=black]  (!.parent anchor) circle[radius=2pt]
                    (!1.child anchor) circle[radius=2pt]
                    (!2.child anchor) circle[radius=2pt]
                    (!3.child anchor) circle[radius=2pt]
                    (!4.child anchor) circle[radius=2pt];
\end{forest}
}
\lgraf
{\fontsize{1}{1}\selectfont
\begin{forest}
 for tree={grow=90, l=1mm}
[[[]]]
 \path[draw]  (!.parent anchor) circle[radius=2pt]
              (!1.child anchor) circle[radius=2pt]
              (!11.child anchor) circle[radius=2pt];
\end{forest}
}
\biggr]
\lgraf
{\fontsize{1}{1}\selectfont
\begin{forest}
 for tree={grow=90, l=1mm}
[[[]]]
 \path[fill=black]  (!.parent anchor) circle[radius=2pt]
                    (!1.child anchor) circle[radius=2pt]
                    (!11.child anchor) circle[radius=2pt];
\end{forest}
}\quad
\end{rcases}
\quad\Rightarrow O(t^2)\\
&\quad \cdots
\end{aligned}
\end{equation}
Grouping terms on the RHS of \eqref{Sigma_tree_CMZE}, eventually we obtain the following self-consistent expansion for the Mori-Zwanzig self-energy $\hat{\Sigma}$:
\begin{equation}\label{Sigma_tree_renormalized}
\begin{aligned}
\hat\Sigma
&=\hat\Sigma_s^0[G(t),G(0)]+\hat\Sigma_s^1[G(t),G(0)]+\hat\Sigma_s^2[G(t),G(0)]+\cdots\\
&=
-{\fontsize{1}{1}\selectfont
\begin{forest}
 for tree={grow=90, l=1mm}
[[[]]]
 \path[fill=black]  (!.parent anchor) circle[radius=2pt];
 \path[fill=black] (!1.child anchor) circle[radius=2pt]
                  (!11.child anchor) circle[radius=2pt];
\end{forest}
}
+
{\fontsize{1}{1}\selectfont
\begin{forest}
for tree={grow=90, l=1mm}
[[][]]
\path[fill=black]  (!.parent anchor) circle[radius=2pt];
\path[fill=black] (!1.child anchor) circle[radius=2pt]
                 (!2.child anchor) circle[radius=2pt];
\end{forest}
}\\
&+
\biggl[
{\fontsize{1}{1}\selectfont
\begin{forest}
 for tree={grow=90, l=1mm}
[[[]]]
 \path[fill=black]  (!.parent anchor) circle[radius=2pt]
                    (!1.child anchor) circle[radius=2pt]
                    (!11.child anchor) circle[radius=2pt];
\end{forest}
}
-
{\fontsize{1}{1}\selectfont
\begin{forest}
 for tree={grow=90, l=1mm}
[[[]][]]
 \path[fill=black]  (!.parent anchor) circle[radius=2pt]
                    (!1.child anchor) circle[radius=2pt]
                    (!11.child anchor) circle[radius=2pt]
                    (!2.child anchor) circle[radius=2pt];
\end{forest}
}
\lgraf
{\fontsize{1}{1}\selectfont
\begin{forest}
 for tree={grow=90, l=1mm}
[[]]
 \path[draw]  (!.parent anchor) circle[radius=2pt]
              (!1.child anchor) circle[radius=2pt];
\end{forest}
}
-{\fontsize{1}{1}\selectfont
\begin{forest}
 for tree={grow=90, l=1mm}
[[][]]
 \path[fill=black]  (!.parent anchor) circle[radius=2pt]
                    (!1.child anchor) circle[radius=2pt]
                    (!2.child anchor) circle[radius=2pt];
\end{forest}
}
+
{\fontsize{1}{1}\selectfont
\begin{forest}
 for tree={grow=90, l=1mm}
[[][][]]
 \path[fill=black]  (!.parent anchor) circle[radius=2pt]
                    (!1.child anchor) circle[radius=2pt]
                    (!2.child anchor) circle[radius=2pt]
                    (!3.child anchor) circle[radius=2pt];
\end{forest}
}
\lgraf
{\fontsize{1}{1}\selectfont
\begin{forest}
 for tree={grow=90, l=1mm}
[[]]
 \path[draw]  (!.parent anchor) circle[radius=2pt]
              (!1.child anchor) circle[radius=2pt];
\end{forest}
}
\biggr]\lgraf
\biggl[t
{\fontsize{1}{1}\selectfont
\begin{forest}
 for tree={grow=90, l=1mm}
[[]]
 \path[fill=black]  (!.parent anchor) circle[radius=2pt]
              (!1.child anchor) circle[radius=2pt];
\end{forest}
}
+\frac{t^2}{2}
{\fontsize{1}{1}\selectfont
\begin{forest}
 for tree={grow=90, l=1mm}
[[][]]
 \path[fill=black]  (!.parent anchor) circle[radius=2pt]
              (!1.child anchor) circle[radius=2pt]
              (!2.child anchor) circle[radius=2pt];
\end{forest}
}
+\cdots
\biggr]
\\
&\ \ \
+\frac{1}{2}\biggl[
{\fontsize{1}{1}\selectfont
\begin{forest}
 for tree={grow=90, l=1mm}
[[[]][]]
 \path[fill=black]  (!.parent anchor) circle[radius=2pt]
                    (!1.child anchor) circle[radius=2pt]
                    (!11.child anchor) circle[radius=2pt]
                    (!2.child anchor) circle[radius=2pt];
\end{forest}
}
\lgraf
{\fontsize{1}{1}\selectfont
\begin{forest}
 for tree={grow=90, l=1mm}
[[]]
 \path[draw]  (!.parent anchor) circle[radius=2pt]
              (!1.child anchor) circle[radius=2pt];
\end{forest}
}
\lgraf
{\fontsize{1}{1}\selectfont
\begin{forest}
 for tree={grow=90, l=1mm}
[[][]]
 \path[fill=black]  (!.parent anchor) circle[radius=2pt]
                    (!1.child anchor) circle[radius=2pt]
                    (!2.child anchor) circle[radius=2pt];
\end{forest}
}
\lgraf
{\fontsize{1}{1}\selectfont
\begin{forest}
 for tree={grow=90, l=1mm}
[[[]]]
 \path[draw]  (!.parent anchor) circle[radius=2pt]
              (!1.child anchor) circle[radius=2pt]
              (!11.child anchor) circle[radius=2pt];
\end{forest}
}
+
{\fontsize{1}{1}\selectfont
\begin{forest}
 for tree={grow=90, l=1mm}
[[][]]
 \path[fill=black]  (!.parent anchor) circle[radius=2pt]
                    (!1.child anchor) circle[radius=2pt]
                    (!2.child anchor) circle[radius=2pt];
\end{forest}
}
-
{\fontsize{1}{1}\selectfont
\begin{forest}
 for tree={grow=90, l=1mm}
[[][][]]
 \path[fill=black]  (!.parent anchor) circle[radius=2pt]
                    (!1.child anchor) circle[radius=2pt]
                    (!2.child anchor) circle[radius=2pt]
                    (!3.child anchor) circle[radius=2pt];
\end{forest}
}
\lgraf
{\fontsize{1}{1}\selectfont
\begin{forest}
 for tree={grow=90, l=1mm}
[[]]
 \path[draw]  (!.parent anchor) circle[radius=2pt]
              (!1.child anchor) circle[radius=2pt];
\end{forest}
}
\lgraf
{\fontsize{1}{1}\selectfont
\begin{forest}
 for tree={grow=90, l=1mm}
[[][]]
 \path[fill=black]  (!.parent anchor) circle[radius=2pt]
                    (!1.child anchor) circle[radius=2pt]
                    (!2.child anchor) circle[radius=2pt];
\end{forest}
}
\lgraf
{\fontsize{1}{1}\selectfont
\begin{forest}
 for tree={grow=90, l=1mm}
[[[]]]
 \path[draw]  (!.parent anchor) circle[radius=2pt]
              (!1.child anchor) circle[radius=2pt]
              (!11.child anchor) circle[radius=2pt];
\end{forest}
}
\\
&\ \ \ \ \ \ \ \ \ \ \ \ \ \
-
{\fontsize{1}{1}\selectfont
\begin{forest}
 for tree={grow=90, l=1mm}
[[[]]]
 \path[fill=black]  (!.parent anchor) circle[radius=2pt]
                    (!1.child anchor) circle[radius=2pt]
                    (!11.child anchor) circle[radius=2pt];
\end{forest}
}
+
{\fontsize{1}{1}\selectfont
\begin{forest}
 for tree={grow=90, l=1mm}
[[[]][]]
 \path[fill=black]  (!.parent anchor) circle[radius=2pt]
                    (!1.child anchor) circle[radius=2pt]
                    (!11.child anchor) circle[radius=2pt]
                    (!2.child anchor) circle[radius=2pt];
\end{forest}
}
\lgraf
{\fontsize{1}{1}\selectfont
\begin{forest}
 for tree={grow=90, l=1mm}
[[]]
 \path[draw]  (!.parent anchor) circle[radius=2pt]
              (!1.child anchor) circle[radius=2pt];
\end{forest}
}
-
{\fontsize{1}{1}\selectfont
\begin{forest}
 for tree={grow=90, l=1mm}
[[[]][][]]
 \path[fill=black]  (!.parent anchor) circle[radius=2pt]
                    (!1.child anchor) circle[radius=2pt]
                    (!11.child anchor) circle[radius=2pt]
                    (!2.child anchor) circle[radius=2pt]
                    (!3.child anchor) circle[radius=2pt];
\end{forest}
}
\lgraf
{\fontsize{1}{1}\selectfont
\begin{forest}
 for tree={grow=90, l=1mm}
[[[]]]
 \path[draw]  (!.parent anchor) circle[radius=2pt]
              (!1.child anchor) circle[radius=2pt]
              (!11.child anchor) circle[radius=2pt];
\end{forest}
}
-
{\fontsize{1}{1}\selectfont
\begin{forest}
 for tree={grow=90, l=1mm}
[[][][]]
 \path[fill=black]  (!.parent anchor) circle[radius=2pt]
                    (!1.child anchor) circle[radius=2pt]
                    (!2.child anchor) circle[radius=2pt]
                    (!3.child anchor) circle[radius=2pt];
\end{forest}
}
\lgraf
{\fontsize{1}{1}\selectfont
\begin{forest}
 for tree={grow=90, l=1mm}
[[]]
 \path[draw]  (!.parent anchor) circle[radius=2pt]
              (!1.child anchor) circle[radius=2pt];
\end{forest}
}
+
{\fontsize{1}{1}\selectfont
\begin{forest}
 for tree={grow=90, l=1mm}
[[][][][]]
 \path[fill=black]  (!.parent anchor) circle[radius=2pt]
                    (!1.child anchor) circle[radius=2pt]
                    (!2.child anchor) circle[radius=2pt]
                    (!3.child anchor) circle[radius=2pt]
                    (!4.child anchor) circle[radius=2pt];
\end{forest}
}
\lgraf
{\fontsize{1}{1}\selectfont
\begin{forest}
 for tree={grow=90, l=1mm}
[[[]]]
 \path[draw]  (!.parent anchor) circle[radius=2pt]
              (!1.child anchor) circle[radius=2pt]
              (!11.child anchor) circle[radius=2pt];
\end{forest}
}
\biggr]
\lgraf
\biggl[t
{\fontsize{1}{1}\selectfont
\begin{forest}
 for tree={grow=90, l=1mm}
[[]]
 \path[fill=black]  (!.parent anchor) circle[radius=2pt]
              (!1.child anchor) circle[radius=2pt];
\end{forest}
}
+\frac{t^2}{2}
{\fontsize{1}{1}\selectfont
\begin{forest}
 for tree={grow=90, l=1mm}
[[][]]
 \path[fill=black]  (!.parent anchor) circle[radius=2pt]
              (!1.child anchor) circle[radius=2pt]
              (!2.child anchor) circle[radius=2pt];
\end{forest}
}
+\cdots
\biggr]^2
\\
&\quad \cdots
\end{aligned}
\end{equation}
To be noticed that the above series expansion is just $F(C(t))=\sum_{k=0}^{\infty}F_k\lgraf \frac{C^k(t)}{k!}$ which equals to $C_{\tP}(t)$ according to the functional equation \eqref{CMZE-DSE}. In conclusion, the bare expansion \eqref{Sigma_manybody} of the many-body self-energy $\Sigma$ is similar to the bare expansion \eqref{bare_expansion_CMZE_sigma} of the Mori-Zwanzig self-energy $\hat{\Sigma}$. The skeleton expansion is a resummation of the bare expansion \eqref{Sigma_manybody} which leads to a self-consistent expansion \eqref{renormalized_sigma} involving only the interactive Green's function $G(t)$. In comparison, the combinatorial expansion is a resummation of the bare expansion \eqref{bare_expansion_CMZE_sigma} which leads to a self-consistent expansion \eqref{Sigma_tree_renormalized} involving only $G(t)$. For time-dependent nonequilibrium systems, similar connections can be established. However, this part will not be detailed in this paper. 
\paragraph{Remark} As mentioned in the introduction, the skeleton diagrammatic expansion was widely used in the physics community to construct efficient numerical approximation schemes to Dyson's equation for strongly interactive many-body systems. However, as recently pointed out by Lin and Lindsey \cite{lin2021bold,lin2021bold2} among many others, the very existence of such a skeleton expansion for the {\em fermionic} system is still under debate since the resummation technique used in deriving \eqref{renormalized_sigma} is only a persuasive combinatorial {\em argument}, and not yet a rigorous combinatorial {\em proof}. One may need complicated combinatorial and analytical techniques, like the ones used in \cite{lin2021bold,lin2021bold2}, to rigorously justify such a renormalization procedure. In contrast, since essentially the CMZE is derived using a series of recurrence relations which were proved rigorously, the existence of the combinatorial (renormalized) expansion \eqref{Sigma_tree_renormalized} is theoretically guaranteed by Lemma \ref{lemma1} through homomorphisms that map words into trees and operators.
\subsection{A short summary}
Up to this point, we have laid down the theoretical foundation of the combinatorial Mori-Zwanzig theory. Before we enter the next section which mainly addresses different applications of the theory, it is worthy to provide a short summary on what we have done and comment on  the motivation of the theory, the methodology we have adopted, possible computational strategies, and further development of the theory. As we mentioned in the introduction, the primary goal of developing such a combinatorial expansion of the Mori-Zwanzig equation is to find a systematic way to derive a self-consistent equation of motion for the correlation (Green's) function of the many-body system. Ideally, the derivation has to be exact, directly comes from the modeling Hamiltonian, and does not rely on approximations other than the truncation of the expansion series. The simple case studies provided in Section \ref{sec:main_thm} and \ref{sec:Time-d-CMZE} have shown that the derived CMZE satisfies all these criteria. From this perspective, CMZE resembles the combinatorial Dyson-Schwinger equation (DSE) for quantum electrodynamics (QED) \cite{yeats2017combinatorial,yeats2008growth} and the skeleton expansions for quantum many-body systems \cite{stefanucci2013nonequilibrium} since the latter two are also formally exact for quantum fields with arbitrarily large interaction strength. In fact, the resemblance between functional equations \eqref{CMZE-DSE}, \eqref{CMZE-DSE_noneq} and \eqref{DSE}, as well as the similarities between tree diagram and Feynman diagrams clearly illustrates such connections. 

On the other hand, the difference is also worth noticing. First, as we already pointed out, CMZE is a temporal-domain renormalized perturbation theory, while the combinatorial DSE for QED \cite{yeats2008growth} and the skeleton expansions are perturbation theory renormalized with respect to parameters such as the interaction strength. Eventually, this is due to the fact that MBPT is built upon the Dyson series expansion in the interaction picture, while the CMZE is built on the Taylor series expansion in the Heisenberg picture.
Secondly, the methodology is different. The Feynman-graph-based method can be viewed as a jigsaw puzzle game where the first step of solving the puzzle is to find skeleton diagrams and the second step is to follow the combinatorial rule to assemble them together. The CMZE, on the other hand, sets up an ansatz claiming that two jigsaw puzzles are equivalent (e.g. $C_{\tP}(t)=F(C(t))$), and one finds the combinatorial rule to dissemble the first puzzle into pieces and assemble them together to get the other one (i.e. finding $F$). As we will see, the basic building block of these two puzzles in CMZE are trees with branches corresponding to $\P\L^i\P$, which can be further specified to be statistical moments, or equivalently, the initial condition of many-body Green's function, at the statistical equilibrium (see Section \ref{sec:app_mod_coup} and \ref{sec:app_hubbard}). Thirdly, CMZE is more flexible as a computational framework for two reasons. On one hand, since the projection operator $\P$ can be any finite-rank projection operator, we can define CMZE-induced endomorphism $\text{Ran}(\P)\rightarrow\text{Ran}(\P)$ in submanifold $\text{Ran}(\P)$, where $\text{dim}(\text{Ran}(\P))$ can be an arbitrary positive integer (see examples in Section \ref{sec:app_hubbard}). With small $\text{dim}(\text{Ran}(\P))$, one may be able to compute the combinatorial expansion for $\hat{\Sigma}$ to high orders with relatively low cost, which exceeds the commonly used Born or GW approximation to the self-energy $\Sigma$. On the other hand, as pointed out by P.Fulde \cite{fulde1995electron}, the MZ framework is {\em free-Wick's theorem}. We remind that Wick's theorem was not used when deriving the tree diagrams. Hence the CMZE equation applies to virtually all lattice fields following all kinds of algebra. This includes the statistical fields, random fields, and fermionic fields that will be considered in Section \ref{sec:app}, and in principle should include other lattice fields such as the Ising model, quantum spin chains \cite{mahan2013many} and the $t-J$ model \cite{fulde1995electron}. Moreover, the resummation technique we briefly sketched at the end of Section \ref{sec:Time-d-CMZE} hinted the existence of more variants of the self-consistent EOMs for the correlation (Green's) function. 

In conclusion, CMZE introduces a {\em family} of new self-consistent EOMs for the correlation (Green's) function. Eventually, their usefulness has to be evaluated based on numerical results. Here we only briefly comment on the computational strategy and leave the specific implementation as our future work. We first note that all the combinatorial expansions can be obtained by solving recurrence equations. In principle, this part can be done automatically using symbolic code. The building blocks $\P\L^i\P$ correspond to the equilibrium statistical moments, which can also be calculated using an enumerative combinatorial algorithm we developed in \cite{zhu2020generalized,zhu2021effective}. Lastly, one may already notice that for the time-dependent system, our CMZE \eqref{op_id_1_nonc_type2} is an {\em incomplete} combinatorial expansion since $\F_n(s)$ and $\G_m(t)$ depend on $\Q\L_1,\Q\L_2,\cdots$. Complete combinatorial anatomy of the time-dependent MZE should further decompose them into operator polynomials involving $\P\L_1\P,\P\L_2\P,\cdots,\P\L_1^2\P,\cdots$. We will discuss this extension and the symbolic code realization of the combinatorial expansion in the coming work \cite{zhu2022complete}. 
\section{Physics applications}
\label{sec:app}
In this section, we focus on the application of CMZE in different many-body systems, especially for those strongly interactive or correlated ones where perturbation theory generally fails. We will consider three applications, which exemplify the usage of CMZE in the classical Hamiltonian system, stochastic dynamical systems, and quantum many-body systems. More specifically, CMZE will be used to derive a self-consistent evolution equation for the density fluctuation in liquids. This equation can be viewed as a generalized mode-coupling equation and may be useful in predicting the glass transition in supercooled liquids. The application to a stochastic particle system yields novel expansion ansatzes of the memory kernel for the generalized Langevin equation of colloidal particles. Lastly, we apply the CMZE in electron systems and get new evolution equations of the one-particle Green's function that are different from the Kadanoff-Baym equation. Numerical studies of these new equations will be our future work. 

\subsection{Generalized mode-coupling equation for glassy dynamics}
\label{sec:app_mod_coup}
Consider a classical Hamiltonian system $\H$ for $N$ identical interactive particles with mass $m$, positions $\{r_i\}_{i=1}^N$ and momenta $\{p_i\}_{i=1}^N$. The Hamiltonian system equation of motion for any observable function $A=A(r_i,p_i)$ is given by:
\begin{align*}
    \frac{d}{dt}A(t)=\{A(t),\H\}=i\L A(t),
\end{align*}
where $\{\cdot,\cdot\}$ the classical Poisson bracket and $\L$ is known as the Liouville operator. Note that, to be consistent with the literature \cite{reichman2005mode}, this definition of the Liouville operator is slightly different from what we have seen in Section \ref{sec:CCMZE}. As a result, we can derive MZEs using semigroups $e^{it\L}$ and $e^{it\Q\L}$ with loss of generality. For liquid systems, the Hamiltonian often takes the form of 
\begin{align}\label{glass_hamiltonian}
    \H=\sum_i\frac{p_i^2}{2m}+\frac{1}{2}\sum_{i,j\neq i}\phi(r_{ij}),
\end{align}
where $\phi$ is the pairwise particle interaction potential and $r_{ij}=r_i-r_j$ defines the relative displacement between $i$-th and $j$-th particle. Liquid particles often have non-harmonic interactive potential energy whose magnitude is comparable with the kinetic energy. Hence the defined Hamiltonian system is a classical strongly interactive many-body system that cannot be solved analytically. When rapidly cooling a liquid such as water below its melting point, particles in the liquid will enter a glass state which looks like being "trapped" by surrounding particles \cite{gotze2009complex}. Trying to understand this phenomenon on the microscopic scale posed long-during challenges to physicists. Although still not perfect, the mode-coupling theory and its various generalization gained much success in terms of the prediction of this glass transition. Detailed discussions can be found in the physics literature \cite{gotze2009complex}. In this subsection, we will follow the derivation in \cite{reichman2005mode} to derive the classical MZE for the fluid density fluctuation, and then show that the combinatorial expansion of it will lead to EOM that is similar to but different from any existing (generalized) mode-coupling equations. 

To study the glass transition, we choose an observable pair $\{\delta\rho_{q},j_{q}^L\}$ from the Hamiltonian system \eqref{glass_hamiltonian}. Here $\delta\rho_{q}$ is the $q$-th Fourier mode of the fluid density fluctuation defined by:
\begin{align*}
   \delta\rho_{q}:=\sum_{i}e^{iq\cdot r}-(2\pi)^3\rho\delta(q),
\end{align*}
where $\rho=N/V$ is the average density of the system. $j_{q}^L$ is the $q$-th Fourier mode of the longitudinal current. Assuming the spatial homogeneity of the liquid, it is related to the time derivative of  $\delta\rho_{q}$ by 
\begin{align*}
   \dot\delta\rho_{q}=i|q|j_{q}^L=\frac{i|q|}{m}\sum_i(\hat q\cdot p_i)e^{-q\cdot r_i},
\end{align*}
where $\cdot$ denotes the vector inner product, and $\hat q$ is the unified $q$ vector. The indicator for the glass transition is the dynamic behavior of the density fluctuation. As the system temperature decreases, the equilibrium time autocorrelation function for the fluid density fluctuation will not decrease to 0 exponentially but wander around a non-zero plateau for a rather long time, which is a clear sign of non-ergodicity \cite{reichman2005mode,gotze2009complex}. To quantitatively study this dynamic process, we use the Mori-Zwanzig equation and introduce the following Mori-type projection operator:
\begin{align}\label{Mori_P_MCT}
    \P
\begin{bmatrix}
u\\
v
\end{bmatrix}
=
\left\langle
\begin{bmatrix}
\delta \rho_{-q}\\
j^L_{-q}
\end{bmatrix}
[u,v]
\right\rangle
\begin{bmatrix}
\frac{1}{NS(|q|)} & 0\\
0 & \frac{m}{Nk_BT}
\end{bmatrix}
\begin{bmatrix}
\delta \rho_{q}\\
j^L_{q}
\end{bmatrix}.
\end{align}
Here $\langle,\rangle$ is the Hilbert space inner product with respect to the equilibrium density $e^{-\beta\H}$, where $\beta= 1/k_BT$. At the statistical equilibrium, $S(|q|)=\langle\delta \rho_{-q}\delta \rho_{q}\rangle/N$ is known as the static structure factor for the liquid and it only depends on the Fourier mode modulus $|q|$. $\langle j^L_{-q}j^L_{q}\rangle=Nk_BT/m$ and this comes from the equipartition theorem. Using the projected time-independent MZE \eqref{eqn:time_ind_PMZE_OP}, we get the following EOM for the time autocorrelation of the density fluctuation $F(q,t)=\langle\delta\rho_{-q},\delta\rho_{q}(t)\rangle$ \cite{reichman2005mode}:
\begin{align}\label{MCT_exact}
    \frac{d^2}{dt^2}F(q,t)+\frac{|q|^2k_BT}{mS(|q|)}F(q,t)+\frac{m}{Nk_BT}\int_0^tds\langle R_{-q}R_{q}(s)\rangle\frac{d}{dt}F(q,t-s).
\end{align}
This EOM is only formally exact since the MZ memory kernel $K(s)=\langle R_{-q}R_{q}(s)\rangle$ is defined {\em formally} using the orthogonal semigroup $e^{is\Q\L}$:
\begin{align*}
  K(s)=\langle R_{-q}R_{q}(s)\rangle=  \langle R_{-q} e^{is\Q\L}R_{q}\rangle,\qquad \text{where}\quad R_{q}=\frac{d}{dt}j_q^L-i\frac{|q|k_BT}{mS(|q|)}\delta\rho_{q}.
\end{align*}
Different approximations to the memory kernel lead to various evolution equations for predicting the glass transition. The classical mode-coupling theory uses several (uncontrollable) assumptions \cite{reichman2005mode}, such as the convolution approximation, and eventually arrives at a self-consistent, nonlinear evolution equation for $F(q,t)$:
\begin{align}\label{MCT_schematic}
    \frac{d^2}{dt^2}F(q,t)+\Omega F(q,t)+\lambda\int_0^tdsF^2(q,s)\frac{d}{dt}F(q,t-s)=0,
\end{align}
where $\Omega,\lambda$ depends on the fluid density $\rho$ and temperature $T$ among many other parameters. By varying them, the mode-coupling equation \eqref{MCT_schematic} predicts a glass transition \cite{bengtzelius1984dynamics}, although it underestimates the true transition temperature \cite{reichman2005mode}. Starting from the formally exact MZE \eqref{MCT_exact}, instead of introducing these uncontrollable approximations, we propose to use the time-independent CMZE to derive a combinatorial series expansion to $K(s)$ and hence a self-consistent EOM for $F(q,t)$. Since the underlying Hamiltonian system \eqref{glass_hamiltonian} is time-invariant, we retain the series expansion \eqref{op_id_1_nonc} up to the second order and get the following generalized mode-coupling equation for predicting the glassy dynamics:
\begin{equation}\label{CMZE-MCT}
\begin{aligned}
\frac{d^2}{dt^2}F(q,t)&+\frac{|q|^2k_BT}{mS(|q|)}F(q,t)-(\omega^2_0(q)+\frac{1}{2}\omega_2^2(q))\int_0^tds\frac{d}{dt}F(q,t-s)ds
\\
&-\frac{m\omega_2^2(q)}{|q|^2k_BT}\int_0^tds\frac{d^2}{ds^2}F(q,s)\frac{d}{dt}F(q,t-s)\\
&-\frac{1}{2}\frac{m\omega_2^2(q)}{|q|^2S(|q|)k_BT}\int_0^tds
\left[\frac{d}{ds}F(q,s)\right]^2
\frac{d}{dt}F(q,t-s)\\
&+\frac{1}{2}\frac{m^2N^2\omega_2^2(q)}{|q|^4k^2_BT^2}\int_0^tds
\left[\frac{d^2}{ds^2}F(q,s)\right]^2
\frac{d}{dt}F(q,t-s)=0.
\end{aligned}
\end{equation}
The derivation of \eqref{CMZE-MCT} is given in \ref{app:MCT_derivation}. In Eqn \eqref{CMZE-MCT}, $\omega^2_0(q)$ and $\omega_2^2(q)$ are functions of the Fourier mode $q$ whose specific form can be explicitly determined given the high-order static moment in the statistical equilibrium (details provided in \ref{app:MCT_derivation}). Retaining more terms in the combinatorial expansion, we can get higher-order mode coupling equations that contain nonlinear terms such as the quartic of $F(q,s)$ or its derivatives. We particularly note that the {\em only} approximation we used in the derivation is the truncation of the combinatorial expansion series.
\subsection{Self-consistent generalized Langevin equation for coarse-grained particles}
\label{sec:app_GLE}
CMZE also applies to non-Hamiltonian systems generated by non-Liouvillian operators. Consider the Langevin dynamics for a $d$-dimensional interactive particle system with total $N$ interacting particles \cite{zhu2021generalized}:
\begin{align}\label{eqn:LE}
\begin{dcases}
d q_i=\frac{1}{m_i} p_idt\\
d p_i=-\nabla_iV(q)dt-\frac{\gamma}{m_i}p_idt+\sigma d \W_{i}(t)
\end{dcases}.
\end{align}
Here $m_i$ is the mass of each particle, $V$ is the interaction potential energy and $\{q_j, p_j\}$ are, respectively, the generalized coordinate and momentum of the $j$-th particle.
$\W_{i}(t)$ is a $d$-dimensional Wiener process which satisfies $d\W_i(t)=\xi_i(t)dt$ with
\begin{align*}
\langle \xi_{ij}(t)\rangle=0,\qquad 
\langle \xi_{ij}(t)\xi_{i'j'}(s)\rangle=2\gamma k_{B}T\delta_{ii'}\delta_{jj'}\delta(t-s),\qquad 1\leq j\leq d,
\end{align*}   
where $\xi(t)$ is the Gaussian white noise. The parameter $\gamma$ is the friction coefficient which is related to $\sigma$ through the fluctuation-dissipation relation $\sigma=(2\gamma/\beta)^{1/2}$, where $\beta=1/k_BT$, $k_B$ is the Boltzmann constant and $T$ the temperature of the equilibrium system. 
The stochastic dynamical system 
\eqref{eqn:LE} is widely used in the mesoscopic modeling of liquids and gases. The Kolmogorov backward operator associated with the stochastic differential equation \eqref{eqn:LE} is given by
\begin{align}\label{Kolmo:LE}
\K=\sum_{i=1}^N\frac{p_i}{m_i}\cdot\partial_{q_i}+
\sum_{i=1}^N\left[-\partial_{q_i}V(q)-
\frac{\gamma p_i}{m_i}\right]\cdot\partial_{p_i}+\sum_i^N\frac{\gamma}{\beta}\partial_{p_i}\cdot\partial_{p_i},
\end{align}
where ``$\cdot$'' denotes the standard dot product. 
If the interaction potential $V(q)$ is strictly positive 
at infinity then the Langevin equation \eqref{eqn:LE} 
admits a unique invariant Gibbs measure $e^{-\beta H}/Z$, where $H(p,q)=\sum_{i=1}^N\|p_i\|_2^2/(2m_i)+V(q)$ is the Hamiltonian and $Z$ is the partition function. Suppose we are interested in the $x$-directional self-diffusion of a coarse-grained particles in the equilibrium. The MZE can be used to derive the EOM for the tagged particle momentum in the $x$-direction, i.e. $p_{jx}$. To this end, we choose $u(0)=p_{jx}$ and introduce Mori's projection operator:
\begin{align}\label{example2_Mori_P}
  \P(\cdot)=\frac{\langle\cdot,p_{jx}(0)\rangle}{\langle  p_{jx}^2(0)\rangle} p_{jx}(0)
\end{align}
For stochastic systems, the dynamics is generated by the Kolmogorov  backward operator \eqref{Kolmo:LE}. Hence in MZEs, all $\L$ operator has to be replaced by $\K$. With this modification, the resulting EOMs, in particular, MZE \eqref{eqn:time_ind_PMZE_POP} and \eqref{eqn:time_ind_MZE_u_t}, can be interpreted as the exact evolution equation for the time 
autocorrelation function $C(t)=\langle p_{jx}(t)p_{jx}(0)\rangle/\langle p_{jx}^2(0)\rangle$ and the full dynamics $p_{jx}(t)$. More detailed discussion on the MZE for stochastic systems can be found in \cite{zhu2021effective,zhu2021generalized}. Specifically, using \eqref{example2_Mori_P} and Eqn \eqref{eqn:time_ind_PMZE_POP}, we can get the following EOM for the time autocorrelation function:
\begin{align}\label{exp2_GLEC}
  \frac{d}{dt}C(t)=\Omega C(t)+\int_0^tK(t-s)C(s)ds.
\end{align}
where $\Omega=-\gamma/m_i$. We can also obtain the full EOM for $p_{jx}$ using Eqn \eqref{eqn:time_ind_MZE_u_t}. This equation is also known as the {\em generalized Langevin equation} and has the following form:
\begin{align}\label{exp2_GLE}
    \frac{d}{dt}p_{jx}(t)=\Omega p_{jx}(t)+\int_0^tK(t-s)p_{jx}(s)ds+g(t),
\end{align}
where $g(t)=e^{t\Q\K}\Q\K p_{jx}(0)$. In the reduced-order modeling for the coarse-grained particle, we normally seek an approximation scheme to calculate the memory kernel $K(t-s)$ in \eqref{exp2_GLE} and then use the second fluctuation-dissipation theorem to build effective models for the noise term $g(t)$ \cite{lei2016data,zhu2020generalized}. In this section, we propose to use the CMZE to finish the first task (the more difficult one), where the combinatorial expansion can be employed to provide a first-principle method or data-driven approximation ansatz for calculating $K(t-s)$. Since \eqref{example2_Mori_P} is a rank-$1$ projection operator, the CMZE is commutative and hence the result in Section \ref{sec:CCMZE} directly applies. Specifically, with the combinatorial expansion \eqref{K_combi_expansion}, we get the following self-consistent EOM for the correlation function:
\begin{align}
\frac{d}{dt}C(t)=-\frac{\gamma}{m_i} C(t)+\sum_{n=0}^{\infty}\sum_{k=0}^n\frac{(-1)^{n-k}}{(n-k)!k!}f_n\int_0^tC^{k}(s)C(t-s)ds.\label{exp2_GLE_CMZEC}
\end{align}
Application to Eqn \eqref{eqn:time_ind_MZE_u_t} leads to the self-consistent generalized Langevin equations for the coarse-grained particle:
\begin{equation}\label{exp2_GLE_CMZE}
\begin{aligned}
\begin{dcases}
\frac{d}{dt}q_{jx}(t)&=\frac{1}{m_i}p_{jx}(t)\\
\frac{d}{dt}p_{jx}(t)&=-\frac{\gamma}{m_i} p_{jx}(t)+\sum_{n=0}^{\infty}\sum_{k=1}^n\frac{(-1)^{n-k}}{(n-k)!k!}f_n\int_0^tC^k(t-s)p_{jx}(s)ds+g(t),
\end{dcases}
\end{aligned}
\end{equation}
where $g(t)$ can be constructed using stochastic process series expansion \cite{zhu2021generalized,zhu2020generalized} once knowing $K(t)$. By direct computation, we find that the first few $f_n$s are given by: 
\begin{align*}
    f_0&=-\frac{1}{m_i}\langle\partial_{q_{ix}}^2V(q)\rangle\\
    f_1&=0\\
    f_2&=\left(-\frac{1}{m_i^2}+\frac{1}{\gamma}-\frac{1}{m_i}\right)\langle\partial_{q_{ix}}^2V(q)\rangle
    +\left(\frac{2}{m_i\gamma}+\frac{1}{\gamma^2}\right)\langle\partial_{q_{ix}}^2V(q)\rangle^2\\
    &\ \ \ \ 
    +\frac{m_i}{\gamma^2}\langle\partial_{q_{ix}}V(q)\partial_{q_{ix}}^2V(q)\rangle
    +\frac{1}{m_i\gamma}\langle\partial_{q_{ix}}^3V(q)\rangle+2\frac{\gamma^2}{m_i^2}-\frac{\gamma}{m_i^2}.
\end{align*}
With these three terms, we can construct a second-order approximation to Eqns \eqref{exp2_GLE_CMZEC} and then solve \eqref{exp2_GLE_CMZE}. Higher-order combinatorial expansions can be obtained in a similar way which leads to more accurate approximations of the generalized Langevin equation. For general stochastic systems, it is possible to use a field-theoretical approach and Feynman-like diagrams to get the evolution equation for the correlation function as shown in \cite{cardy2008non,reichman2005mode,bouchaud1996mode}. Interested readers may compare the method adopted therein with the CMZE in terms of the EOM and the graphical representation.     
\subsection{Self-consistent EOM for strongly correlated electron systems}
\label{sec:app_hubbard}
Lastly, we consider the application of CMZE in quantum many-body systems and show that by choosing different projection operators, we get various  self-consistent EOMs for Green's function. These new equations are compared with the Kadanoff-Baym equation in the Keldysh formalism \cite{stefanucci2013nonequilibrium}. To this end, we consider the tight-binding Hubbard model for $N$ interacting electrons:
\begin{align}\label{H_hubbard}
    \H=\H_0+\H_I=
    \underbrace{\sum_{i,\sigma}(\epsilon_0-\mu)n_{i\sigma}+t\sum_{\langle i,j\rangle,\sigma}c_{i\sigma}^{\dagger}c_{j\sigma}}_{\H_0}+\underbrace{U\sum_{i}n_{i\uparrow}n_{i\downarrow}}_{\H_I},
\end{align}
where $\epsilon_0$ is the atomic level, $t$ is the hopping integral between sites, and $\langle i,j\rangle$ indicates the hopping is restricted to neighboring sites. $U$ is the strength of the Coulomb interaction potential. $c^{\d}_{i\sigma},c_{i\sigma}$ is the creation and annihilation operator for an electron with spin $\sigma$ on-site $i$. A chemical potential $\mu$ is added to the atomic level $\epsilon_0$. Furthermore, we restrict ourselves to the paramagnetic case with an equal number of spin-up ($\uparrow$) and spin-down ($\downarrow$) electrons in the lattice. We study the single-particle excitation of the Hubbard model \eqref{H_hubbard} with the retarded Green's function defined by:
\begin{align*}
    G^{R}_{i\sigma,j\sigma'}(t,t'):=-i\theta(t-t')\langle[c_{i\sigma}(t),c^{\d}_{j\sigma'}(t')]_+\rangle.
\end{align*}
Here $\theta(t)$ is the step function, $c_{i\sigma}(t),c_{j\sigma'}(t')$ are the creation and annihilation operators in the Heisenberg picture. $[,]_+$ is the anticommutator. The ensemble average $\langle\cdot\rangle$ is normally taken with respect to the grand canonical ensemble. To get the evolution equation for the retarded Green's function, we introduce the Liouville superoperator $\L:=[\H,\cdot]$ and an inner product between operators $\A$ and $\B$: $(\A|\B)=\langle[\A^{\d},\B]_+\rangle$. Since at the statistical equilibrium, we have $G^{R}_{i\sigma,j\sigma'}(t,t')=G^{R}_{i\sigma,j\sigma'}(t-t')$. Without loss of generality, we may set $t>0,t'=0$ and work on the retarded Green's function $G^{R}_{i\sigma,j\sigma'}(t)=-i\langle [c_{i\sigma}(t),c^{\d}_{j\sigma'}(0)]_+\rangle=-i(c_{j\sigma'}|c_{i\sigma}(t))$. In the Heisenberg picture, the EOM for any operator $\A$ is given by $\A(t)=e^{it\H}\A e^{-it\H}=e^{it\L}\A=\U(t,0)\A$. Hence we prompt to use the time-independent MZE:
\begin{align}\label{eqn:time_ind_PMZE_OP_hubbard}
\frac{d}{dt}\P e^{it\L}\P=i\P e^{it\L}\P\L\P
+\int_0^t\P e^{i(t-s)\L}i\P\L e^{is\Q\L}i\Q\L\P ds.
\end{align}
Depending on the choice of $\P$, different EOMs for the Green's function $G^{R}_{i\sigma,j\sigma'}(t)$ can be obtained. Here are two typical examples.
\paragraph{One-dimensional $\P$}
If we are only interested in the evolution equation for the onsite Green's function $G_{i\sigma,i\sigma}^R(t)$, we can define a one-dimensional Mori's projection operator:
\begin{align}\label{Mori_P_hubbard_one}
    \P=|c_{i\sigma})(c_{i\sigma}|.
\end{align}
Note that this is indeed a projection operator since $(c_{i\sigma}|c_{i\sigma})=1$. Since the range space of \eqref{Mori_P_hubbard_one} is one-dimensional, with this projection operator, the combinatorial expansion for \eqref{eqn:time_ind_PMZE_OP_hubbard} yields a commutative CMZE. Applying the result in Section \ref{sec:CCMZE}, we can get that  
\begin{align}\label{Green_f_CMZE}
    i\frac{d}{dt}G^R(t)=\Omega G^R(t)+\sum_{n=0}^{\infty}\sum_{k=0}^n\frac{(-1)^{n-k}}{k!(n-k)!}f_n\int_0^t[iG^{R}(s)]^kiG(t-s)ds,
\end{align}
where $G^R(t)=G^{R}_{i\sigma,i\sigma}(t)$, $\Omega=(\epsilon_0-\mu)+U\langle n_{i\bar\sigma}\rangle$, and the first two $f_n$s are given by:
\begin{equation}\label{f_n_hubbard_1P}
\begin{aligned}
f_0&=U^2(\langle n_{i\bar\sigma}\rangle^2-\langle n_{i\bar\sigma}\rangle)+(\epsilon_0-\mu)U(\langle n_{i\bar\sigma}\rangle-1)-2t^2\\
f_1&=-\frac{1}{(\epsilon_0-\mu)+U\langle n_{i\bar\sigma}\rangle}[(\epsilon_0-\mu)^3+3U(\epsilon_0-\mu)^2\langle n_{i\bar\sigma}\rangle+6t^2(\epsilon_0-\mu)+3U^2(\epsilon_0-\mu)\langle n_{i\bar\sigma}\rangle\\
&\ \ \ 
+t^2U(4\langle n_{i\bar\sigma}\rangle+\langle n_{i-1\bar\sigma}\rangle
+\langle n_{i+1\bar\sigma}\rangle)
+U^3\langle n_{i\bar\sigma}\rangle
]\\
&\ \ \
-U^2\langle n_{i\bar\sigma}\rangle^2+2U^2\langle n_{i\bar\sigma}\rangle+2U(\epsilon_0-\mu)+4t^2+(\epsilon_0-\mu)^2.
\end{aligned}    
\end{equation}
The derivation of $\Omega$ and \eqref{f_n_hubbard_1P} are provided in \ref{app:hubbard}. According to Theorem \ref{thm_combi_expansion}, we know $f_n=L_n(\gamma_0,\cdots,\gamma_{n-1})$, where  $\gamma_j=(c_{i\sigma}|(i\L)^j c_{i\sigma})$ are polynomials of statistical moments of the form:
\begin{align}\label{formula}
    \langle n^{j_{1}}_{1,\sigma_{1}}\cdots n^{j_{i-1}}_{i-1,\sigma_{i-1}}n^{j_{i}}_{i,\sigma_i}n^{j_{i+1}}_{i+1,\sigma_{i+1}}\cdots n^{j_{N}}_{N,\sigma_{N}},
    n^{j'_{1}}_{1,\bar\sigma_{1}}\cdots n^{j'_{i-1}}_{i-1,\bar\sigma_{i-1}}n^{j'_{i}}_{i,\bar\sigma_i}n^{j'_{i+1}}_{i+1,\bar\sigma_{i+1}}\cdots n^{j'_{N}}_{N,\bar\sigma_{N}}\rangle,
\end{align} 
where the power index $j_1,j_2\cdots \in\{0,1\}$, $j'_1,j'_2\cdots \in\{0,1\}$, $\sigma_1,\sigma_2\cdots \in\{\uparrow,\downarrow\}$ and $\bar \sigma_1,\bar \sigma_2\cdots \in\{\downarrow,\uparrow\}$. To be noticed that \eqref{formula} can also be understood as a many-particle Green's function at the equilibrium. To solve Eqn \eqref{Green_f_CMZE}, the problem boils down to the determination of $f_n$s and the evaluation of statistical moments of the form \eqref{formula}. For the first task, we have already shown in the \ref{app:hubbard} that low-order $f_n$s can be calculated manually by hand. High-order $f_n$s are better to be calculated using a quantum analogue of the enumerative combinatorics algorithm developed in \cite{amati2019memory,zhu2020generalized,zhu2021effective}.  The ensemble average \eqref{formula} is hard to evaluate for a general grand canonical ensemble. However, if we are interested in the extremely correlated case where $U\gg t$, they may be calculated approximated using $e^{-\beta(\H-\mu\N)}/Z\approx e^{-\beta(\H_I+(\epsilon-\mu)\N)}/Z_{e}$, where the non-hopping Hamiltonian $\H_I+(\epsilon-\mu)\N$ is diagonalized. Once $f_n$s are determined, we can truncate the expansion series in \eqref{Green_f_CMZE} to a certain order and solve the equation using any accurate time integrator.

\paragraph{$N$-dimensional $\P$}
To get the EOM for {\em all} one-particle Green's function, we introduce a multi-dimensional Mori's projection operator:
\begin{align}\label{Mori_P_hubbard_nd}
    \P=\sum_{j\sigma'}|c_{j\sigma'})(c_{j\sigma'}|,
\end{align}
where the summation $\sum_{j\sigma'}$ runs over all possible states $\sigma'$ for all electrons in the lattice. This is indeed a projection operator since $(c_{i\sigma}|c_{j\sigma'})=\delta_{ij}\delta_{\sigma\sigma'}$. Due to the symmetry of the Hubbard model, we have $G^R_{i\sigma,j\sigma}(t)=G^R_{i\bar \sigma,j\bar \sigma}(t)=G^R_{ij}(t)$ \cite{schuler2020nessi}. Hence without loss of generality, we can drop the spin index and restrict ourselves to Green's function $G^R_{ij}(t)$. The projection operator \eqref{Mori_P_hubbard_nd} naturally becomes $\P=\sum_{j}|c_{j\sigma})(c_{j\sigma}|$ and we get a $N$-dimensional MZE for the retarded Green's function: 
\begin{align}\label{MZE_nd_matrix_nd}
    i\frac{d}{dt} G^{R}(t)=\Omega G^{R}(t)+i\int_0^t K^{R}(s) G^{R}(t-s)ds.
\end{align}
In Eqn \eqref{MZE_nd_matrix_nd}, $G^R(t),\Omega$ and $K(s)$ are $N\times N$ matrices with entries: 
\begin{equation}\label{MZE_nd_hubbard_nd1}
    \begin{aligned}
    G^{R}_{ij}(t)&=(G^{R})_{ij}(t)\\
   \Omega^{R}_{ij}&=(\Omega^{R})_{ij}=(c_{j\sigma}|i\L c_{i\sigma})\\
   K^{R}_{ij}(s)&=(K^{R})_{ij}(s)
   =\left(c_{j\sigma}\bigg|i\L e^{is\Q\L}i\Q\L c_{i\sigma}\right).
    \end{aligned}
\end{equation}
For the Hubbard model, the frequency matrix $\Omega_{ij}$ is explicitly given by:
\begin{align}\label{omega_hubbard_MZ}
    \Omega_{ij}&=(H_0)_{_{ij}}+U\langle n_{i\bar\sigma}\rangle\delta_{ij}=[(\epsilon_0-\mu)\delta_{ij}+t(\delta_{i-1j}+\delta_{i+1j})]+U\langle n_{i\bar\sigma}\rangle\delta_{ij}.
\end{align}
The combinatorial expansion can be applied to approximate the memory matrix $K(s)$. Since the range space of the projection operator is $N$-dimensional, the resulting CMZE is non-commutative and admits a general form:  
\begin{align}\label{Green_f_CMZE_nd}
    i\frac{d}{dt}G^R(t)=\Omega G^R(t)+\sum_{n=0}^{\infty}\sum_{k=0}^n\frac{(-1)^{n-k}}{k!(n-k)!}\int_0^t[iG^{R}(s)\Omega_n]^kiG^R(t-s)ds.
\end{align}
Truncating the combinatorial expansion at the first order, we get the following self-consistent EOM for the retarded Green's function of the Hubbard model:
\begin{align}\label{1st-order_CMZE_hubbard}
i\frac{d}{dt} G^{R}(t)=\Omega G^{R}(t)+i\int_0^t(\Omega_0-\Omega_1) G^{R}(t-s)ds-\int_0^tG^R(s)\Omega_1G^R(t-s)ds,
\end{align}
where the explicit expression for $\Omega_0,\Omega_1$ is given in \ref{app:hubbard}. Due to the symmetry of the hopping term in the Hubbard Hamiltonian, all $\Omega_n$s are Toeplitz matrices. For large $n$, the descending diagonal element in $\Omega_n$ can be calculated like $f_n$s using the same method we proposed in the previous paragraph. 

The above Mori's projection operators and the resulting CMZEs are two extreme cases corresponding to $\text{dim}(\text{Ran}(\P))=1$ and $\text{dim}(\text{Ran}(\P))=N$ respectively. Since $\text{dim}(\text{Ran}(\P))$ can be specified to be any $n$ for $1\leq n\leq N$, we actually obtain a family of CMZEs for the {\em resolved} Green's function, i.e. the Green's function on restricted sites of the whole lattice (See schematic illustration in Figure \ref{fig:Variant_P_Hubbard}). In applications, if one is only required to know a low-dimensional resolved Green's function that of theoretical interest, say $\text{dim}(\text{Ran}(\P))=3$, we may use the aforementioned combinatorial algorithm to get a high-order approximation to the CMZE. Presumably, it can be solved with relatively low computational cost due to the low dimensionality of the equation. This is in contrast with the traditional field-theoretical approach adopted in the Kadanoff-Baym equation (KBE) since for the latter such a dimension-reduction is not easily obtained due to the structure of the Feynman diagram used in approximating the self-energy. More explanations on the derivation of KBE can be found in the following paragraph.     
\begin{figure}
\centering
\includegraphics[height=2.5cm]{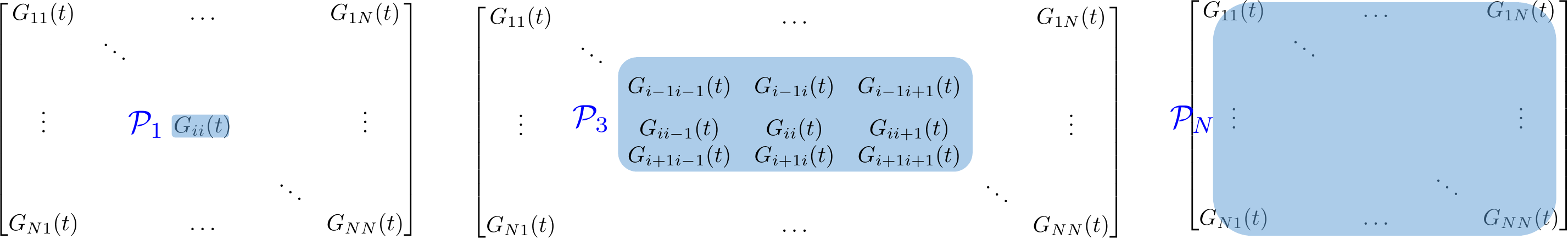}
\caption{Schematic illustration for the resolved one-particle Green's function under different Mori's projection operator $\P$. Here the projection operator $\P$ is re-denoted as $\P_n$, where $n=\text{dim}(\text{Ran}(\P))$ indicates the dimension of the range space. As $n\rightarrow N$, the resolved Green's function that is included in CMZE increases as $O(n^2)$. This scaling also applies to the time-dependent CMZE for non-equilibrium systems.}
\label{fig:Variant_P_Hubbard}
\end{figure}
\paragraph{Comparison with the Kadanoff-Baym equation}
The CMZE can be compared with the self-consistent approximation to the Kadanoff-Baym equation (KBE). The KBE is a general ansatz to derive the EOM for the Green's function \cite{stefanucci2013nonequilibrium}. In the statistical equilibrium when $G^{R}(t,t')=G^{R}(t),t'=0$, using the Langreth rule to simplify the Keldysh contour integral, the KBE for the Hubbard model \eqref{H_hubbard} takes the form (see e.g. Eqn (151) in \cite{schuler2020nessi}):
\begin{align}\label{KBE_gl_G}
    i\partial_t G^{R}(t)&= H_0(t) G^{R}(t)
    +\int_0^t\Sigma^R(t,s) G^{R}(s)ds, 
\end{align}
where $G^{R}(t)$, $H_0(t)$ and $\Sigma^R(t,s)$ are $N\times N$ matrices. $\Sigma^R(t,s)$ is identified as the self-energy of the system. The approximation and computation of $\Sigma^R(t,s)$ is the main technical difficulty in the whole scheme. The commonly adopted approach to this end is the many-body perturbation theory and Feynman diagrammatic expansion we briefly sketched in Section \ref{sec:tree_feynman_diag}. Here we only state the result and refer interested readers to the standard textbook \cite{stefanucci2013nonequilibrium} for technical details. Specifically, the first two Feynman diagrams in \eqref{renormalized_sigma} represent the first-order Hartree-Fock approximation to the self-energy. The Hartree-Fock term is time-local. Its contribution to the memory integral makes it convolutionless therefore the resulting term is normally merged into $ H_0(t)$. By direct computation, we find that after this step, $H_0(t)$ is explicitly given by
\begin{align}\label{H_0_hubbard_KBE}
   (H_0)_{ij}(t)=[(\epsilon_0-\mu)\delta_{ij}+t(\delta_{i-1j}+\delta_{i+1j})]+U\langle n_{i\bar\sigma}(t)\rangle\delta_{ij}.
\end{align}
To be noticed that the first $t$ on the RHS represents the hopping integral. The last two diagrams \eqref{renormalized_sigma} correspond to the second-order Born approximation to the self-energy:
\begin{align}\label{2nd_Born_app_hubbard}
    \Sigma_{ij}(t,s)=U^2G^R_{ij}(t-s)G^R_{ij}(t-s)G^R_{ji}(s-t).
\end{align}
Here we used equality $G_{ij,\sigma}(t)=G_{ij,\bar\sigma}(t)=G_{ij}(t)$ to simplify the notation. Using Born's approximation, we get the commonly used self-consistent EOM for the retarded Green's function:
\begin{align}\label{KBE_gl_G_2ndBorn}
    i\partial_t G^{R}(t)&= H_0(t) G^{R}(t)
    +U^2\int_0^t(G^R(t-s)\odot G^R(t-s)\odot [G^{R}(s-t)]^T)G^{R}(s)ds,
\end{align}
where $\odot$ is the Hadamard product.
Comparing \eqref{omega_hubbard_MZ} and \eqref{H_0_hubbard_KBE}, we find that at $t=0$, the frequency matrix $\Omega$ in CMZE \eqref{1st-order_CMZE_hubbard} equals to the Hartree-Fock streaming matrix $H_0(t=0)$ in \eqref{KBE_gl_G_2ndBorn}. The memory kernel matrix $K(t-s)$ can be understood as the self-energy for the Hubbard model under the  Mori-Zwanzig framework. Moreover, the combinatorial expansion to the memory kernel yields self-consistent EOM \eqref{1st-order_CMZE_hubbard} for the Green's function that is similar to \eqref{KBE_gl_G_2ndBorn}. 

Lastly, we note that the above derivation can be readily generalized to get the EOM for other Green's functions in the Keldysh contour. Moreover, using CMZE \eqref{op_id_1_nonc_type2}, we can get similar equations for the non-equilibrium Green's function for time-dependent many-body systems. All these extensions will be addressed in a coming paper \cite{zhu2022complete}.  
\section{Summary}
\label{sec:summary}
In this paper, we developed a theoretical framework to get the combinatorial series expansion for the Mori-Zwanzig equation. For the simplest commutative case, this expansion is directly derived from the classical Fa\`a di Bruno formula. This method is then generalized to obtain the noncommutative expansion series for operators. To this end, we introduced several new noncommutative polynomials of words and a symbolic equation for these polynomials. As a result, the noncommutative combinatorial expansion for the time-independent and time-dependent Mori-Zwanzig equation can be derived from the abstract word equation after introducing several algebraic homomorphisms that map the associative algebra of words into the algebra of operators. 

We also introduced an equivalent, tree representation of the combinatorial Mori-Zwanzig equation (CMZE). In this representation, the word equation that defines the combinatorial expansion for the Mori-Zwanzig memory kernel is reformulated into the functional equation for the graded exponential generating function of trees. This enables us to make direct connections between the CMZE and the combinatorial Dyson-Schwinger equation. In addition, we demonstrated the similarities as well as the differences between the tree diagrams and the Feynman diagrams in terms of the representation of the self-consistent expansions for the self-energy (memory kernel).  

The combinatorial Mori-Zwanzig theory provides a systematical way to get the self-consistent equation of motions (EOM) for many-body systems. The generality of this newly developed framework makes it widely applicable to classical, stochastic, and quantum mechanical many-body systems. We provide example applications to all three cases. In particular, we have shown how CMZE leads to a general mode-coupling equation that can be used to predict glass transition. For the stochastic particle system, the CMZE yields novel series expansions of the generalized Langevin equation for coarse-grained particles. Lastly, the application of CMZE into the Hubbard model of interacting electrons introduces a new EOM for Green's function that is comparable with the self-consistent approximation to the Kadanoff-Baym equation. These examples demonstrated that the combinatorial Mori-Zwanzig theory provides a possible alternative framework to get the EOM for the correlation (Green's) function that is generally different from the renormalized many-body perturbation theory commonly adopted in many-body physics and quantum field theory. With that being said, a detailed assessment of its practical utility awaits future numerical investigations. 

\section{Acknowledgement} This material is based upon work supported by the U.S. Department of Energy, Office of Science, Office of Advanced Scientific Computing Research and Office of Basic Energy Sciences, Scientific Discovery through Advanced Computing (SciDAC) program under Award Number DE-SC0022198 (Contract No. DE-AC02-05CH11231).


\appendix
\section{Table for $k_{n,m}(t,s)$}
\label{Appendix_table}
Setting $v=u(0)$, the first few $k_{n,m}(t,s)$ for $0 \leq n,m\leq 2$ is summarized in Table \ref{Tab:1}.
\begin{landscape}
\begin{table}
\centering
\caption{The first few $k_{n,m}(t,s)$}
\begin{tabular}{|c|c|c|c|}\hline
\multirowcell{1}{}& 
\multirowcell{1}{$m=0$}&
\multirowcell{1}{$m=1$}&
\multirowcell{1}{$m=2$}\\
\hline
\multirowcell{5}{$n=0$}&
\multirowcell{5}{$\langle\L(s)\Q\L(t)v,v\rangle$}&
\multirowcell{5}{$\displaystyle-\frac{\langle\L(s)\Q\L_1\Q\L(t)v,v\rangle}{\langle\L_1v,v\rangle}$}&
\multirowcell{5}{$\displaystyle\frac{\langle\L(s)[(\Q\L_1)^2-\Q\L_2]\Q\L(t)v,v\rangle}{\langle\L_1v,v\rangle^2}$\\
$\displaystyle
+\frac{\langle(\L_1^2+\L_2)v,v\rangle\langle\L(s)\Q\L_1\Q\L(t)v,v\rangle}{\langle\L_1v,v\rangle^3}$}\\
& & &\\
& & &\\
& & &\\
& & &\\ 
\hline
\multirowcell{5}{$n=1$}&
\multirowcell{5}{$\displaystyle\frac{\langle\L(s)\Q\L_1\Q\L(t)v,v\rangle}{\langle\L_1v,v\rangle}$}&
\multirowcell{5}{$\displaystyle-\frac{\langle\L(s)(\Q\L_1)^2\Q\L(t)v,v}{\langle\L_1v,v\rangle^2}$}&
\multirowcell{5}{$\displaystyle\frac{\langle\L(s)\Q\L_1[(\Q\L_1)^2-\Q\L_2]\Q\L(t)v,v\rangle}{\langle\L_1v,v\rangle^3}$\\
$\displaystyle
+\frac{\langle(\L_1^2+\L_2)v,v\rangle\langle\L(s)(\Q\L_1)^2\Q\L(t)v,v\rangle}{\langle\L_1v,v\rangle^4}$}\\
& & &\\
& & &\\
& & &\\
& & &\\ 
\hline
\multirowcell{10}{$n=2$}&
\multirowcell{10}{$\displaystyle\frac{\langle\L(s)[(\Q\L_1)^2+\Q\L_2]\Q\L(t)v,v\rangle}{\langle\L_1v,v\rangle^2}$\\
$\displaystyle
+\frac{\langle(\L_1^2+\L_2)v,v\rangle\langle\L(s)\Q\L_1\Q\L(t)v,v\rangle}{\langle\L_1v,v\rangle^3}$}&
\multirowcell{10}{$\displaystyle\frac{-\langle\L(s)[(\Q\L_1)^2+\Q\L_2]\Q\L_1\Q\L(t)v,v\rangle}{\langle\L_1v,v\rangle^3}$\\
$\displaystyle
-\frac{\langle(\L_1^2+\L_2)v,v\rangle\langle\L(s)(\Q\L_1)^2\Q\L(t)v,v\rangle}{\langle\L_1v,v\rangle^3}$}&
\multirowcell{10}{$\displaystyle\frac{\langle\L(s)[(\Q\L_1)^4-(\Q\L_2)^2]\Q\L (t)v,v\rangle}{\langle\L_1v,v\rangle^4}$\\
$\displaystyle
+\frac{\langle \L(s)\Q\L_1[(\Q\L_1)^2-\Q\L_2]\Q\L(t)v,v\rangle\langle(\L_1^2+\L_2)v,v\rangle}{\langle\L_1v,v\rangle^5}$\\
$\displaystyle
+\frac{\langle \L(s)[(\Q\L_1)^2+\Q\L_2]\Q\L_1\Q\L(t)v,v\rangle\langle(\L_1^2+\L_2)v,v\rangle}{\langle\L_1v,v\rangle^5}$\\
$\displaystyle
+\frac{\langle\L(s)(\Q\L_1)^2\Q\L(t)v,v\rangle\langle(\L_1^2+\L_2)v,v\rangle}{\langle\L_1v,v\rangle^6}$
}\\
& & &\\
& & &\\
& & &\\
& & &\\ 
& & &\\ 
& & &\\ 
& & &\\ 
& & &\\
& & &\\
\hline
\end{tabular}
\label{Tab:1}
\end{table}
\end{landscape}
\section{Tree representation of the noncommutative polynomials}\label{Appendix_tree_repre}
The first few terms of the Type-II noncommutative Bell polynomials can be represented using trees as:
\begin{equation}\label{Tree_graphs_bell2}
\begin{aligned}
F_0^{\hB}&=
{\fontsize{4}{1}\selectfont
\begin{forest}
 for tree={grow=90, l=1mm}
[]
 \path[fill=black]  (!.parent anchor) circle[radius=2pt];
\end{forest}
}
\\
F_1^{\hB}&=
{\fontsize{4}{1}\selectfont
\begin{forest}
 for tree={grow=90, l=1mm}
[[]]
 \path[fill=black]  (!.parent anchor) circle[radius=2pt];
 \path[fill=black] (!1.child anchor) circle[radius=2pt];
\end{forest}
}
\\
F_2^{\hB}&=
{\fontsize{4}{1}\selectfont
\begin{forest}
 for tree={grow=90, l=1mm}
[[[]]]
 \path[fill=black]  (!.parent anchor) circle[radius=2pt];
 \path[fill=black] (!1.child anchor) circle[radius=2pt]
                   (!11.child anchor) circle[radius=2pt];
\end{forest}
}
+
{\fontsize{4}{1}\selectfont
\begin{forest}
for tree={grow=90, l=1mm}
[[][]]
\path[fill=black]  (!.parent anchor) circle[radius=2pt];
\path[fill=black] (!1.child anchor) circle[radius=2pt]
                 (!2.child anchor) circle[radius=2pt];
\end{forest}
}
\\
F_3^{\hB}&=
{\fontsize{4}{1}\selectfont
\begin{forest}
 for tree={grow=90, l=1mm}
[[[[]]]]
 \path[fill=black]  (!.parent anchor) circle[radius=2pt];
 \path[fill=black] (!1.child anchor) circle[radius=2pt]
                    (!11.child anchor) circle[radius=2pt]
                   (!111.child anchor) circle[radius=2pt];
\end{forest}
}
+\frac{3}{2}
{\fontsize{4}{1}\selectfont
\begin{forest}
 for tree={grow=90, l=1mm}
[[[]][]]
 \path[fill=black]  (!.parent anchor) circle[radius=2pt];
 \path[fill=black] (!1.child anchor) circle[radius=2pt]
                  (!11.child anchor) circle[radius=2pt]
                  (!2.child anchor) circle[radius=2pt];
\end{forest}
}
+\frac{3}{2}
{\fontsize{4}{1}\selectfont
\begin{forest}
 for tree={parent anchor=east, child anchor=east, grow=90, l=1mm}
[[[][]]]
 \path[fill=black]  (!.parent anchor) circle[radius=2pt];
 \path[fill=black] (!1.child anchor) circle[radius=2pt]
                   (!11.child anchor) circle[radius=2pt]
                   (!12.child anchor) circle[radius=2pt];
\end{forest}
}
+
{\fontsize{4}{1}\selectfont
\begin{forest}
 for tree={grow=90, l=1mm}
[[][][]]
 \path[fill=black]  (!.parent anchor) circle[radius=2pt];
 \path[fill=black] (!1.child anchor) circle[radius=2pt]
                   (!2.child anchor) circle[radius=2pt]
                   (!3.child anchor) circle[radius=2pt];
\end{forest}
}
\\
F_{4}^{\hB}&=
{\fontsize{4}{1}\selectfont
\begin{forest}
 for tree={grow=90, l=1mm}
 [[[[[]]]]]
 \path[fill=black]  (!.parent anchor) circle[radius=2pt];
 \path[fill=black] (!1.child anchor) circle[radius=2pt]
                  (!11.child anchor) circle[radius=2pt]
                  (!111.child anchor) circle[radius=2pt]
                  (!1111.child anchor)  circle[radius=2pt];
\end{forest}
}
+2
{\fontsize{4}{1}\selectfont
\begin{forest}
 for tree={grow=90, l=1mm}
 [[[[][]]]]
 \path[fill=black]  (!.parent anchor) circle[radius=2pt];
 \path[fill=black] (!1.child anchor) circle[radius=2pt]
                  (!11.child anchor) circle[radius=2pt]
                  (!111.child anchor) circle[radius=2pt]
                  (!112.child anchor)  circle[radius=2pt];
\end{forest}
}
+2
{\fontsize{4}{1}\selectfont
\begin{forest}
 for tree={grow=90, l=1mm}
[[[[]][]]]
 \path[fill=black]  (!.parent anchor) circle[radius=2pt];
 \path[fill=black] (!1.child anchor) circle[radius=2pt]
                  (!11.child anchor) circle[radius=2pt]
                  (!111.child anchor) circle[radius=2pt]
                  (!12.child anchor) circle[radius=2pt];
\end{forest}
}
+2
{\fontsize{4}{1}\selectfont
\begin{forest}
 for tree={grow=90, l=1mm}
[[[[]]][]]
 \path[fill=black]  (!.parent anchor) circle[radius=2pt];
 \path[fill=black] (!1.child anchor) circle[radius=2pt]
                  (!11.child anchor) circle[radius=2pt]
                  (!111.child anchor) circle[radius=2pt]
                  (!2.child anchor) circle[radius=2pt];
\end{forest}
}
+3
{\fontsize{4}{1}\selectfont
\begin{forest}
 for tree={grow=90, l=1mm}
 [ [[][]] []]
 \path[fill=black]  (!.parent anchor) circle[radius=2pt];
 \path[fill=black] (!1.child anchor) circle[radius=2pt]
                  (!11.child anchor) circle[radius=2pt]
                  (!12.child anchor) circle[radius=2pt]
                  (!2.child anchor)  circle[radius=2pt];
\end{forest}
}
+2
{\fontsize{4}{1}\selectfont
\begin{forest}
 for tree={grow=90, l=1mm}
 [[[]] [] [] ]
 \path[fill=black]  (!.parent anchor) circle[radius=2pt];
 \path[fill=black] (!1.child anchor) circle[radius=2pt]
                  (!11.child anchor) circle[radius=2pt]
                  (!2.child anchor) circle[radius=2pt]
                  (!3.child anchor)  circle[radius=2pt];
\end{forest}
}
+2
{\fontsize{4}{1}\selectfont
\begin{forest}
 for tree={grow=90, l=1mm}
 [ [[][][]] ]
 \path[fill=black]  (!.parent anchor) circle[radius=2pt];
 \path[fill=black] (!1.child anchor) circle[radius=2pt]
                  (!11.child anchor) circle[radius=2pt]
                  (!12.child anchor) circle[radius=2pt]
                  (!13.child anchor)  circle[radius=2pt];
\end{forest}
}
+
{\fontsize{4}{1}\selectfont
\begin{forest}
 for tree={grow=90, l=1mm}
 [ [][][] []]
 \path[fill=black]  (!.parent anchor) circle[radius=2pt];
 \path[fill=black] (!1.child anchor) circle[radius=2pt]
                  (!2.child anchor) circle[radius=2pt]
                  (!3.child anchor) circle[radius=2pt]
                  (!4.child anchor)  circle[radius=2pt];
\end{forest}
}
\\
&\ \ \vdots
\end{aligned}
\end{equation}
Similarly, the first few terms of the noncommutative bipartition polynomials can be represented using trees as:
\begin{equation}\label{Tree_graphs_bipart}
\begin{aligned}
F_0^{\tP}&=
{\fontsize{4}{1}\selectfont
\begin{forest}
 for tree={grow=90, l=1mm}
[]
 \path[fill=black]  (!.parent anchor) circle[radius=2pt];
\end{forest}
}
\\
F_1^{\tP}&=
{\fontsize{4}{1}\selectfont
\begin{forest}
 for tree={grow=90, l=1mm}
[[]]
 \path[fill=black]  (!.parent anchor) circle[radius=2pt];
 \path[fill=black] (!1.child anchor) circle[radius=2pt];
\end{forest}
}
\\
F_2^{\tP}&=-
{\fontsize{4}{1}\selectfont
\begin{forest}
 for tree={grow=90, l=1mm}
[[[]]]
 \path[fill=black]  (!.parent anchor) circle[radius=2pt];
 \path[fill=black] (!1.child anchor) circle[radius=2pt]
                   (!11.child anchor) circle[radius=2pt];
\end{forest}
}
+
{\fontsize{4}{1}\selectfont
\begin{forest}
for tree={grow=90, l=1mm}
[[][]]
\path[fill=black]  (!.parent anchor) circle[radius=2pt];
\path[fill=black] (!1.child anchor) circle[radius=2pt]
                 (!2.child anchor) circle[radius=2pt];
\end{forest}
}
\\
F_3^{\tP}&=
{\fontsize{4}{1}\selectfont
\begin{forest}
 for tree={grow=90, l=1mm}
[[[[]]]]
 \path[fill=black]  (!.parent anchor) circle[radius=2pt];
 \path[fill=black] (!1.child anchor) circle[radius=2pt]
                    (!11.child anchor) circle[radius=2pt]
                   (!111.child anchor) circle[radius=2pt];
\end{forest}
}
-
{\fontsize{4}{1}\selectfont
\begin{forest}
 for tree={grow=90, l=1mm}
[[[]][]]
 \path[fill=black]  (!.parent anchor) circle[radius=2pt];
 \path[fill=black] (!1.child anchor) circle[radius=2pt]
                  (!11.child anchor) circle[radius=2pt]
                  (!2.child anchor) circle[radius=2pt];
\end{forest}
}
-
{\fontsize{4}{1}\selectfont
\begin{forest}
 for tree={parent anchor=east, child anchor=east, grow=90, l=1mm}
[[[][]]]
 \path[fill=black]  (!.parent anchor) circle[radius=2pt];
 \path[fill=black] (!1.child anchor) circle[radius=2pt]
                   (!11.child anchor) circle[radius=2pt]
                   (!12.child anchor) circle[radius=2pt];
\end{forest}
}
+
{\fontsize{4}{1}\selectfont
\begin{forest}
 for tree={grow=90, l=1mm}
[[][][]]
 \path[fill=black]  (!.parent anchor) circle[radius=2pt];
 \path[fill=black] (!1.child anchor) circle[radius=2pt]
                   (!2.child anchor) circle[radius=2pt]
                   (!3.child anchor) circle[radius=2pt];
\end{forest}
}
\\
F_{4}^{\tP}&=-
{\fontsize{4}{1}\selectfont
\begin{forest}
 for tree={grow=90, l=1mm}
 [[[[[]]]]]
 \path[fill=black]  (!.parent anchor) circle[radius=2pt];
 \path[fill=black] (!1.child anchor) circle[radius=2pt]
                  (!11.child anchor) circle[radius=2pt]
                  (!111.child anchor) circle[radius=2pt]
                  (!1111.child anchor)  circle[radius=2pt];
\end{forest}
}
+
{\fontsize{4}{1}\selectfont
\begin{forest}
 for tree={grow=90, l=1mm}
 [[[[][]]]]
 \path[fill=black]  (!.parent anchor) circle[radius=2pt];
 \path[fill=black] (!1.child anchor) circle[radius=2pt]
                  (!11.child anchor) circle[radius=2pt]
                  (!111.child anchor) circle[radius=2pt]
                  (!112.child anchor)  circle[radius=2pt];
\end{forest}
}
+
{\fontsize{4}{1}\selectfont
\begin{forest}
 for tree={grow=90, l=1mm}
[[[[]][]]]
 \path[fill=black]  (!.parent anchor) circle[radius=2pt];
 \path[fill=black] (!1.child anchor) circle[radius=2pt]
                  (!11.child anchor) circle[radius=2pt]
                  (!111.child anchor) circle[radius=2pt]
                  (!12.child anchor) circle[radius=2pt];
\end{forest}
}
+
{\fontsize{4}{1}\selectfont
\begin{forest}
 for tree={grow=90, l=1mm}
[[[[]]][]]
 \path[fill=black]  (!.parent anchor) circle[radius=2pt];
 \path[fill=black] (!1.child anchor) circle[radius=2pt]
                  (!11.child anchor) circle[radius=2pt]
                  (!111.child anchor) circle[radius=2pt]
                  (!2.child anchor) circle[radius=2pt];
\end{forest}
}
-
{\fontsize{4}{1}\selectfont
\begin{forest}
 for tree={grow=90, l=1mm}
 [ [[][]] []]
 \path[fill=black]  (!.parent anchor) circle[radius=2pt];
 \path[fill=black] (!1.child anchor) circle[radius=2pt]
                  (!11.child anchor) circle[radius=2pt]
                  (!12.child anchor) circle[radius=2pt]
                  (!2.child anchor)  circle[radius=2pt];
\end{forest}
}
-
{\fontsize{4}{1}\selectfont
\begin{forest}
 for tree={grow=90, l=1mm}
 [[[]] [] [] ]
 \path[fill=black]  (!.parent anchor) circle[radius=2pt];
 \path[fill=black] (!1.child anchor) circle[radius=2pt]
                  (!11.child anchor) circle[radius=2pt]
                  (!2.child anchor) circle[radius=2pt]
                  (!3.child anchor)  circle[radius=2pt];
\end{forest}
}
-
{\fontsize{4}{1}\selectfont
\begin{forest}
 for tree={grow=90, l=1mm}
 [ [[][][]] ]
 \path[fill=black]  (!.parent anchor) circle[radius=2pt];
 \path[fill=black] (!1.child anchor) circle[radius=2pt]
                  (!11.child anchor) circle[radius=2pt]
                  (!12.child anchor) circle[radius=2pt]
                  (!13.child anchor)  circle[radius=2pt];
\end{forest}
}
+
{\fontsize{4}{1}\selectfont
\begin{forest}
 for tree={grow=90, l=1mm}
 [ [][][] []]
 \path[fill=black]  (!.parent anchor) circle[radius=2pt];
 \path[fill=black] (!1.child anchor) circle[radius=2pt]
                  (!2.child anchor) circle[radius=2pt]
                  (!3.child anchor) circle[radius=2pt]
                  (!4.child anchor)  circle[radius=2pt];
\end{forest}
}
\\
&\ \ \vdots
\end{aligned}
\end{equation}
Together with \eqref{Tree_graphs}, it is easy to see that these three polynomials only differ in terms of the tree coefficients. 
\section{Derivation of the generalized mode-coupling equation}
\label{app:MCT_derivation}
We first define a two-dimensional observable vector $A=[\delta_{\rho_q},j^L_{q}]^T$. The Mori-type projection operator \eqref{Mori_P_MCT} can be symbolically rewritten as:
\begin{align*}
    \P B=(A,B)(A,A)^{-1}A,
\qquad (A,A)^{-1}=\begin{bmatrix}
\frac{1}{NS(|q|)} & 0\\
0 & \frac{m}{Nk_BT}
\end{bmatrix},
\end{align*}
where $(A,B)=\langle A^*B^T\rangle$ and $B=[u,v]^T$. With this projection operator, applying operator identity \eqref{op_id_1_nonc} to $A$, it is easy to get that the LHS of \eqref{op_id_1_nonc} equals to
\begin{align}\label{MCT_LHS}
    \frac{d}{dt}\P e^{it\L}\P A=
    \frac{d}{dt}
    \begin{bmatrix}
\langle\delta_{\rho_{-q}},\delta_{\rho_q}(t)\rangle & \langle\delta_{\rho_{-q}},j_{q}^L(t)\rangle \\
\langle j^L_{-q},\delta_{\rho_q}(t)\rangle & \langle j^L_{-q},\delta_{\rho_q}(t)\rangle
\end{bmatrix}
    (A,A)^{-1}A
=\frac{d}{dt}C(t)(A,A)^{-1}A.
\end{align}
As we will see, the observable vector $(A,A)^{-1}A$ appears in all other terms in the CMZE, hence will be canceled out in the follow-up derivations. In the end, Eqn \eqref{MCT_LHS} leaves us the time derivative of the equilibrium time-autocorrelation matrix $C(t)$. On the other hand, since $i\L$ is a skew-Hermitian operator with respect to the Hilbert space inner product $\langle,\rangle$, we have $\langle i\L a,a\rangle=0$ for any scalar observable function $a$. Combining this fact and the statistical thermodynamics equality: $\langle j^L_{-q},\delta\dot{\rho_{q}}\rangle=-i|q|\frac{Nk_BT}{m}$, we can get that the first term on the RHS of \eqref{op_id_1_nonc} is explicitly given by 
\begin{align}\label{MCT_RHS1}
    \frac{d}{dt}i\P e^{it\L}\P\L\P A=i\Omega C(t)(A,A)^{-1}A,
\end{align}
where $C(t)$ is defined as in \eqref{MCT_LHS} and the matrix $i\Omega$ is given by 
\begin{align}\label{matrix_iomega}
    i\Omega=
    \begin{bmatrix}
0 & -i|q|\\
-i|q|\frac{k_BT}{mS(|q|)} & 0
\end{bmatrix}
.
\end{align}
Now we calculate the memory integral. After the change of notation $\L\rightarrow i\L$, the zeroth order term in the combinatorial expansion \eqref{op_id_1_nonc} now becomes:
\begin{align}\label{MCT_0th_order}
    \int_0^t\P e^{i(t-s)\L}\F_0 Ads,
\end{align}
where $\F_0 A$ can be explicitly calculated as: 
\begin{equation}
\begin{aligned}
\F_0 A=i\P\L i\Q\L A&=
\P(i\L)^2\P A-(i\P\L\P)^2A\\
&=
\begin{bmatrix}
-|q|^2\frac{k_BT}{mS(|q|)} & 0 \\
0 &  -\frac{m}{Nk_BT}\left\langle
\frac{d}{dt}j_{-q}^L,\frac{d}{dt}j_{q}^L\right\rangle
\end{bmatrix}A
-
\begin{bmatrix}
-q^2\frac{k_BT}{mS(|q|)}& 0\\
0 & -q^2\frac{k_BT}{mS(|q|)}
\end{bmatrix}
A\\
&=
\begin{bmatrix}
\omega_0^1(q)& 0\\
0 & \omega_0^2(q)
\end{bmatrix}A
=\Omega_0A.
\end{aligned}
\end{equation}
Here we have set
\begin{equation}\label{w_0^2(q)}
\begin{aligned}
\omega_0^1(q)&=0\\
\omega_0^2(q)&=-\frac{m}{Nk_BT}\left\langle\frac{d}{dt}j_{-q}^L,\frac{d}{dt}j_{q}^L\right\rangle+\frac{|q|^2k_BT}{mS(|q|)}.
\end{aligned}
\end{equation}
As a result, \eqref{MCT_0th_order} simplifies to
\begin{align}\label{MCT_0th_final}
    \int_0^t\P e^{i(t-s)\L}\F_0 Ads=\int_0^t\Omega_0C(t-s)ds(A,A)^{-1}A.
\end{align}
The first-order term in the combinatorial expansion \eqref{op_id_1_nonc} can be simplified as :
\begin{align}\label{MCT_1th_order}
    \int_0^t\P e^{i(t-s)\L}\F_1(\P e^{is\L}\P)^k Ads=\int_0^t\P e^{i(t-s)\L}C(s)(A,A)^{-1}\F_1Ads,\quad k=0,1.
\end{align}
According to the explicit expression \eqref{first_few_f_n}, we know $\F_1$ contains four terms. In order to proceed, we also need to find a subspace $V\subset H$ such that $i\P\L\P|_V$ is invertible in this subspace. This is easy for the CMZE with Mori-type projection operator \eqref{Mori_P_MCT}. To this end, we only need to choose $V=\text{Span}\{\delta_{\rho_{q}},j^L_q\}$. According to the definition \eqref{Mori_P_MCT}, we have $i\P\L\P A=i\P\L A=i\Omega A$. Hence the subspace $V\subset H$ is an invariant subspace of operator $i\P\L\P$, in which $i\P\L\P$ admits a matrix representation $i\P\L\P A=i\Omega A$. Since $i\Omega$ defined in \eqref{matrix_iomega} is invertible for non-zero $q$. $i\P\L\P$ is invertible in $V$ with inverse $-i\Omega^{-1}$ ($q\neq 0$). Having confirmed the invertibility of $i\P\L\P$ in $V$, we introduce the following auxiliary matrices to calculate $\F_1 A$:
\begin{align}\label{auxi_matrices}
    D_0=(A,A), 
    \quad 
    D_1=(A,i\L A)(A,A)^{-1},
    \quad 
    D_2=(A,(i\L)^2A)(A,A)^{-1},
    \quad 
    \cdots.
\end{align}
Note that $D_0^{-1}=(A,A)^{-1}$ and $i\Omega=D_1$. With these auxiliary matrices, the four terms in $\F_1 A$ can be calculated individually as:
\begin{equation}\label{matrixF2}
\begin{aligned}
(i\P\L\P)^2 A&=D_1^2 A=\Omega_1^1A
\\
-(i\P\L\P)[\P(i\L)^2\P](i\P\L\P|_V)^{-1} A&=-
D_1^{-1}D_2D_1 A=\Omega_1^2A
\\
-\P(i\L)^2\P A&=-D_2A=\Omega_1^3A\\
\P(i\L)^3\P (i\P\L\P|_V)^{-1}A&=D_1^{-1}D_3 A
=\Omega_1^4 A
\end{aligned}
\end{equation}
where $\Omega_1^i=\Omega_1^i(q)$, $1\leq i\leq 4$ are diagonal matrices. By direct calculation, we obtain
\begin{align}
\Omega_1(q)=\sum_{i=1}^4\Omega_{1}^i(q)=
\begin{bmatrix}
\omega_1^1(q) & 0\\
0 & \omega_1^2(q)
\end{bmatrix}
=0.
\end{align}
Hence the first-order term \eqref{MCT_1th_order} vanishes since
\begin{align}
   \int_0^t\P e^{i(t-s)\L}\F_1(\P e^{is\L}\P)^k Ads
   =\int_0^t[C(s)D^{-1}_0(q)]^{k}\Omega_1(q)C(t-s)ds(A,A)^{-1}A=0,\quad k=0,1.
\end{align}
Similarly, using \eqref{first_few_f_n}, we can calculate eight terms in $\F_2 A$ explicitly as:
\begin{equation}
\begin{aligned}
(i\P\L\P)[\P(i\L)^2\P](i\P\L\P|_V)^{-1}(\P(i\L)^2\P)(i\P\L\P|_V)^{-2} A&=D_1^{-2}D_2D_1^{-1}D_2D_1 A=\Omega_2^1A
\\
-[\P (i\L)^3\P](i\P\L\P|_V)^{-1}[\P(i\L)^2\P](i\P\L\P|_V)^{-2} A&=-
D_1^{-2}D_2D_1^{-1}D_3 A=\Omega_2^2A
\\
-(i\P\L\P)^2 A&=-D_1^{2}A=\Omega_2^3A\\
(i\P\L\P)[\P(i\L)^2\P](i\P\L\P|_V)^{-1} A&=D_1^{-1}D^2D_1A=\Omega_2^4A\\
-(i\P\L\P)[\P(i\L)^3\P](i\P\L\P|_V)^{-2} A&=-D_1^{-2}D_3D_1A=\Omega_2^5A\\
\P(i\L)^2\P A&=D_2 A=\Omega_2^6 A\\
-[\P(i\L)^3\P] (i\P\L\P|_V)^{-1}A&=-D_1^{-1}D_3 A
=\Omega_2^7 A\\
[\P(i\L)^4\P] (i\P\L\P|_V)^{-2}A&=-D_1^{-2}D_4 A
=\Omega_2^8 A,
\end{aligned}
\end{equation}
where $\Omega_2^i=\Omega_2^i(q)$, $1\leq i\leq 8$ are diagonal matrices. By direct calculation, we obtain
\begin{align}
\Omega_2(q)=\sum_{i=1}^8\Omega_{2}^i(q)=
\begin{bmatrix}
\omega_2^1(q) & 0\\
0 & \omega_2^2(q)
\end{bmatrix}
\end{align}
where 
\begin{equation}\label{w_2^2(q)}
\begin{aligned}
    \omega_2^1(q)&=0\\
    \omega_2^2(q)&=-\frac{m^2S(|q|)}{|q|^2(k_BT)^2N}\left[
    \left\langle\frac{d^2}{dt^2}j_{-q}^L,\frac{d^2}{dt^2}j_{q}^L\right\rangle 
    -\frac{m}{Nk_BT}\left\langle
    \frac{d}{dt}j_{-q}^L,\frac{d}{dt}j_{q}^L\right\rangle^2
    \right].
\end{aligned}
\end{equation}
Hence the second-order contribution to the combinatorial expansion is given by:
\begin{align}\label{MCT_1th_final}
   \int_0^t\P e^{i(t-s)\L}\F_2(\P e^{is\L}\P)^k Ads
   =\int_0^t[C(s)D^{-1}_0(q)]^{k}\Omega_2(q)C(t-s)ds(A,A)^{-1}A=0,\quad k=0,1,2.
\end{align}
Combining \eqref{MCT_LHS}, \eqref{MCT_RHS1}, \eqref{MCT_0th_final} and \eqref{MCT_1th_final}, we get the second-order approximation to the CMZE \eqref{op_id_1_nonc}, which is a mode-coupling equation for the correlation matrix $C(t)$:
\begin{align}\label{2nd_MCT_C(t)}
    \frac{d}{dt}C(t)&\approx i\Omega(q) C(t)+\Omega_0(q)\int_0^tC(t-s)ds+\sum_{k=1}^2\frac{(-1)^{2-k}}{(2-k)!k!}\int_0^t[C(s)D_0^{-1}]^k\Omega_2(q)C(t-s)ds.
\end{align}
Note that the {\em only} approximation we have used in the whole derivation is the truncation of higher-order terms in the combinatorial expansion. Concentrating on the lower left corner of the above matrix equation, we get the following second-order generalized mode-coupling equation for the density fluctuation: 
\begin{equation}\label{2nd_MCT_F(q,t)}
\begin{aligned}
\frac{d^2}{dt^2}F(q,t)&+\frac{|q|^2k_BT}{mS(|q|)}F(q,t)-(\omega^2_0(q)+\frac{1}{2}\omega_2^2(q))\int_0^tds\frac{d}{dt}F(q,t-s)ds
\\
&-\frac{m\omega_2^2(q)}{|q|^2k_BT}\int_0^tds\frac{d^2}{ds^2}F(q,s)\frac{d}{dt}F(q,t-s)\\
&-\frac{1}{2}\frac{m\omega_2^2(q)}{|q|^2S(|q|)k_BT}\int_0^tds
\left[\frac{d}{ds}F(q,s)\right]^2
\frac{d}{dt}F(q,t-s)\\
&+\frac{1}{2}\frac{m^2N^2\omega_2^2(q)}{|q|^4k^2_BT^2}\int_0^tds
\left[\frac{d^2}{ds^2}F(q,s)\right]^2
\frac{d}{dt}F(q,t-s)=0.
\end{aligned}
\end{equation}
To solve \eqref{2nd_MCT_F(q,t)} numerically, we only need to know $\omega^2_0(q)$ and $\omega_2^2(q)$, where the definitions are given in \eqref{w_0^2(q)} and \eqref{w_2^2(q)}. It is easy to see these two coefficients are functions of the static moment in the statistical equilibrium. When retaining more terms in the combinatorial expansion, we get higher-order mode-coupling equations for the density fluctuation with low-order terms got {\em corrected}. For instance, the additional term $\frac{1}{2}\omega_2^2(q)$ added to the zeroth order contribution $\omega_0^2(q)$ in \eqref{2nd_MCT_F(q,t)} can be viewed as a higher-order correction to this linear term. Naturally, high-order approximations are expected to be more accurate while it requires knowing static moments of higher-order derivatives of $\delta_{\rho_q}$ and $j_{q}^L$ contained in auxiliary matrices $D_4=(A,\L^4 A)(A,A)^{-1}$, $D_5=(A,\L^5 A)(A,A)^{-1},\cdots$. As far as we are concerned, the generalized mode-coupling equation we obtained using the CMZE is different from any existing mode-coupling equations in the literature. In particular, it is constructive to compare it with the ones introduced by Janssen and Reichman  \cite{janssen2015microscopic}. In the latter, a system of EOMs that resembles the Martin-Schwinger hierarchy is used as a generalized mode-coupling equation to predict the glass transition. Correspondingly, it is required to know {\em multi-point} static density correlation functions as the input to solve the hierarchical equation.  
\section{Derivation of the CMZE for Hubbard model}
\label{app:hubbard}
We start with the time-independent MZE:
\begin{align}\label{eqn:time_ind_PMZE_OP_hubbard_app}
\frac{d}{dt}\P e^{it\L}\P=i\P e^{it\L}\P\L\P
+\int_0^t\P e^{i(t-s)\L}i\P\L e^{is\Q\L}i\Q\L\P ds.
\end{align} 
Using this equation, we will first derive the EOM for the full retarded Green's function $G^R_{ij}(t)$. One-dimensional EOM \eqref{Green_f_CMZE} can be obtained from the result rather easily. Define $N$-dimensional Mori-type projection operator:
\begin{align}\label{Mori_P_hubbard}
    \P(\cdot)=\sum_{j}|c_{j\sigma})(c_{j\sigma}|,
\end{align}
where the summation $\sum_{j}$ runs over all electrons in the lattice. It is easy to verify that this \eqref{Mori_P_hubbard} is finite-rank and $\text{Ran}(\P)=\text{Span}\{|c_{j\sigma})\}_{j}$. Substituting this projection operator into MZE \eqref{eqn:time_ind_PMZE_OP_hubbard_app} and applying the equation to observable $|c_{j\sigma})$, we obtain the EOM for the retarded Green's function:
\begin{align}\label{MZE_nd_matrix}
    i\frac{d}{dt} G^{R}(t)=\Omega G^{R}(t)+i\int_0^t K^{R}(s) G^{R}(t-s)ds,
\end{align}
where 
\begin{equation}\label{MZE_nd_hubbard_nd_app}
    \begin{aligned}
    G^{R}_{ij}(t)&=(G^{R})_{ij}(t)\\
   \Omega^{R}_{ij}&=(\Omega^{R})_{ij}=(c_{j\sigma}|i\L c_{i\sigma})\\
   K^{R}_{ij}(s)&=(K^{R})_{ij}(s)
   =\left(c_{j\sigma}\bigg|i\L e^{is\Q\L}i\Q\L c_{i\sigma}\right).
    \end{aligned}
\end{equation}
In the case of the tight-binding Hubbard model, the frequency matrix $\Omega_{ij}$ is explicitly given by 
\begin{align}\label{omega_hubbard_MZ_app}
    \Omega_{ij}&=(H_0)_{_{ij}}+U\langle n_{i\bar\sigma}\rangle\delta_{ij}=[(\epsilon_0-\mu)\delta_{ij}+t(\delta_{i-1j}+\delta_{i+1j})]+U\langle n_{i\bar\sigma}\rangle\delta_{ij},
\end{align}
To get the CMZE for the Hubbard model, we introduce auxiliary matrices as we did in \eqref{auxi_matrices}:
\begin{align}\label{auxi_matrices_hubbard}
    D_1=(A,i\L A),
    \quad 
    D_2=(A,(i\L)^2A),
    \quad 
    D_3=(A,(i\L)^2A),
    \quad 
    \cdots.
\end{align}
Here $A=[|c_{0\sigma}),|c_{0\sigma}),\cdots,|c_{N\sigma})]^T$, $A^*=[(c_{0\sigma}|,(c_{0\sigma}|,\cdots,(c_{N\sigma}|]^T$ and $(A,B)=(A^*|B)$. We further note that with these definitions, $D_0=(A,A)=\mathbb{I}$. As we have seen in \ref{app:MCT_derivation}, $D_1,D_2,D_3,\cdots$ are just the matrix representation of operators $\P i\L\P A$, $\P (i\L)^2\P A$, $\P(i\L)^3\P A,\cdots$ in the closed subspace $\text{Ran}(\P)$. By direct calculation, we find that $D_1,D_2,D_3$ are Toeplitz matrices of the form:
\begin{equation}\label{auxi_matrix_hubbard}
\begin{aligned}
D_1=
\begin{bmatrix}
a_0^1& a_1^1&\cdots\\
a_{-1}^1& a_0^1&\cdots\\
\vdots &\vdots
\end{bmatrix}
,\quad
D_2=
\begin{bmatrix}
a_0^2& a_1^2& a_2^2&\cdots\\
a_{-1}^2& a_0^2& a_1^2&\cdots\\
a_{-2}^2& a_{-1}^2& a_0^2&\cdots\\
\vdots &\vdots &\vdots
\end{bmatrix}
,\quad 
D_3=
\begin{bmatrix}
a_0^3& a_1^3& a_2^3& a_3^3&\cdots\\
a_{-1}^3& a_0^3& a_1^3& a_2^3&\cdots\\
a_{-2}^3& a_{-1}^3& a_0^3& a_1^3&\cdots\\
a_{-3}^3& a_{-2}^3& a_{-1}^3&a_0^3&\cdots\\
\vdots &\vdots&\vdots&\vdots,
\end{bmatrix}
\end{aligned}
\end{equation}
where 
\begin{equation}\label{a_0^1}
\begin{aligned}
a_0^1&=-i(\epsilon_0-\mu)-iU\langle n_{i\bar\sigma}\rangle\\
a_1^1&=a_{-1}^1=-it\\
\end{aligned}
\end{equation}
 
\begin{equation}\label{a_0^2}
\begin{aligned}
a_0^2&=-[(\epsilon_0-\mu)^2+2(\epsilon_0-\mu)\langle n_{i\bar\sigma}\rangle+2t^2+U^2\langle n_{i\bar\sigma}\rangle]\\
a_1^2&=-[2t(\epsilon_0-\mu)+Ut(\langle n_{i\bar\sigma}\rangle+
\langle n_{i+1\bar\sigma}\rangle)]\\
a_{-1}^2&=-[2t(\epsilon_0-\mu)+Ut(\langle n_{i\bar\sigma}\rangle+
\langle n_{i-1\bar\sigma}\rangle)]\\
a_2^2&=a_{-2}^2=-t^2\\
\end{aligned}
\end{equation}

\begin{equation}\label{a_0^3}
\begin{aligned}
a_0^3&=i[(\epsilon_0-\mu)^3+3U(\epsilon_0-\mu)^2\langle n_{i\bar\sigma}\rangle+6t^2(\epsilon_0-\mu)+3U^2(\epsilon_0-\mu)\langle n_{i\bar\sigma}\rangle\\
&\ \ \ 
+4t^2U\langle n_{i\bar\sigma}\rangle+t^2U\langle n_{i-1\bar\sigma}\rangle
+t^2U\langle n_{i+1\bar\sigma}\rangle
+U^3\langle n_{i\bar\sigma}\rangle]\\
a_{1}^3&=
i[3t(\epsilon_0-\mu)^2+3tU(\epsilon_0-\mu)(\langle n_{i\bar\sigma}\rangle+\langle n_{i+1\bar\sigma}\rangle)\\
&\ \ \ 
+3t^3+tU^2(\langle n_{i\bar\sigma}\rangle+\langle n_{i+1\bar\sigma}\rangle+\langle n_{i\bar\sigma}n_{i+1\bar\sigma}\rangle)
]
\\
a_{-1}^3&=i[3t(\epsilon_0-\mu)^2+3tU(\epsilon_0-\mu)(\langle n_{i\bar\sigma}\rangle+\langle n_{i-1\bar\sigma}\rangle)\\
&\ \ \ 
+3t^3+tU^2(\langle n_{i\bar\sigma}\rangle+\langle n_{i-1\bar\sigma}\rangle+\langle n_{i\bar\sigma}n_{i-1\bar\sigma}\rangle)
]\\
a_{2}^3&=i[3t^2(\epsilon_0-\mu)]+t^2U(\langle n_{i\bar\sigma}\rangle+\langle n_{i+1\bar\sigma}\rangle+\langle n_{i+2\bar\sigma}\rangle)\\
a_{-2}^3&=
i[3t^2(\epsilon_0-\mu)]+t^2U(\langle n_{i\bar\sigma}\rangle+\langle n_{i-1\bar\sigma}\rangle+\langle n_{i-2\bar\sigma}\rangle)
\\
a_3^3&=a_{-3}^3=it^3,\\
\end{aligned}
\end{equation}
and all other descending diagonals are 0. As a result, The first-order expansion term \eqref{MCT_0th_order} simplifies to
\begin{align}\label{MCT_0th_final_hubbard}
    \int_0^t\P e^{i(t-s)\L}\F_0 Ads=\int_0^t\Omega_0iG^{R}(t-s)ds A,
\end{align}
where $\Omega_0=D_2-D_1^2$. On the other hand, Gershgorin circle theorem implies that $D_1$ is invertible when it is strictly diagonally dominant. Since $D_1$ is a Toeplitz matrix, this happens when the Hubbard model is strongly correlated, i.e. $U\gg t$. We assume this is the case, then 
\begin{align}\label{MCT_1th_order_hubbard}
    \int_0^t\P e^{i(t-s)\L}\F_1(\P e^{is\L}\P)^k Ads=\int_0^t[iG^R(s)]^k\Omega_1iG^R(t-s)dsA,\quad k=0,1,
\end{align}
where $\Omega_1=\sum_{i=1}^4\Omega_1^i$, and $\Omega_1^i$ is defined as in \eqref{matrixF2} with $D_1=\mathbb{I}$ and other $D_i$s given by \eqref{auxi_matrix_hubbard}. If we truncate the combinatorial expansion at this order, we obtain the first-order self-consistent EOM for the retarded Green's function of the Hubbard model:
\begin{align*}
i\frac{d}{dt} G^{R}(t)=\Omega G^{R}(t)+i\int_0^t(\Omega_0-\Omega_1) G^{R}(t-s)ds-\int_0^tG^R(s)\Omega_1G^R(t-s)ds.
\end{align*}
On the other hand, it is easy to check that with the one-dimensional projection operator $\P=|c_{i\sigma})(c_{i\sigma}|$, the statistical moments $\gamma_i|c_{i\sigma})=\P(i\L)^{i+1}\P|c_{i\sigma})=(c_{i\sigma}|(i\L)^{i+1}c_{i\sigma})$ is given by the diagonals of $D_n$s with $\gamma_{i-1}=a_0^{i}$. Since the first few $a_0^i$s are given explicitly in \eqref{a_0^1}-\eqref{a_0^3}, combining with \eqref{f_n_CCMZE}, we obtain \eqref{f_n_hubbard_1P}.
\bibliographystyle{plain}
\bibliography{combinatorial_QFT}
\end{document}